\documentclass[11pt]{article}

\usepackage{amsfonts,amsmath,amssymb, bbm, bm, xspace, mathrsfs}
\usepackage{multirow,graphicx, epsfig,subfigure,algorithm,algorithmic,rotating}
\usepackage{cite,url}
\usepackage{fullpage}
\usepackage[small,bf]{caption}
\usepackage[table]{xcolor}
\usepackage{tabularx}
\usepackage{psfrag}
\usepackage{subfigure}
\usepackage{verbatim}
\def \endprf{\hfill {\vrule height6pt width6pt depth0pt}\medskip}
\newenvironment{proof}{\noindent {\bf Proof} }{\endprf\par}

\numberwithin{equation}{section}

\newtheorem{theorem}{Theorem}
\newtheorem{proposition}[theorem]{Proposition}

\newtheorem{lemma}[theorem]{Lemma}
\newtheorem{conjecture}[theorem]{Conjecture}
\newtheorem{corollary}[theorem]{Corollary}
\newtheorem{definition}[theorem]{Definition}
\newtheorem{example}[theorem]{Example}

\newcommand{\mbf}{\boldsymbol}
\newcommand{\mc}{\mathcal}

\newcommand{\mcB}{\mc{B}}
\newcommand{\mcC}{\mc{C}}
\newcommand{\mcE}{\mc{E}}
\newcommand{\mcF}{\mc{F}}
\newcommand{\mcI}{\mc{I}}
\newcommand{\mcJ}{\mc{J}}
\newcommand{\mcK}{\mc{K}}
\newcommand{\mcL}{\mc{L}}
\newcommand{\mcR}{\mc{R}}
\newcommand{\mcS}{\mc{S}}

\newcommand{\mcQ}{\mc{Q}}
\newcommand{\mcU}{\mc{U}}

\newcommand{\bfa}{\mbf{a}}

\newcommand{\bfb}{\mbf{b}}
\newcommand{\bfB}{\mbf{B}}
\newcommand{\bfc}{\mbf{c}}
\newcommand{\bfD}{\mbf{D}}
\newcommand{\bfe}{\mbf{e}}
\newcommand{\bff}{\mbf{f}}
\newcommand{\bfF}{\mbf{F}}
\newcommand{\bfg}{\mbf{g}}
\newcommand{\bfh}{\mbf{h}}
\newcommand{\bfH}{\mbf{H}}
\newcommand{\bfI}{\mbf{I}}

\newcommand{\bfn}{\mbf{n}}
\newcommand{\bfp}{\mbf{p}}
\newcommand{\bfP}{\mbf{P}}

\newcommand{\bfr}{\mbf{r}}

\newcommand{\bfs}{\mbf{s}}
\newcommand{\bfS}{\mbf{S}}
\newcommand{\bft}{\mbf{t}}
\newcommand{\bfT}{\mbf{T}}
\newcommand{\bfu}{\mbf{u}}
\newcommand{\bfU}{\mbf{U}}
\newcommand{\bfv}{\mbf{v}}

\newcommand{\bfw}{\mbf{w}}
\newcommand{\bfW}{\mbf{W}}
\newcommand{\bfz}{\mbf{z}}
\newcommand{\bfZ}{\mbf{Z}}
\newcommand{\bfy}{\mbf{y}}
\newcommand{\bfY}{\mbf{Y}}
\newcommand{\bfx}{\mbf{x}}
\newcommand{\bfX}{\mbf{X}}
\newcommand{\bfPhi}{\mbf{\Phi}}

\newcommand{\bflambda}{\mbf{\lambda}}
\newcommand{\bfgamma}{\mbf{\gamma}}
\newcommand{\bfeta}{\mbf{\eta}}

\newcommand{\bfzero}{\mbf{0}}

\newcommand{\w}{\omega}
\newcommand{\kk}{k}
\newcommand{\iterk}{\delta} 
\newcommand{\K}{\mcK}
\newcommand{\mcBB}{\tilde{\mcB}}
\newcommand{\mcKK}{\tilde{\mcK}}

\newcommand{\bffhat}{\mbf{\hat{f}}}
\newcommand{\bffbar}{\mbf{\bar{f}}}

\let\emptyset\varnothing
\newcommand{\real}{\mathbb{R}}
\newcommand{\GF}{\mathbb{F}}
\newcommand{\perm}{\pi}
\newcommand{\error}{\mathrm{error}}
\newcommand{\new}{\mathrm{new}}
\newcommand{\cover}{\operatorname{cover}}

\newcommand{\codebook}{\mathcal{C}}

\newcommand{\blocklength}{N}
\newcommand{\checknumber}{M}
\newcommand{\checkdegree}{{d_c}}
\newcommand{\variabledegree}{d_v}
\newcommand{\Nev}{\mathcal{N}_v}
\newcommand{\Nec}{\mathcal{N}_c}

\newcommand{\PP}{\mathbb{P}} 
\newcommand{\PS}{\mathbb{S}}
\renewcommand{\SS}{\mathbb{S}'}

\newcommand{\tightcodepolytope}{\mathbb{V}}
\newcommand{\relaxedcodepolytope}{\mathbb{U}}
\newcommand{\cwtightcodepolytope}{\tightcodepolytope'}
\newcommand{\cwrelaxedcodepolytope}{\relaxedcodepolytope'}
\newcommand{\tightcodepolytopeN}{\tightcodepolytope^{\normal}} 
\newcommand{\relaxedcodepolytopeN}{\relaxedcodepolytope^{\normal}}

\newcommand{\normal}{\mathsf{N}}
\newcommand{\mcCN}{\mcC^{\normal}}
\newcommand{\mcEN}{\mcE^{\normal}}

\newcommand{\bfeN}{\bfe^{\normal}}
\newcommand{\bfuN}{\bfu^{\normal}}
\newcommand{\bfvN}{\bfv^{\normal}}

\newcommand{\argmin}{\operatorname{argmin}}
\newcommand{\argmax}{\operatorname{argmax}}
\newcommand{\conv}{\operatorname{conv}}
\newcommand{\vect}{\operatorname{vec}}
\newcommand{\diag}{\operatorname{diag}}
\newcommand{\st}{\operatorname{subject~to}}
\newcommand{\Larg}{\mcL}
\newcommand{\Proj}{\Pi}
\newcommand{\penpara}{\mu} 

\newcommand{\embed}{\mathsf{f}}
\newcommand{\embedvec}{\mathsf{F}_v}
\newcommand{\embedmat}{\mathsf{F}}
\newcommand{\embedbit}{\mathsf{b}}
\newcommand{\embedbitmat}{\mathsf{B}}
\newcommand{\relative}{\mathsf{R}}
\newcommand{\relativeinv}{\mathsf{R}^{-1}}
\newcommand{\relativecw}{\mathsf{R}'}
\newcommand{\relativecwinv}{\mathsf{R}'^{-1}}
\newcommand{\embedllr}{\mathsf{L}}
\newcommand{\embedllrvec}{\mathsf{L}_v}
\newcommand{\rotmat}{\bfD}

\newcommand{\cwembed}{\mathsf{f}'}
\newcommand{\cwembedvec}{\mathsf{F}_v'}
\newcommand{\cwembedmat}{\mathsf{F}'}
\newcommand{\cwembedllr}{\mathsf{L}'}
\newcommand{\cwembedllrvec}{\mathsf{L}_v'}

\newcommand{\cwrotmat}{\bfD'}

\newcommand{\LPDFlanagan}{\textbf{LP-FT}\xspace}
\newcommand{\LPDCW}{\textbf{LP-CT}\xspace}
\newcommand{\LPDFlanaganRelax}{\textbf{LP-FR}\xspace}
\newcommand{\LPDCWRelax}{\textbf{LP-CR}\xspace}

\newcommand{\nameF}{Flanagan's embedding\xspace}
\newcommand{\nameCW}{constant-weight embedding\xspace}
\newcommand{\nameTheCW}{the \nameCW}
\newcommand{\minimizer}{solution\xspace}
\begin{document}
\title{ADMM LP decoding of non-binary LDPC codes in $\GF_{2^m}$\thanks{This material was presented in part at the IEEE 2014 Int. Symp. on Inf. Theory (ISIT), Honolulu, HI, July, 2014.This work was supported by the National Science
Foundation (NSF) under Grants CCF-1217058 and by a Natural Sciences and Engineering Research Council of Canada (NSERC) Discovery Research Grant. This paper was submitted to \emph{IEEE Trans. Inf. Theory.}}} 

\author{Xishuo Liu
\thanks{X.~Liu is with the Dept.~of Electrical and Computer
  Engineering, University of Wisconsin, Madison, WI 53706
  (e-mail: xishuo.liu@wisc.edu).}  
, Stark C. Draper
\thanks{S.~C.~Draper is with the Dept.~of Electrical and Computer
  Engineering, University of Toronto, ON M5S 3G4, Canada (e-mail: stark.draper@utoronto.ca).
  }
}

\maketitle
\begin{abstract}
In this paper, we develop efficient decoders for non-binary low-density parity-check (LDPC) codes using the alternating direction method of multipliers (ADMM). We apply ADMM to two decoding problems. The first problem is linear programming (LP) decoding. In order to develop an efficient algorithm, we focus on non-binary codes in fields of characteristic two. This allows us to transform each constraint in $\GF_{2^m}$ to a set of constraints in $\GF_{2}$ that has a factor graph representation. Applying ADMM to the LP decoding problem results in two types of non-trivial sub-routines. The first type requires us to solve an unconstrained quadratic program. We solve this problem efficiently by leveraging new results obtained from studying the above factor graphs. The second type requires Euclidean projection onto polytopes that are studied in the literature, a projection that can be solved efficiently using off-the-shelf techniques, which scale linearly in the dimension of the vector to project. ADMM LP decoding scales linearly with block length, linearly with check degree, and quadratically with field size. The second problem we consider is a penalized LP decoding problem. This problem is obtained by incorporating a penalty term into the LP decoding objective. The purpose of the penalty term is to make non-integer solutions (pseudocodewords) more expensive and hence to improve decoding performance. The ADMM algorithm for the penalized LP problem requires Euclidean projection onto a polytope formed by embedding the constraints specified by the non-binary single parity-check code, which can be solved by applying the ADMM technique to the resulting quadratic program. Empirically, this decoder achieves a much reduced error rate than LP decoding at low signal-to-noise ratios.
\end{abstract} 

\section{Introduction}
\label{section.introduction}
Using optimization theory to solve channel decoding problems dates back at least to~\cite{breitbach1998soft}, where Breitbach \emph{et al.} form an integer programming problem and solve it using a branch-and-bound approach. Linear programming (LP) decoding of binary low-density parity-check (LDPC) codes is introduced by Feldman \textit{et al.} in~\cite{feldman2005using}. Feldman's LP decoding problem is defined by a linear objective function and a set of linear constraints. The linear objective function is constructed using the log-likelihood ratios and is the same objective as in maximum likelihood (ML) decoding. The set of constraints is obtained by first relaxing each parity-check constraint to a convex polytope called the ``parity polytope'' and then intersecting all constraints. Compared to classic belief propagation (BP) decoding, LP decoding is far more amenable to analysis. Many decoding guarantees (e.g.,~\cite{feldman2007lp, daskalakis2008probabilistic, arora2012message, halabi2011lp, bazzi2014linear}) have been developed for LP decoding in the literature. However, LP decoding suffers from high computational complexity when approached using generic solvers. Since~\cite{feldman2005using}, many works have sought to reduce the computational complexity of LP decoding of binary LDPC codes~\cite{vontobel2006towards, taghavi2008adaptive, burshtein2009iterative, taghavi2011efficient, barman2013decomposition, zhang2013efficient, zhang2013large}. Of particular relevance is~\cite{barman2013decomposition} where the authors use the alternating direction method of multipliers (ADMM) to decompose the LP decoding problem into simpler sub-problems and develop an efficient decoder that has complexity comparable to the sum-product BP decoder. In a more recent work~\cite{liu2014the}, Liu \emph{et al.} use ADMM to try to solve a penalized LP decoding objective and term the problem \emph{ADMM penalized decoding}.\footnote{We write ``try to'' because there is no guarantee that ADMM solves the non-convex program, in contrast to LP decoding, for which ADMM is guaranteed to produce the global optimum of the LP.} The penalized LP objective is constructed by adding a non-convex term to the LP decoding objective to penalize pseudocodewords, to which LP decoding can fail. Empirically, ADMM penalized decoding outperforms ADMM LP decoding in terms of both word-error-rate (WER) and computational complexity. There are two important lessons from~\cite{barman2013decomposition} and~\cite{liu2014the}. First, ADMM based decoding algorithms can be efficient algorithms that are promising for practical implementation. Second, these algorithms can achieve lower WERs than BP does at both low and high signal-to-noise ratios (SNRs). In particular, simulation results in~\cite{liu2014the} show that ADMM penalized decoding can outperform BP decoding at low SNRs and does not suffer from an error-floor at WERs above $10^{-10}$ for the $[2640,1320]$ ``Margulis'' LDPC code.

LP decoding of non-binary LDPC codes is introduced by Flanagan \textit{et al.} in~\cite{flanagan2009linearprogramming}. Results in~\cite{barman2013decomposition} and~\cite{liu2014the} motivate us to extend ADMM decoding to non-binary codes. However, this extension is non-trivial. One key component of the ADMM decoder in~\cite{barman2013decomposition} is a sub-routine that projects a length-$d$ vector onto the parity polytope -- the convex hull of all length-$d$ binary vectors of even parity. The projection algorithm developed in \cite{barman2013decomposition} has a complexity of $O(d\log d)$. More recently,~\cite{zhang2013efficient} and~\cite{zhang2013large} propose more efficient projection algorithms, e.g., the algorithm proposed in~\cite{zhang2013large} has linear complexity in dimension $d$. Unfortunately, these techniques cannot be directly applied to LP decoding of non-binary codes as the polytope induced has two major differences. First, unlike the binary case where $0,1$ can be directly embedded into the real space, elements of $\GF_q$ need to be mapped to binary vectors of dimension at least $q-1$~\cite{flanagan2009linearprogramming}. Second, permutations of a codeword of non-binary single parity-check (SPC) codes are not necessarily codewords whereas for binary SPC codes all permutations of a codeword are codewords.\footnote{Permutation symmetry is key to the algorithms in~\cite{barman2013decomposition} and~\cite{liu2014the}.}

We briefly describe our approach to developing an efficient LP decoder using ADMM. We focus on non-binary codes in fields of characteristic two and represent each element in $\GF_{2^m}$ by a binary vector of dimension $m$, where each entry corresponds to a coefficient of the polynomial representing that element. This allows us to build a connection between an element's embedding (i.e., those defined in~\cite{flanagan2009linearprogramming}) and its binary vector representation. We then obtain a factor graph characterization of valid embeddings. With this characterization, the polytopes considered in~\cite{flanagan2009linearprogramming} can be further relaxed to intersections of simplexes and parity polytopes. We develop two sets of important results pertinent to this relaxation. First, we show several properties of the aforementioned factor graph of embeddings. As we will see in the main text, these properties are crucial for developing an efficient ADMM update algorithm. Further, these properties are useful in determining the computational complexity of the decoder. Second, we regain the permutation properties by introducing a rotation step similar to the one used in the study of BP decoding of non-binary codes in~\cite{declercq2007decoding}. This rotation step is generic and can be applied to any finite field. It rotates the polytope so that it has a ``block permutation'' property. This result is not only an interesting theoretical characterization of the polytope but also plays an important role in the algorithm.

Several other works also address the complexity of LP decoding of non-binary LDPC codes. For example, in~\cite{goldin2013iterative} and~\cite{punekar2012low}, algorithms are presented that perform coordinate-ascent on an approximation of the dual of the original LP decoding problem. Using binary vectors to represent elements of $\GF_{2^m}$ for LP decoding is considered in~\cite{honda2012fast} and~\cite{touri2008some}. In particular in~\cite{honda2012fast}, Honda and Yamamoto introduced a subtle connection between a constant weight embedding of fields of characteristic two and the binary vector representation of fields of characteristic two in~\cite{honda2012fast}. Although the authors of~\cite{honda2012fast} did not provide an efficient algorithm to leverage this connection, we think that this work is important and inspiring, and build off it herein.

We list below our contributions in this paper.
\begin{itemize}
\item We extend the ideas of~\cite{honda2012fast} and establish connections between Flanagan's embedding of finite fields~\cite{flanagan2009linearprogramming} and the binary vector representation of finite fields. We introduce a new factor graph representation of non-binary constraints based on embeddings (Section~\ref{section.combinatorics_SPC_GF2m}).
\item We study the geometry of the LP decoding constraints. We show that the polytope formed from embeddings of the non-binary single parity-check code can be rotated so that the vertices have a ``block permutation'' property. In addition, we compare two relaxations in terms of their tightness (Section~\ref{section.relaxation}). A conjecture that characterizes the polytope for $\GF_{2^2}$ is discussed in Appendix~\ref{appendix.conjecture}.
\item In the context of LP decoding, we conduct an extensive study of both Flanagan's embedding method and \nameTheCW from~\cite{honda2012fast}. The main result is that LP decoding using either embedding method achieves the same WER performance (Section~\ref{subsec.equivalence_embeddings}).
\item We formulate two ADMM decoding algorithms. The formulation in Section~\ref{subsec.genericADMM} applies to LP decoding and to penalized decoding (cf. Section~\ref{section.penalized}) of non-binary codes in arbitrary fields. However, in arbitrary fields we have not been successful in developing a computationally simple projection sub-routine. In Section~\ref{section.LPandADMM_noADMMProj} we limit our focus to LP decoding of codes in $\GF_{2^m}$ and develop an efficient algorithm.
\item We show through numerical results that the penalized decoder improves the WER performance significantly at low SNRs (Section~\ref{section.numerical}).
\end{itemize}

\section{Notation and definitions}
\label{section.notations}
In this section, we fix some notation and definitions that will be used through out the paper. 
\subsection{General font conventions}
We use $\real$ to denote the set of real numbers and $\GF_q$ to denote finite fields of size $q$. We denote by $\GF_{2^m}$ fields of characteristic two. We use bold capital letters to denote matrices and bold small letters to denote vectors. We use non-bold small letters to denote entries of matrices and vectors. For example, $\bfX$ is a matrix and $x_{ij}$ is the entry at the $i$-th row and the $j$-th column. $\bfx$ is a vector and its $i$-th entry is $x_i$. We use script capital letters to denote discrete sets and use $|\cdot|$ to denote the cardinality of a set. Further, we denote by $[n]$ the set $\{1\dots,n\}$. We use sans-serif font to represent mappings and functions, e.g. $\embed(\cdot)$ and $\embedmat(\cdot)$. We also explicitly define the following operations: $\conv(\cdot)$ is the operation of taking the convex hull of a set in $\real^n$. It is defined as
\small
\begin{equation*}
\conv(\mcS) := \left\{\sum_{i = 1}^{|\mcS|} \alpha_i \bfs_i \middle| \sum_{i = 1}^{|\mcS|}\alpha_i = 1 \text{ and } \alpha_i \geq 0 \text{ for all }i \right\}.
\end{equation*}
\normalsize
The operation $\vect(\cdot)$ vectorizes a matrix, i.e., it stacks all columns of a matrix to form a column vector. For example, let $\bfX$ be a $m \times n$ matrix; then $\vect(\bfX) = (x_{11},x_{21},\dots,x_{1m},x_{21},x_{22}\dots,x_{m(n-1)},x_{1n},\dots,x_{mn})^T$.

\subsection{Notation and definitions of convex geometries}
Throughout this paper we use roman blackboard bold symbols to denote convex sets, e.g., $\PP$.\footnote{Note that $\real$ (for real numbers) and $\GF$ (for finite fields) are two special cases where we use roman blackboard bold symbols.} In particular, we define the following polytopes that will be used extensively throughout this paper:
\begin{definition}
Let $\PP_d$ denote the \emph{parity polytope} of dimension $d$,
\small
\begin{equation}
\label{eq.def_ppd}
\PP_d := \conv(\{\bfe \in \{0,1\}^d | \|\bfe\|_1 \text{ is even}\}).
\end{equation}
\normalsize
Let $\PS_{d-1}$ denote the positive orthant of a unit $\ell_1$ ball of dimension $d-1$,
\small
\begin{equation}
\label{eq.def_ps}
\PS_{d-1} := \left\{(x_1,\dots,x_{d-1})\in \real^{d-1}\middle|
 \sum_{i = 1}^{d-1} x_i \leq 1, \text{ and }x_i \geq 0 \text{ for all }i\right\}.
\end{equation}
\normalsize
Let $\SS_d$ denote the \emph{standard $d$-simplex}, 
\begin{equation}
\label{eq.def_ss}
\SS_d := \left\{(x_0,\dots,x_{d-1})\in \real^{d} \big|
  \sum_{i = 0}^{{d-1}} x_i = 1, \text{ and }x_i \geq 0 \text{ for all }i\right\}.
\end{equation}
\end{definition}
It is easy to verify that $(x_1, \dots ,x_{d-1}) \in \PS_{d-1}$ implies $(x_0,\dots,x_{d-1})\in \SS_d$ where $x_0 = 1 - (\sum_{i = 1}^{d-1} x_i )$. Because of this relationship, we use $\PS_{d-1}$ with one dimension less than $\SS_d$. Further, we let the vector index start from $0$ for vectors in $\SS_d$ and $1$ for vectors in $\PS_{d-1}$. 

\subsection{Notation of non-binary linear codes}
We consider non-binary linear codes of length $\blocklength$ with $\checknumber$ checks, and use $\bfH$ to denote the parity-check matrix of a code where each entry $h_{ji} \in \GF_q$. We use $j$ to represent a check node index and $i$ to represent a variable node index. Then the set of codewords is $\{\bfc | \bfH\bfc = 0 \text{ in }\GF_q\}$. When represented by a factor graph, the set of variable and check nodes are denoted by $\mcI$ and $\mcJ$ respectively. Let $\Nev(i)$ and $\Nec(j)$ be the respective set of neighboring nodes of variable node $i$ and of check node $j$.

\subsection{Integer representation of finite fields}
In this paper we often consider finite fields of characteristic two. Since we will have more occasion to add in finite fields than to multiply, we represent each element of $\GF_{2^m}$ using an integer in $\{0,\dots,2^m-1\}$. Without loss of generality, an element $\alpha \in \GF_{2^m}$ can be represented by a polynomial $p(x) = \sum_{i = 1}^{m} b_i x^{i-1}$. In this paper we often use the integer representation of $\alpha$ which is given by $p(2) = \sum_{i = 1}^{m} b_i 2^{i-1}$.
\begin{example}
\label{example.finite_fields_representations}
We consider $\GF_{2^3}$ and let the primitive polynomial be $x^3+x+1$. Let $\xi $ be the primitive element of the field. We list different representations of elements in $\GF_{2^3}$ in Table~\ref{table.integer_rep}.
\begin{table}
\centering
\begin{tabular}{| c | c | c |c| }
  \hline                    
    & polynomial&  bits ($b_3b_2b_1$) & integer   \\
  \hline
  $\xi^0$ & $1$ & $001$ & $1$\\
  $\xi^1$ & $\xi$ & $010$ & $2$\\
  $\xi^2$ & $\xi^2$ & $100$ & $4$\\
  $\xi^3$ & $\xi+1$ & $011$ & $3$\\
  $\xi^4$ & $\xi^2 + \xi$ & $110$ & $6$\\
  $\xi^5$ & $\xi^2 + \xi + 1$ & $111$ &$7$ \\
  $\xi^6$ & $\xi^2 + 1$ & $101$ & $5$\\
  \hline
\end{tabular} 
\caption{Different representations of elements in $\GF_{2^3}$.}
\label{table.integer_rep}
\end{table}
Note that although we use the integer representations, multiplications are still performed using field operations. For example, using Table~\ref{table.integer_rep}, $4 \cdot 6 = 5$ in $\GF_{2^3}$.

In $\GF_{2^2}$ we let the primitive polynomial be $x^2+x+1$. Let $\xi $ be the primitive element of the field. We then obtain different representations of elements in $\GF_{2^2}$ in Table~\ref{table.integer_rep_gf4}.
\begin{table}
\centering
\begin{tabular}{| c | c | c |c| }
  \hline                    
    & polynomial&  bits ($b_2b_1$) & integer   \\
  \hline
  $\xi^0$ & $1$ & $01$ & $1$\\
  $\xi^1$ & $\xi$ & $10$ & $2$\\
  $\xi^2$ & $\xi + 1$ & $11$ & $3$\\
  \hline  
\end{tabular} 
\caption{Different representations of elements in $\GF_{2^2}$.}
\label{table.integer_rep_gf4}
\end{table}
\end{example}

\section{Embedding methods and representations of non-binary single parity-check codes}
\label{sec.embedding_and_SPC}
Embedding finite fields into reals is a crucial step in the LP decoding of non-binary codes. In this section, we first review two somewhat similar embedding methods. We refer these embeddings as ``\nameF'' and ``\nameTheCW''. 
In Section~\ref{section.combinatorics_SPC_GF2m}, we express the constraints of non-binary SPC codes using these embeddings and then show properties of the constraints. In Section~\ref{section.relaxation}, we characterize the properties of the relaxed constraints. In particular, we introduce a rotations step that normalizes the geometry under consideration such that we only need to study the geometry induced by the SPC code defined by the all-ones check.

\subsection{Embedding finite fields}
\label{sec.embedding_finite_fields}
One key component of LP decoding is in building a connection between the finite fields, in which the codes are defined, and the Euclidean space, in which optimization algorithms operates. In~\cite{feldman2005using}, the mapping $\GF_2 \mapsto \real$ is defined by $0 \rightarrow 0$ and $1 \rightarrow 1$. This straightforward transformation cannot be applied to the ML decoding problem of non-binary linear codes, where codes live in $\GF_q^\blocklength$. 
We focus on embedding methods in this section and defer our discussions of the LP decoding formulation to Section~\ref{section.LPandADMM}.

In~\cite{flanagan2009linearprogramming}, Flanagan \textit{et al.} introduce a mapping that embeds elements in $\GF_q$ into the Euclidean space of dimension $q-1$. 
\begin{definition} (\nameF~\cite{flanagan2009linearprogramming})
\label{def.flanagan_embedding}
Let $\embed: \GF_{q} \mapsto \{0,1\}^{q-1}$ be a mapping such that for $\alpha \in \GF_q$, 
\begin{equation*}
\embed(\alpha) := \bfx = (x_1,x_2,\dots,x_{q-1})
\end{equation*}
where
\begin{equation*}
x_{\delta} = \begin{cases}
1, & \text{if }\delta = \alpha,\\
0, & \text{if }\delta \neq \alpha.
\end{cases}
\end{equation*}
\end{definition}
We note that $\embed(\alpha)$ is a vector of length $q-1$ with at most one $1$. Using this mapping, $0\in \GF_{q}$ is mapped to the all-zeros vector of length $q-1$, which has Hamming weight $0$. On the other hand, all other elements of $\GF_{q}$ are mapped to binary vectors of Hamming weight $1$.

In a more recent work~\cite{honda2012fast}, the authors use a slightly different embedding: 
\begin{definition}(\nameCW~\cite{honda2012fast})
\label{def.constant_weight_embedding}
Let $\cwembed: \GF_{q} \mapsto \{0,1\}^{q}$ be a mapping such that for $\alpha \in \GF_q$,  
\begin{equation*}
\cwembed(\alpha) := \bfx = (x_0,x_1,\dots,x_{q-1})
\end{equation*}
where
\begin{equation*}
x_{\delta} = \begin{cases}
1, & \text{if }\delta = \alpha,\\
0, & \text{if }\delta \neq \alpha.
\end{cases}
\end{equation*}
\end{definition}
Now $\cwembed(\alpha)$ is a vector of length $q$. In addition, the Hamming weight of $\cwembed(\alpha)$ is $1$ for all values of $\alpha \in \GF_{q}$. For this reason, we name this embedding method \textit{\nameCW}.

Using \nameF in Definition~\ref{def.flanagan_embedding}, a vector $\bfc \in \GF_q^d$ can be mapped to a length $(q-1) d$ vector $\embedvec(\bfc) = (\embed(c_1)|\embed(c_2)|\cdots|\embed(c_n))^T$. We can also write the above vector in an equivalent matrix of dimension $(q-1)\times d$ as: $\embedmat(\bfc) = (\embed(c_1)^T|\embed(c_2)^T|\cdots|\embed(c_n)^T)$. In this representation, each column represents a coordinate of vector $\bfc$. We use the term \textit{embedded matrix} to refer to the matrix-form embedding of a codeword. Further, let $\embedmat(\mcC)$ and $\embedvec(\mcC)$ be images of set $\mcC$ under the two mappings respectively. That is, $\embedmat(\mcC) =  \left\{ \bfy \left\lvert  \bfy = \embedmat(\bfc) \text{ for } \bfc \in \mcC \right.\right\}$ and $\embedvec(\mcC) =  \left\{ \bfy \left\lvert  \bfy = \embedvec(\bfc) \text{ for } \bfc \in \mcC \right.\right\}$.
Similar definitions also apply to \nameTheCW in Definition~\ref{def.constant_weight_embedding}. Further, we use $\cwembedmat(\cdot)$ and $\cwembedvec(\cdot)$ to denote the respective operations of embedding a vector using matrices and vectors following \nameTheCW method. Following the convention of Definition~\ref{def.constant_weight_embedding}, we index rows of the matrix obtained by $\cwembedmat(\cdot)$ from zero.

The two embedding methods have many aspects in common. In particular, many definitions and lemmas for one embedding method can easily be extended to the other embedding. We show in Section~\ref{subsec.equivalence_embeddings} that both embeddings achieve the same WER for LP decoding. On the other hand, there are two important differences between the two embeddings that make each useful for particular purposes. First, \nameF saves exactly $1$ coordinate for each embedded vector which typically translates into lower complexity. In particular, we show in Section~\ref{section.numerical} that ADMM LP decoding using \nameTheCW is slower than when \nameF is used. Second, \nameTheCW is symmetric to all elements in the field while \nameF treats $0$ differently from other elements. This symmetry makes \nameTheCW useful when trying to solve non-convex programs, such as the penalized decoder described in Section~\ref{section.penalized}.

\subsection{Embedding of single parity-check codes in $\GF_{2^m}$}
\label{section.combinatorics_SPC_GF2m}
We now restrict our discussion to fields of characteristic two. We further focus on \nameF for conciseness. It is easy to extend all the definitions and lemmas to \nameTheCW. These extensions are presented in Appendix~\ref{appendix.def_for_cw_embedding}.

A $\checkdegree$-dimensional\footnote{We use $\checkdegree$ to denote the length of the code because it is related to the check degree of an LDPC code.} SPC code in $\GF_{2^m}$ is defined as
\begin{equation}
\label{eq.spc_code}
 \codebook = \left\{ \bfc \in \GF_{2^m}^\checkdegree  \middle| \sum_{j = 1}^\checkdegree c_j h_j = 0 \right\},
\end{equation}
where the addition is in $\GF_{2^m}$ and where $\bfh = (h_1,h_2,\dots,h_\checkdegree)\in \GF_{2^m}^\checkdegree$ is the check vector. We are interested in characterizing $\embedmat(\codebook)$ because it arises in the relaxation of LP decoding. In~\cite{honda2012fast} Honda and Yamamoto use a set of constraints derived from~\eqref{eq.spc_code} that tests whether an embedding is in $\embedmat(\codebook)$ to introduce their relaxation. However, the authors did not formally study these constraints (on SPC code embeddings). In this paper, we make the following contributions: First, we characterize $\embedmat(\codebook)$ and $\cwembedmat(\codebook)$, i.e., SPC code embeddings for both embedding methods. Second, we explicitly state a factor graph representation of the embedding of the SPC code and derive properties of the factor graph. Third, we discuss the redundancy of these constraints. Finally, we characterize $\embedmat(\codebook)$ for $\GF_{2^2}$ in Appendix~\ref{appendix.conjecture}.

To begin with, we first formally define the binary vector representation of $\GF_{2^m}$ that appeared in Example~\ref{example.finite_fields_representations}: 
\begin{definition}
For an element $\alpha \in \GF_{2^m}$, let $$\embedbit(\alpha) = (b_1,b_2,\dots,b_m)$$ be the binary vector representation of $\alpha$. In other words, the polynomial representation for $\alpha$ is $$p(x) = b_m x^{m-1} + b_{m-1}x^{m-2} +\cdots + b_1.$$ Then, for any vector $\bfc \in \GF_{2^m}^\checkdegree$, the ``bit matrix'' representation of $\bfc$ is a binary matrix of dimension $m\times \checkdegree$. Denote this representation as 
$$\embedbitmat(\bfc) = (\embedbit(c_1)^T|\embedbit(c_2)^T|\dots|\embedbit(c_\checkdegree)^T).$$
\end{definition}
\begin{example}
In $\GF_{2^2}^3$, $\embedbitmat((1,2,3)) = \begin{pmatrix}
1 & 0 & 1\\
0 & 1 & 1
\end{pmatrix}$.
\end{example}

Consider the set of embeddings defined in Definition~\ref{def.validembedding} when \nameF is used. We show in Lemma~\ref{lemma.valid_codeword} that Definition~\ref{def.validembedding} characterizes $\embedmat(\codebook)$.
\begin{definition}
\label{def.validembedding}
Under \nameF let $q = 2^{m}$ and denote by $\mcE$ the set of \textit{valid embedded matrices} defined by a length-$\checkdegree$ check $\bfh$ where any $(q-1)\times \checkdegree$ binary matrix $\bfF \in \mcE$ if and only if it satisfies the following conditions: 
\begin{enumerate}
\item[$(a)$] $f_{ij} \in \{0,1\}$, where $f_{ij}$ denotes the entry in the $i$-th row and $j$-th column of matrix $\bfF$.
\item[$(b)$] $\sum_{i = 1}^{q-1} f_{ij} \leq 1$ where $q = 2^m$.
\item[$(c)$] For any non-zero $h\in \GF_{2^m}$ and $k\in[m]$, let 
\begin{equation}
\label{nonbinary.eq.b_set_definition}
\mcB(k, h) := \{ \alpha | \embedbit(h \alpha)_k = 1, \alpha \in \GF_{2^m}\},
\end{equation}
where $\cdot_k$ denotes the $k$-th entry of the vector. Let $g^k_j =  \sum_{i \in \mcB(k, h_j)} f_{ij}$,\footnote{Note that the sum is in the real space but can only be $1$ or $0$. Thus $g^k_j$ is still in $\GF_2$.} then
\begin{equation}
\sum_{j = 1}^\checkdegree g^k_j = 0 \qquad \text{for all }k \in [m]
\end{equation}
where the addition is in $\GF_2$.
\end{enumerate} 
\end{definition}
\begin{lemma}
\label{lemma.valid_codeword}
In $\GF_{2^m}$, let $\mcC$ be the set of codewords that correspond to check $\bfh = (h_1,h_2,\dots,h_\checkdegree)$. Then
\begin{equation}
\embedmat(\mcC) = \mcE,
\end{equation}
where $\mcE$ is defined by Definition~\ref{def.validembedding}. 
\end{lemma}
\begin{proof}
See Appendix~\ref{proof.valid_codeword}.
\end{proof}
\begin{proposition}
\label{proposition.number_bit_constraints}
In $\GF_{2^m}$, for all $k \in [m]$ and $h \neq 0$, 
\begin{equation}
|\mcB(k,h)| = 2^{m-1}.
\end{equation}
\end{proposition}
\begin{proof}
See Appendix~\ref{proof.number_bit_constraints}.
\end{proof}

We can obtain a similar definition for \nameTheCW by changing two statements in Definition~\ref{def.validembedding}. First, the dimension of valid embedded matrices become $q \times \checkdegree$. Second, condition $(b)$ from Definition~\ref{def.validembedding} becomes $\sum_{i = 0}^{q-1} f_{ij} = 1$. Note that here the vector index starts from $0$. Using this definition, a similar statement holds for \nameTheCW (cf. Lemma~\ref{lemma.valid_cw_codeword}). In Appendix~\ref{appendix.def_for_cw_embedding} we show the equivalent definition and lemma for the \nameTheCW. 

\begin{definition}
\label{def.factor_graph_SPC}
The conditions in Definition~\ref{def.validembedding} can be represented using a factor graph. In this graph, there are three types of nodes: (i) \emph{variable nodes}, (ii) \emph{parity-check nodes} and (iii) \emph{at-most-one-on-check node}. Each variable node corresponds to an entry $f_{ij}$ for $1\leq i \leq 2^m-1$ and $1\leq j \leq \checkdegree$. Each parity-check node corresponds to a parity-check of the $k$-th bit where $k \in [m]$; it connects all variable nodes in $\mcB(k,h_j)$ for all $j\in[d_c]$. Each at-most-one-on-check node corresponds to a constraint $\sum_{i = 1}^{q-1} f_{ij} \leq 1$. We use circles to represent variables nodes and squares to represent parity-check nodes. In addition, we use triangles to represent at-most-one-on-check nodes because the polytope we get by relaxing this constraint is $\PS_{q - 1}$.
\end{definition}
\begin{example}
\label{eg.valid_codeword}
Consider $\GF_{2^2}$ and let $\bfh = (1,2,3)$. Let $\mcC$ be the SPC code defined by $\bfh$. For a word $\bfc \in \GF_{2^2}^3$, denote by 
$$\bfF := \embedmat(\bfc) = \begin{pmatrix}
f_{11} & f_{12} & f_{13}\\
f_{21} & f_{22} & f_{23}\\
f_{31} & f_{32} & f_{33}\\
\end{pmatrix},$$ 
and denote by
$$\bfB := \embedbitmat(\bfc) = \begin{pmatrix}
b_{11} & b_{12} & b_{13}\\
b_{21} & b_{22} & b_{23}\\
\end{pmatrix}.$$
By Lemma~\ref{lemma.valid_codeword}, $\bfc \in \mcC$ if and only if $\bfF$ satisfies the following conditions:
\begin{enumerate}
\item[$(a)$] $f_{ij} \in \{0,1\}$.
\item[$(b)$] $\sum_{i = 1}^{3} f_{ij} \leq 1$, where addition is in $\real$. 
\item[$(c)$] We can calculate from Table~\ref{table.integer_rep_gf4} that $\mcB(1,1) = \{1,3\}$, $\mcB(2,1)= \{2,3\}$, $\mcB(1,2) = \{2,3\}$, $\mcB(2,2)= \{1,2\}$, $\mcB(1,2) = \{1,3\}$, and $\mcB(2,3)= \{1,3\}$. (Cf.~\eqref{nonbinary.eq.b_set_definition} for the definition of $\mcB(k,h)$). Therefore by Definition~\ref{def.validembedding}, $\bfg^2 = (f_{21} + f_{31}, f_{12} + f_{22}, f_{13} + f_{33})$ and $\sum_j g^2_j = 0$ (in $\GF_2$). Similarly, $\bfg^1 = (f_{11} + f_{31}, f_{22} + f_{32}, f_{13} + f_{23})$ and $\sum_j g^1_j = 0$ (in $\GF_2$). 
\end{enumerate}
Using Definition~\ref{def.factor_graph_SPC}, we draw the factor graph in Figure~\ref{fig.factor_graph_for_polytope}. Each square represents a parity-check and each triangle represents an at-most-one-on-check.
\end{example}
\begin{figure}
\psfrag{F11}{\scriptsize{\vspace{-0.5mm}\hspace{-0.5mm}$f_{11}$}}
\psfrag{F21}{\scriptsize{\vspace{-0.5mm}\hspace{-0.5mm}$f_{21}$}}
\psfrag{F31}{\scriptsize{\vspace{-0.5mm}\hspace{-0.5mm}$f_{31}$}}
\psfrag{F12}{\scriptsize{\vspace{-0.5mm}\hspace{-0.5mm}$f_{12}$}}
\psfrag{F22}{\scriptsize{\vspace{-0.5mm}\hspace{-0.5mm}$f_{22}$}}
\psfrag{F32}{\scriptsize{\vspace{-0.5mm}\hspace{-0.5mm}$f_{32}$}}
\psfrag{F13}{\scriptsize{\vspace{-0.5mm}\hspace{-0.5mm}$f_{13}$}}
\psfrag{F23}{\scriptsize{\vspace{-0.5mm}\hspace{-0.5mm}$f_{23}$}}
\psfrag{F33}{\scriptsize{\vspace{-0.5mm}\hspace{-0.5mm}$f_{33}$}}
\psfrag{g1}{\scriptsize{$\sum_j g^2_j = 0$}}
\psfrag{g2}{\scriptsize{$\sum_j g^1_j = 0$}}
\begin{center}
\includegraphics[width = 9.2cm]{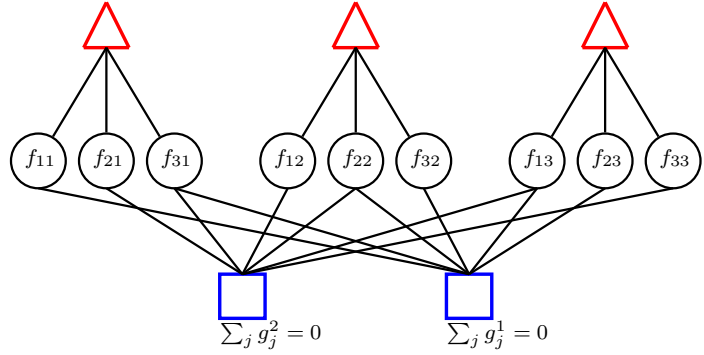}
\end{center}
\caption{The factor graph for valid embeddings defined by check $(1,2,3)\in\GF_{2^2}^3$. Each square represents a parity-check node; each triangle represents an at-most-one-on-check node.}
\label{fig.factor_graph_for_polytope}
\end{figure}

One important remark regarding Definition~\ref{def.validembedding} is that it is not the unique way of defining valid embeddings. In fact in~\cite{honda2012fast}, the authors' relaxation is implicitly based on the conditions in Lemma~\ref{lemma.equiv_integer_constraint} below. 
\begin{lemma}
\label{lemma.equiv_integer_constraint}
Consider \nameF and let $\mcF$ be the set of $(q-1)\times \checkdegree$ binary matrices defined by the first two conditions of Definition~\ref{def.validembedding} but with the third condition replaced by the condition $(c^*)$ defined below, then $\mcF = \mcE$. The complete set of conditions are
\begin{enumerate}
\item[$(a)$] $f_{ij} \in \{0,1\}$, where $f_{ij}$ denotes the entry in the $i$-th row and $j$-th column of matrix $\bfF$.
\item[$(b)$] $\sum_{i = 1}^{q-1} f_{ij} \leq 1$ where $q = 2^m$.
\item[$(c^*)$] Let $\K$ be any nonempty subset of $[m]$. For any non-zero $h\in \GF_{2^m}$, let $\mcBB(\K, h) := \left\{\alpha| \sum_{k \in \K}\embedbit(h \alpha)_k = 1\right\}$, where $\embedbit(h \alpha)_k$ denotes the $k$-th entry of the vector $\embedbit(h \alpha)$. Let $g^\K_j =  \sum_{i \in \mcBB(\K, h_j)} f_{ij}$, then $\sum_{j = 1}^\checkdegree g^\K_j = 0$ for all $\K \subset [m]$, $\K\neq \emptyset$, where the addition is in $\GF_{2}$.
\end{enumerate} 
\end{lemma}
\begin{proof}
See Appendix~\ref{proof.equiv_integer_constraint}
\end{proof}
\begin{lemma}
\label{lemma.number_bitset_constraints}
In $\GF_{2^m}$, for all $\K \subset [m]$ and $h \neq 0$, 
\begin{equation*}
|\mcBB(\K,h)| = 2^{m-1}.
\end{equation*}
\end{lemma}
\begin{proof}
See Appendix~\ref{proof.number_bitset_constraints}.
\end{proof}
\subsubsection*{Remarks}
Note that many redundant constraints are used to define $\mcF$ in Lemma~\ref{lemma.equiv_integer_constraint}. Similar to redundant parity-checks for binary linear codes (see e.g.,~\cite{zhang2012adaptive}), these redundant constraints can be used to tighten the polytope to be discussed in Section~\ref{section.relaxation}. On the other hand, adding such extra constraints generally increases computational complexity. Especially, note that there are $m$ parity-checks in Definition~\ref{def.validembedding} whereas the number of parity-checks in Lemma~\ref{lemma.equiv_integer_constraint} is $2^{m-1}$.

\section{Geometry of the relaxation of embeddings}
\label{section.relaxation}
In this section, we consider embeddings of non-binary SPC codes and study the polytopes obtained from these embeddings. There are two ways of obtaining polytopes relevant to LP decoding. In~\cite{flanagan2009linearprogramming}, the authors study the convex hull of $\embedmat(\mcC)$. However, this polytope is not easy to characterize. On the other hand, using the factor graph representation of SPC code embeddings from the previous section we can get a relaxation of $\conv(\embedmat(\mcC))$ following the methodology of LP decoding of binary LDPC codes in~\cite{feldman2005using}. In other words, the factor graph representation of embeddings resembles a mini binary (``generalized'') LDPC code. Instead of taking the convex hull of all codewords, we can take the intersection of many local polytopes for simplicity. This idea is first seen in~\cite{honda2012fast} for \nameTheCW. We first restate the relaxation method in~\cite{honda2012fast} for \nameF; readers are referred to Appendix~\ref{appendix.def_for_cw_embedding} for the corresponding definitions when \nameTheCW is used. We note that this restatement is important because it allows us to leverage off-the-shelf techniques to develop an efficient algorithm. We then show a tightness result of the parity polytope, which can imply a tightness result on the polytope of SPC code embeddings. Finally, we introduce several rotation properties of the polytope of SPC code embeddings that are useful for developing an efficient decoding algorithm.

\subsection{Relaxation of single parity-check code embeddings}
We first describe the relaxation considered in~\cite{flanagan2009linearprogramming}.
\begin{definition}
\label{def.tight_code_polytope}
Let $\mcC$ be a non-binary SPC code defined by check vector $\bfh$. Denote by $\tightcodepolytope$ the ``tight code polytope'' for \nameF where 
$$\tightcodepolytope = \conv(\embedmat(\mcC)).$$ 
\end{definition}

The following definition extends the relaxation in~\cite{honda2012fast} to \nameF.
\begin{definition}
\label{def.relaxed_code_polytope}
Let $\mcC$ be a non-binary SPC code defined by check vector $\bfh$. In $\GF_{2^m}$ denote by $\relaxedcodepolytope$ the ``relaxed code polytope''\footnote{From the visual appearances of the letters, one can think of $\relaxedcodepolytope$ as a relaxation of $\tightcodepolytope$.} for \nameF where a $(q-1) \times \checkdegree$ matrix $\bfF \in \relaxedcodepolytope$ if and only if the following constraints hold:
\begin{enumerate}
\item[$(a)$] $f_{ij} \in [0,1]$.
\item[$(b)$] $\sum_{i = 1}^{q-1} f_{ij} \leq 1$.
\item[$(c)$] Let $\mcBB(\K,h)$ be the same set defined in Lemma~\ref{lemma.equiv_integer_constraint}. Let $\bfg^\K$ be a vector of length $\checkdegree$ such that $g^\K_j = \sum_{i \in \mcBB(\K, h_j)} f_{ij}$, then
\begin{equation}
\bfg^\K \in \PP_{\checkdegree} \qquad \text{for all }\K \subset [m] \text{ and } \K\neq \emptyset
\end{equation}
\end{enumerate}
\end{definition} 
We note that in Definition~\ref{def.relaxed_code_polytope} conditions $(a)$ and $(b)$ define the simplex $\PS_{q-1}$. In addition, condition $(c)$ leverages parity polytope $\PP_{\checkdegree}$. Both polytopes are studied extensively in the literature (e.g.~\cite{barman2013decomposition} and~\cite{duchi2008efficient}). We show how to leverage this definition in developing efficient decoding algorithms in Section~\ref{section.LPandADMM}.
\subsection{A tightness result for the relaxed code polytope} 
\label{subsection.tightness}
We note that Definition~\ref{def.relaxed_code_polytope} is not the only way to relax the constraints in Lemma~\ref{lemma.equiv_integer_constraint}. In this subsection we study another possible relaxation and show that the relaxation in Definition~\ref{def.relaxed_code_polytope} is tighter.

Consider $\bfh = (1,2,3) \in \GF_{2^2}^3$ as a motivating example. Lemma~\ref{lemma.equiv_integer_constraint} requires that $\bfg^{\{2\}}$ has even parity ($\K = \{1\}$ in this case). This means $(f_{21} + f_{31}) + (f_{12} + f_{22}) + (f_{13} + f_{33}) = 0$ in $\GF_{2}$ (cf. Example~\ref{eg.valid_codeword}). Now consider the following two possible relaxations. We let $0\leq f_{ij} \leq 1$ and $\sum_i f_{ij} \leq 1$ for both cases.
\begin{enumerate}
\item[(i)] Vector $(f_{21},f_{31}, f_{22}, f_{32}, f_{23}, f_{33}) \in \PP_6$.
\item[(ii)] Vector $(f_{21} + f_{31}, f_{22} + f_{32}, f_{23} + f_{33}) \in \PP_3$ (the same relaxation used in Definition~\ref{def.relaxed_code_polytope}).
\end{enumerate} 

One would argue that both relaxations are applicable to LP decoding because the integral points of both relaxations are embeddings that satisfy Lemma~\ref{lemma.equiv_integer_constraint}. We justify using Proposition~\ref{prop.ppd_tight} that  statement (ii) implies statement (i) for more general cases. Thus, Definition~\ref{def.relaxed_code_polytope} is a better relaxation method. 
\begin{proposition}
\label{prop.ppd_tight}
Let $$\bfx := (\bfx^1,\bfx^2,\dots,\bfx^s)$$ be a non-negative vector obtained by concatenating $s$ row vectors $\bfx^i$, where each $\bfx^i$ is a length-$t$ vector. Let $\bfy := (\|\bfx^1\|_1, \|\bfx^2\|_1,\dots,\|\bfx^s\|_1)$. Then $\bfy \in \PP_s$ implies $\bfx \in \PP_{st}$.
\end{proposition}
\begin{proof}
See Appendix~\ref{proof.ppd_tight}.
\end{proof}

Note that the converse of Proposition~\ref{prop.ppd_tight} does not hold. That is, if $\bfx \in \PP_{st}$ and $\|\bfx^i\|_1 \leq 1$ for all $i$ $\bfy$ does not necessarily belong to $\PP_s$. As a counter example, let $\bfx = (0.5, 0.5, 0,0,0,0)$. Then $\bfx \in \PP_{6}$ but $\bfy = (1,0,0) \notin \PP_3$.

By applying Proposition~\ref{prop.ppd_tight} to our example, we conclude that $(f_{21} + f_{31}, f_{22} + f_{32}, f_{23} + f_{33}) \in \PP_3$ implies that $(f_{21},f_{31}, f_{22}, f_{32}, f_{23}, f_{33}) \in \PP_6$.

\subsection{Rotation}
\label{subsec.rotation}
We describe a rotation operation that normalizes the geometry under consideration. This process simplifies our study of geometries because one only needs to work on a canonical object instead of different polytopes for different checks. A similar idea is also found in BP decoding in~\cite{declercq2007decoding}.

As a side effect, the normalization results in a column permutation property for valid embedded matrices. In particular, post-normalization, any permutation when applied to columns of a valid SPC code embedded matrix is a valid SPC code embedded matrix. Note that the convex hull of all permutations of a vector can be defined using (``single-variate'' a.k.a., ``vector'') majorization~\cite{marshall2009inequalities}. This observation is the key to the algorithm of projection onto the parity polytope~\cite{barman2013decomposition}. Similarly, the convex hull of all column permutations of a matrix can be defined using multivariate majorization~\cite{marshall2009inequalities}. Thus it may be possible to use results for multivariate majorization to characterize $\conv(\embedmat(\mcC))$. We leave this as future work. 

We first describe the rotation operation for $\conv(\embedmat(\mcC))$. We then show that the same operation can be applied to the relaxed code polytope $\relaxedcodepolytope$.

\subsubsection{Rotation for the tight code polytope}
Consider an arbitrary finite field $\GF_{q}$ and let $\mcCN$ be the ``normalized'' set of codewords defined as $\mcCN := \{\bfx | \sum_{j = 1}^\checkdegree x_j = 0\}$. We focus on \nameF and let $\mcEN := \embedmat(\mcCN)$, where the addition is in $\GF_{2^m}$. Then for any check that is not ``normalized'' (i.e., contains entries not equal to $1$), we can obtain the set of codewords $\mcC$ from $\mcCN$ by dividing each entry by the corresponding check value in $\GF_{q}$. Formally, let $\bfh = (h_1,\dots,h_\checkdegree)$ be the check vector and $h_j \neq 0$. Then 
\begin{equation*}
\mcC = \{\bfy| y_j = x_j/h_j \text{ for all $j$; }  \bfx \in\mcCN \}.
\end{equation*}
Note that in the equation above $x_j\neq 0$ and $y_j \neq 0$. We define the rotation matrix $\rotmat$ as follows:
\begin{definition}
\label{def.permutation}
Let $h$ be a non-zero element in $\GF_{q}$. Let $\rotmat(q,h)$ be a $(q-1)\times(q-1)$ permutation matrix satisfying the following equation
\begin{equation*}
\rotmat(q,h)_{ij} = 
\begin{cases}
1 & \quad \text{ if } i \cdot h = j, \\
0 & \quad \text{ otherwise,}
\end{cases}
\end{equation*}
where the multiplication is in $\GF_{q}$. For a vector $\bfh \in \GF_{q}^\checkdegree$ that $h_j\neq 0 $ for all $j \in [d_c]$, let $$\rotmat(q,\bfh) = \diag(\rotmat(q,h_1),\dots,\rotmat(q,h_\checkdegree)).$$
\end{definition}
Using the permutation matrix $\rotmat(q, \bfh)$, we can obtain $\mcE$ by rotating every vector in $\mcEN$:
\begin{equation*}
\mcE := \{\bff| \bff = \rotmat(q, \bfh)\bfe , \bfe \in\mcEN\}.
\end{equation*}
\begin{example}
Let $q = 5$ and $h = 3$. Using division in $\GF_{5}$, we obtain $1 \cdot 3^{-1} = 2$, $2 \cdot 3^{-1} = 4$, $3 \cdot 3^{-1} = 1$ and $4 \cdot 3^{-1} = 3$. Therefore $$\rotmat(5,3) = \begin{pmatrix}
0 & 0 & 1 & 0 \\
1 & 0 & 0 & 0 \\
0 & 0 & 0 & 1 \\
0 & 1 & 0 & 0
\end{pmatrix}.
$$ 
Take the equation $1 \cdot 3^{-1} = 2$ for example. The respective embedded vectors of $1$ and $2$ are $(1,0,0,0)$ and $(0,1,0,0)$. It is easy to verify that $(0,1,0,0)^T = \rotmat(5,3)(1,0,0,0)^T$.  
\end{example}

We use $\tightcodepolytopeN$ to denote the tight code polytope defined by the all-ones check. 
The following lemma shows that we can obtain $\tightcodepolytopeN$ by rotating $\tightcodepolytope$ using $\rotmat(q, \bfh)$.
\begin{lemma}
\label{lemma.general_rotation_equivalence}
Let $\bfh$ be a check such that all entries of $\bfh$ are non-zero. Let $\tightcodepolytope$ be the convex hull of $\mcE$ and $\tightcodepolytopeN$ be the convex hull of $\mcEN$. Then a vector $\bfu \in \tightcodepolytope$ if and only if $\rotmat(q, \bfh)^{-1} \bfu \in \tightcodepolytopeN$. In other words, a vector $\bfw \in \tightcodepolytopeN$ if and only if $\rotmat(q, \bfh) \bfw \in \tightcodepolytope$.
\end{lemma}
\begin{proof}
See Appendix~\ref{proof.rotation}.
\end{proof}

This lemma implies that having different check values is nothing but a rotation of the normalized code polytope. The following lemma states that projections onto $\tightcodepolytope$ for arbitrary checks can be performed via projections onto $\tightcodepolytopeN$. 
\begin{lemma}
\label{lemma.general_rotation_projection}
Let $\bfh$ be a check such that all entries of $\bfh$ are non-zero. Let $\tightcodepolytope$ be the convex hull of $\mcE$ and let $\tightcodepolytopeN$ be the convex hull of $\mcEN$. Let $\bfv$ be any vector and let $\bfvN := \rotmat(q, \bfh)^{-1} \bfv$. Then, if $\bfuN$ is the projection of $\bfvN$ onto $\tightcodepolytopeN$, then $\bfu := \rotmat(q, \bfh) \bfuN$ is the projection of $\bfv$ onto $\tightcodepolytope$. 
\end{lemma}
\begin{proof}
See Appendix~\ref{proof.rotation}.
\end{proof}
Lemma~\ref{lemma.general_rotation_projection} shows that we can obtain the projection of a vector $\bfv$ onto any tight code polytope $\tightcodepolytope$ in three steps: First rotate the vector and get $\bfvN$, then project $\bfvN$ onto the normalized polytope $\tightcodepolytopeN$, and finally rotate the projection result $\bfvN$ back to the original coordinate system.

\subsubsection{Rotation for the relaxed code polytope}
In $\GF_{2^m}$, we can relax $\tightcodepolytope$ using Definition~\ref{def.relaxed_code_polytope} and obtain a relaxed code polytope $\relaxedcodepolytope$. The same permutation matrix $\rotmat(q,\bfh)$ can also be applied to $\relaxedcodepolytope$.
\begin{lemma}
\label{lemma.rotation_equivalence}
Let $\bfh$ be a check such that all entries of $\bfh$ are non-zero. Let $\relaxedcodepolytope$ be the relaxed code polytope obtained by Definition~\ref{def.relaxed_code_polytope} for check $\bfh$ and let $\relaxedcodepolytopeN$ be the relaxed code polytope obtained by Definition~\ref{def.relaxed_code_polytope} for the all-ones check. Then a vector $\bfv \in \relaxedcodepolytope$ if and only if $\rotmat(q, \bfh)^{-1} \bfv \in \relaxedcodepolytopeN$. In other words, a vector $\bfw \in \relaxedcodepolytopeN$ if and only if $\rotmat(q, \bfh) \bfw \in \relaxedcodepolytope$.
\end{lemma}
\begin{proof}
See Appendix~\ref{proof.rotation}.
\end{proof}
Similar to the previous case, we have the following lemma which applies to the projection operation.
\begin{lemma}
\label{lemma.rotation_projection}
Let $\bfh$ be a check such that all entries of $\bfh$ are non-zero. Let $\relaxedcodepolytope$ be the relaxed code polytope obtained by Definition~\ref{def.relaxed_code_polytope} for check $\bfh$ and let $\relaxedcodepolytopeN$ be the relaxed code polytope obtained by Definition~\ref{def.relaxed_code_polytope} for the all-one check. Let $\bfv$ be any vector and let $\bfvN := \rotmat(q, \bfh)^{-1} \bfv$. Then, if $\bfuN$ is the projection of $\bfvN$ onto $\relaxedcodepolytopeN$, then $\bfu := \rotmat(q, \bfh) \bfuN$ is the projection of $\bfv$ onto $\relaxedcodepolytope$. 
\end{lemma}
\begin{proof} 
See Appendix~\ref{proof.rotation}.
\end{proof}

\subsection{Remarks}
For $\GF_{2^2}$, we can obtain a characterization of SPC code embeddings simpler than Definition~\ref{def.validembedding}. In addition, we note that it is mentioned by~\cite{honda2012fast} (and we paraphrase) that $\tightcodepolytope = \relaxedcodepolytope$ for codes in $\GF_{2^2}$. Yet no proof was given in~\cite{honda2012fast}. We validated this statement numerically and believe that it holds. These results for $\GF_{2^2}$ are discussed in Appendix~\ref{appendix.conjecture}.

For $\GF_{2^m}$ where $m \geq 3$, it is obvious that $\tightcodepolytope \subset \relaxedcodepolytope$. However, we do not have evidence that $\tightcodepolytope = \relaxedcodepolytope$. In fact, simulation results in Section~\ref{subsection.snr_performance} indicate that $\tightcodepolytope \neq \relaxedcodepolytope$ because the WER performance of LP decoding based on $\relaxedcodepolytope$ is worse than that based on $\tightcodepolytope$. Understanding how loose $\relaxedcodepolytope$ is compared to $\tightcodepolytope$ is important future work.

\section{LP decoding and ADMM formulations}
\label{section.LPandADMM}
LP decoding of non-binary codes was first introduced by Flanagan \emph{et al.} in~\cite{flanagan2009linearprogramming}. The LP decoding problem in~\cite{flanagan2009linearprogramming} is based on the embedding method presented in Definition~\ref{def.flanagan_embedding}. Honda and Yamamoto proposed a different LP decoding problem in~\cite{honda2012fast}, one that uses the embedding method presented in Definition~\ref{def.constant_weight_embedding}. However, no connection between the problem in~\cite{flanagan2009linearprogramming} and the problem in~\cite{honda2012fast} has been made. We first review these two LP decoding problems and show that the choice of embedding method does not affect the results of LP decoding. We then develop two ADMM algorithms for LP decoding. The first algorithm is generic and can be applied to arbitrary fields. However, it requires a sub-routine that can project either onto the tight code polytope or onto the relaxed code polytope. In our previous work~\cite{liu2014admm}, we proposed an ADMM projection algorithm to project onto the relaxed code polytope $\relaxedcodepolytope$. The resulting LP decoding algorithm executes ADMM iterations such that each iteration contains $\checknumber$ ADMM projection sub-routines, where $\checknumber$ is the number of checks. In other words there are two levels of ADMM algorithms: A global ADMM decoding, and $\checknumber$ ADMM sub-routines per iteration of that global ADMM decoding. We observe that this algorithm is not especially efficient because even if the projection sub-routines converge reasonably fast, the multiplicity of sub-routines is large ($\checknumber$ per decoding iteration). To solve this problem, we introduce another ADMM formulation for codes in $\GF_{2^m}$ in Section~\ref{section.LPandADMM_noADMMProj} in which no ADMM sub-routine is required. This formulation leverages the factor graph representation of embeddings described in Section~\ref{section.combinatorics_SPC_GF2m}.

\subsection{Review of LP decoding problems}
\label{subsec.review_LP}
We first review the LP decoding problem proposed by Flanagan \emph{et al.} in~\cite{flanagan2009linearprogramming}.
We consider a linear code specified by a parity-check matrix $\bfH$, where each entry $h_{ji} \in \GF_q$. Let $\Sigma$ be the output space of the channel. In~\cite{flanagan2009linearprogramming}, a function $\embedllr: \Sigma \mapsto (\real \bigcup \{\pm \infty \} )^{q-1}$ is introduced, it is defined as 
\begin{equation*}
\embedllr(y) := \bflambda = (\lambda_1,\dots,\lambda_{q-1}),
\end{equation*}
where for each $y \in \Sigma$, $\delta \in \GF_{q} \setminus \{0\}$, $\lambda_{\delta}= \log \left(\frac{\Pr[Y = y|X = 0]}{\Pr[Y = y|X = \delta]}\right)$. Similar to $\embedvec(\cdot)$ defined in Section~\ref{sec.embedding_finite_fields}, we define $\embedllrvec(\bfy) = (\embedllr(y_1)|\embedllr(y_2)|\dots|\embedllr(y_n))^T$. Maximum likelihood decoding can then be rewritten as (cf.~\cite{flanagan2009linearprogramming})
\begin{align*}
\hat{\bfc} &= \argmin_{\bfc \in \mcC} \sum_{i = 1}^{n}  \log \left(\frac{\Pr[y_i|0]}{\Pr[y_i|c_i]}\right)\\
&=  \argmin_{\bfc \in \mcC} \sum_{i = 1}^{n}  \embedllr(y_i) \embed(c_i)^T\\
&=  \argmin_{\bfc \in \mcC} \embedllrvec(\bfy) \embedvec(\bfc)^T.
\end{align*}
In~\cite{flanagan2009linearprogramming}, ML decoding is relaxed to a linear program. We can write the LP decoding problem (using the ``tight'' ``Flanagan'' representation, \textbf{FT}) as
\begin{equation}
\label{eq.lpdecoding}
\begin{split}
\mbox{\LPDFlanagan:}\quad  \min \quad  &  \embedllrvec(\bfy)^T \bfx  \\
\st  \quad & \bfP_j \bfx \in \tightcodepolytope_j, \forall j\in \mcJ,
\end{split}
\end{equation}
where $\bfP_j$ selects the variables that participate in the $j$-th check and $\tightcodepolytope_j = \conv(\embedvec(\mcC_j))$. Here $\mcC_j$ is the SPC code defined by the non-zero entries of the $j$-th check. 

Using the relaxation of Definition~\ref{def.relaxed_code_polytope}, we can formulate a different and further relaxed LP decoding problem. Note that this problem only applies to codes in $\GF_{2^m}$. We write this (``Flanagan'' ``relaxed'', \textbf{FR})  problem as
\begin{equation}
\label{eq.lpdecoding_relax}
\begin{split}
\mbox{\LPDFlanaganRelax:}\quad  \min \quad  &  \embedllrvec(\bfy)^T \bfx  \\
\st  \quad & \bfP_j \bfx \in \relaxedcodepolytope_j, \forall j\in \mcJ,
\end{split}
\end{equation}
where $\bfP_j$ is the same as the previous problem and $\relaxedcodepolytope_j$ is defined by Definition~\ref{def.relaxed_code_polytope} for the $j$-th check.

Since the LP decoding problem in~\cite{honda2012fast} uses \nameTheCW, both the objective function and the constraint set are different from those in~\cite{flanagan2009linearprogramming}. Let
\begin{equation*}
\cwembedllr(y) := \bflambda = (\lambda_0,\dots,\lambda_{q-1}),
\end{equation*}
where for each $y \in \Sigma$, $\delta \in \GF_{q}$, $\lambda_{\delta}= \log \left(\frac{1}{\Pr[Y = y|X = \delta]}\right)$. The LP decoding problem (``constant-weight'' ``relaxed'', \textbf{CR}) in~\cite{honda2012fast} is
\begin{equation}
\label{eq.cwlpdecoding_relax}
\begin{split}
\mbox{\LPDCWRelax:}\quad  \min \quad  & \cwembedllrvec(\bfy) \bfx  \\
\st  \quad & \bfP_j' \bfx \in \cwrelaxedcodepolytope_j, \forall j\in \mcJ,
\end{split}
\end{equation}
where $\bfP_j'$ selects the variables that participate in the $j$-th check and $\cwrelaxedcodepolytope_j$ is the relaxed code polytope when \nameTheCW is used (cf. Definition~\ref{def.cw_code_polytopes}).  We can formulate a different, ``constant-weight'', ``tighter'' (\textbf{CT}) LP decoding problem using $\cwtightcodepolytope_j:= \conv(\cwembedvec(\codebook_j))$ as
\begin{equation}
\begin{split}
\mbox{\LPDCW:}\quad  \min \quad  &  \cwembedllrvec(\bfy)^T \bfx  \\
\st  \quad & \bfP_j' \bfx \in \cwtightcodepolytope_j, \forall j\in \mcJ.
\end{split}
\end{equation}

\begin{example}
\label{example.toy_code_gf4}
In this example we demonstrate the selection matrices $\bfP$ and $\bfP'$. Consider a code in $\GF_{2^2}$ defined by the following parity-check matrix:
$
\bfH = \begin{pmatrix}
1 & 2 & 2 & 3\\
2 & 0 & 1 & 2
\end{pmatrix}
$. Then $\bfP_1 = \bfI_{12\times 12}$ and $
\bfP_2 = \begin{pmatrix}
\bfI_{3\times3} & \bfzero_{3\times3} & \bfzero_{3\times3} & \bfzero_{3\times3}\\
\bfzero_{3\times3} & \bfzero_{3\times3} & \bfI_{3\times3} & \bfzero_{3\times3}\\
\bfzero_{3\times3} & \bfzero_{3\times3} & \bfzero_{3\times3} & \bfI_{3\times3}
\end{pmatrix}
$.
 $\bfP_1' = \bfI_{16\times 16}$ and $
\bfP_2' = \begin{pmatrix}
\bfI_{4\times4} & \bfzero_{4\times4} & \bfzero_{4\times4} & \bfzero_{4\times4}\\
\bfzero_{4\times4} & \bfzero_{4\times4} & \bfI_{4\times4} & \bfzero_{4\times4}\\
\bfzero_{4\times4} & \bfzero_{4\times4} & \bfzero_{4\times4} & \bfI_{4\times4}
\end{pmatrix}
$.
\end{example}

\subsection{Equivalence properties between two embedding methods}
\label{subsec.equivalence_embeddings}
In this section, we show several equivalence properties between \nameF and \nameTheCW in the context of LP decoding. As a reminder, we consider two LP decoding objectives and two types of LP decoding constraints. Therefore there are four LP decoding problems under consideration. We summarize the notation in Table~\ref{table.LPproblems}.
\begin{table}[ht]
  \centering
    \begin{tabular}{| c | c | c |}
    \hline
    						& \textbf{T}ight code polytope &	\textbf{R}elaxed code polytope\\
    \hline\hline
    \textbf{F}lanagan Emb.& \LPDFlanagan & \LPDFlanaganRelax \\ \hline
    \textbf{C}onst.-weight Emb. & \LPDCW & \LPDCWRelax \\
    \hline
  \end{tabular}
  \caption{Four LP decoding problems}
\label{table.LPproblems}
\end{table}

\subsubsection{Equivalence between \LPDFlanagan and \LPDCW}
We first focus on LP decoding using the tight code polytope $\tightcodepolytope$.
\begin{theorem}
\label{theorem.equiv_LP}
Consider any finite field $\GF_{q}$ and let $\bfy \in \Sigma$. Let $\bff_i$ be a length-$(q-1)$ vector where $i \in \mcI$. If $\bffhat = (\bff_1, \bff_2, \dots, \bff_{\blocklength})^T$ is the \minimizer of \LPDFlanagan, then $\bffbar = (1- \|\bff_1\|_1,\bff_1,1- \|\bff_2\|_1,\bff_2, \dots,1-\|\bff_{\blocklength}\|_1,\bff_{\blocklength})^T$ is the \minimizer of \LPDCW. Conversely, if $\bffbar = (a_1,\bff_1,a_2,\bff_2, \dots,a_{\blocklength},\bff_{\blocklength})^T$ is the \minimizer of \LPDCW for some $a_1,\dots,a_{\blocklength}$, then $\bffhat = (\bff_1, \bff_2, \dots, \bff_{\blocklength})^T$ is the \minimizer of \LPDFlanagan.
\end{theorem}
\begin{proof}
See Appendix~\ref{proof.equiv_LP}
\end{proof}

This theorem builds the equivalence in the sense that, if \LPDFlanagan and \LPDCW are given the same channel output vector, then there is a bijective mapping between their decoding results. It is then easy to show the following corollary. 
\begin{corollary}
\label{corollary.equiv_error_rate}
Both \LPDFlanagan and \LPDCW achieve the same symbol-error-rate and word-error-rate.
\end{corollary}
\begin{proof}
By Theorem~\ref{theorem.equiv_LP}, \LPDFlanagan decodes a channel output $\bfy$ to a fractional solution if and only if \LPDCW decodes $\bfy$ to a fractional solution. Thus the two decoders either both succeed or both fail.
\end{proof}
\begin{corollary}
Under the channel symmetry condition in the sense of~\cite{flanagan2009linearprogramming}, the probability that \LPDCW fails
is independent of the codeword that was transmitted.
\end{corollary}
\begin{proof}
This is due to Theorem~5.1 in~\cite{flanagan2009linearprogramming} and Corollary~\ref{corollary.equiv_error_rate}.
\end{proof}

\subsubsection{Equivalence between \LPDFlanaganRelax and \LPDCWRelax}
We now consider the LP decoding problem for codes in fields of characteristic two when the relaxed code polytope is used in LP decoding.
\begin{theorem}
\label{theorem.equiv_LP_relax}
Consider finite fields $\GF_{2^m}$ and let $\bfy \in \Sigma$. Let $\bff_i$ be a length-$(2^m-1)$ vector where $i \in \mcI$. If $\bffhat = (\bff_1, \bff_2, \dots, \bff_{\blocklength})^T$ is the \minimizer of \LPDFlanaganRelax, then $\bffbar = (1- \|\bff_1\|_1,\bff_1,1- \|\bff_2\|_1,\bff_2, \dots,1-\|\bff_{\blocklength}\|_1,\bff_{\blocklength})^T$ is the \minimizer of \LPDCWRelax. Conversely, if $\bffbar = (a_1,\bff_1,a_2,\bff_2, \dots,a_{\blocklength},\bff_{\blocklength})^T$ is the \minimizer of \LPDCWRelax for some $a_1,\dots,a_{\blocklength}$, then $\bffhat = (\bff_1, \bff_2, \dots, \bff_{\blocklength})^T$ is the \minimizer of \LPDFlanaganRelax.
\end{theorem}
\begin{proof}
See Appendix~\ref{proof.equiv_LP_relax}.
\end{proof}

\begin{corollary}
\label{corollary.equiv_error_rate_relaxed}
Both \LPDFlanaganRelax and \LPDCWRelax achieve the same symbol-error-rate and word-error-rate.
\end{corollary}
\begin{corollary}
\label{corollary.codewordindependence}
Under the channel symmetry condition in the sense of~\cite{flanagan2009linearprogramming}, the probability that \LPDFlanaganRelax fails is independent of the codeword that was transmitted.
\end{corollary}
\begin{proof}
This is due to Theorem~2 in~\cite{honda2012fast} and Theorem~\ref{theorem.equiv_LP_relax}. 
\end{proof}

We also provide our direct proof of this corollary in Appendix~\ref{proof.codewordindependence}. Our proof leverages the proof technique taken in~\cite{flanagan2009linearprogramming}. In addition, we show a symmetry property of $\relaxedcodepolytope$. 

\subsubsection{Remarks}
\begin{itemize}
\item Our results imply that one has the freedom to choose either \nameF or \nameTheCW for LP decoding. Both embedding methods yield the same error rates. 
\item We note that the relaxed versions of LP decoding (\LPDFlanaganRelax and \LPDCWRelax) are in general worse than the tighter versions (\LPDFlanagan and \LPDCW) in terms of error rates. 
\item LP decoding using \nameF requires fewer variables. In addition, we show in Section~\ref{section.numerical} that the ADMM LP decoder in Section~\ref{section.LPandADMM_noADMMProj}, which uses \nameF, converges faster than its \nameCW counterpart.
\end{itemize}
\subsection{Generic ADMM formulations of the LP decoding problem}
\label{subsec.genericADMM}
In this section, we derive a generic ADMM formulation for LP decoding of non-binary codes. 
We formulate the ADMM algorithm under the assumption that a sub-routine that projects a vector onto $\tightcodepolytope$ or $\relaxedcodepolytope$ is provided. We will use the formulation in this section to develop a penalized LP decoder in Section~\ref{section.penalized}.

We use the shorthand notation $\bfgamma := \embedllrvec(\bfy)$, and following the methodology of~\cite{barman2013decomposition}, cast the LP into a form solvable using ADMM. We introduce replicas $\bfz_j$ for all $j\in\mcJ$. Each $\bfz_j$ is a length-$(q-1)\checkdegree$ vector where $\checkdegree$ is the check degree. We then express~\eqref{eq.lpdecoding} in the following, equivalent, form:
\begin{equation}
\begin{split}
\quad  \min \quad  &  \bfgamma^T \bfx  \\
\st  \quad & \bfP_j \bfx  = \bfz_j,\\ 
\quad & \bfz_j \in \tightcodepolytope_j, \text{ for all } j\in \mcJ,\\
\quad & \bfx_i \in \PS_{q - 1}, \text{ for all } i \in \mcI,
\end{split}
\end{equation}
where $\bfx_i = (x_{(i-1)(q-1) + 1},x_{(i-1)(q-1) + 2},\dots, x_{i(q-1)})$ is the sub-vector selected from the $i$-th ($q-1$)-length block of $\bfx$. In other words, $\bfx_i$ corresponds to the embedded vector of the $i$-th non-binary symbol. $\PS_{q - 1}$ is the $q-1$ simplex defined in~\eqref{eq.def_ps}.

The augmented Lagrangian is
\begin{equation*}
\Larg_{\penpara} (\bfx, \bfz, \bflambda) =\bfgamma^T \bfx + 
\sum_{j\in\mcJ}\bflambda_j^T (\bfP_j \bfx - \bfz_j) + \frac{\penpara}{2}\sum_{j\in \mcJ}\|\bfP_j \bfx - \bfz_j\|_2^2.
\end{equation*}
ADMM iteratively performs the following updates:
\begin{align}
\bfx\text{-update: }\bfx^{k+1} &= \argmin_{\bfx} \Larg_{\penpara}(\bfx,\bfz^{k},\bflambda^{k}), \label{eq.admm_x_update}\\
\bfz\text{-update: }\bfz^{k+1} &= \argmin_{\bfz} \Larg_{\penpara}(\bfx^{k+1},\bfz,\bflambda^{k}), \label{eq.admm_z_update}\\
\bflambda\text{-update: }\bflambda_j^{k+1} &=\bflambda_j^k + \penpara\left(\bfP_j\bfx^{k+1}-\bfz_j^{k+1}\right). \label{eq.admm_lambda_update}
\end{align}
In the $x$-update we solve the following optimization problem:
\begin{equation}
\min_{\forall i \in\mcI, \bfx_i \in \PS_{q-1}} \Larg_{\penpara} (\bfx, \bfz, \bflambda).
\end{equation}
We first introduce some additional notation. Let $\bfgamma_i$ be the vector of log-likelihood ratios that correspond to the $i$-th symbol of the code. In other words, $$\bfgamma_i = (\gamma_{(i-1)(q-1) + 1},\gamma_{(i-1)(q-1) + 2},\dots, \gamma_{i(q-1)}).$$ Let $\bflambda_j^{(i)}$ be the length-$(q-1)$ sub-vector of $\bflambda_j$ that corresponds to the $i$-th symbol of the code. Similarly, we define $\bfz_j^{(i)}$ to be the sub-vector of $\bfz_j$ that correspond to the $i$-th symbol of the code. We then write $\Larg_{\penpara} (\bfx, \bfz, \bflambda)$ as
\begin{align*}
\Larg_{\penpara} (\bfx, \bfz, \bflambda) = \sum_{i = 1}^{\blocklength} \left(
\bfgamma_i^T \bfx_i + \sum_{j\in \Nev(i)} \bflambda_j^{(i) T} (\bfx_i - \bfz_j^{(i)})  + \frac{\penpara}{2}\sum_{j\in\Nev(i)}\|\bfx_i - \bfz_j^{(i)}\|_2^2
\right).
\end{align*}
Note that when $\bfz_j$ and $\bflambda_j$ are fixed for all $j \in \mcJ$, we can decouple $\bfx_i$ for all $i\in\mcI$ in the sense that they can be individually solved for. Therefore 
\begin{equation}
\begin{aligned}
\bfx_i^{k+1} &= \argmin_{\bfx_i \in \PS_{q-1}} \bfgamma_i^T \bfx_i + \sum_{j\in \Nev(i)} \bflambda_j^{(i) T} (\bfx_i - \bfz_j^{(i)}) + \frac{\penpara}{2}\sum_{j\in\Nev(i)}\|\bfx_i - \bfz_j^{(i)}\|_2^2\\
&= \argmin_{\bfx_i \in \PS_{q-1}} \|\bfx_i - \bfv_i\|_2^2,
\end{aligned}
\label{eq.xupdate_complete_square}
\end{equation}
where 
\begin{equation}
\label{eq.xupdate_average}
\begin{aligned}
\bfv_i = \frac{1}{|\Nev(i)|} \left[\sum_{j\in \Nev(i)}\left(\bfz_j^{(i)} - \frac{\bflambda_j^{(i)}}{\penpara}\right) - \frac{\bfgamma_i}{\penpara} \right].
\end{aligned}
\end{equation}
Thus, the $\bfx$-update is equivalent to a Euclidean projection onto $\PS_{q-1}$, which is solvable in linear time (i.e., $O(q)$) using techniques proposed in~\cite{duchi2008efficient}. 
Note that the $\bfv_i$ in~\eqref{eq.xupdate_average} is an average of the corresponding replicas (i.e., $\bfz_j^{(i)}$) plus adjustments from the Lagrange multipliers (i.e., $\bflambda_j^{(i)}$) and the log-likelihood ratios (i.e., $\bfgamma_i$).

In the $\bfz$-update we can solve for each $\bfz_j$ separately. When $\bfx_i$ and $\bflambda_j$ are fixed for all $i\in\mcI$ and $j\in\mcJ$ we complete the square with respect to each $\bfz_j$ and obtain the following update rule:
\begin{align}
\label{eq.zupdate_projection}
\bfz_j^{k+1} = \argmin_{\bfz_j \in \tightcodepolytope_j} \|\bfu_j - \bfz_j\|_2^2,
\end{align}
where $\bfu_j = \bfP_j \bfx + \bflambda_j/\penpara$. Thus the $\bfz$-update is equivalent to the Euclidean projection onto $\tightcodepolytope_j$. 

For \LPDFlanaganRelax, the respective $\bfz$-update is equivalent to projection onto $\relaxedcodepolytope_j$. That is,
\begin{align}
\label{eq.zupdate_projection_relaxed}
\bfz_j^{k+1} = \argmin_{\bfz_j \in \relaxedcodepolytope_j} \|\bfu_j - \bfz_j\|_2^2,
\end{align}
In our previous work~\cite{liu2014admm}, we proposed an ADMM algorithm to solve~\eqref{eq.zupdate_projection_relaxed}. We show in Section~\ref{section.LPandADMM_noADMMProj} that the complexity of the decoding algorithm in~\cite{liu2014admm} can be greatly reduced.
\subsection{Improved ADMM LP decoding for \LPDFlanaganRelax}
\label{section.LPandADMM_noADMMProj}
In this section, we focus on \LPDFlanaganRelax and propose an ADMM decoding algorithm that does not require a sub-routine that projects onto $\relaxedcodepolytope$. We note that it is easy to extend the techniques in this section to problem \LPDCWRelax. 

For an LDPC code, all entries of the embedded vector of a non-binary variable $i\in\mcI$ participate in $|\Nev(i)|$ non-binary checks. When embedding is used, each non-binary check can be decomposed by at-most-one-on-checks and parity-checks. We may think of these constraints having two hierarchies. The first hierarchy consists of constraints defined by the parity-check matrix of a non-binary code. The second hierarchy consists of constraints defined by the SPC code embeddings (i.e., Definition~\ref{def.relaxed_code_polytope}). We illustrated this in Example~\ref{example.two_level_factor_graph}.
\begin{example}
\label{example.two_level_factor_graph}
We build the factor graph of embeddings for an LDPC code in $\GF_{2^2}$. Assume that the $j$-th check is connected to three variable nodes $i_1$, $i_2$ and $i_3$. Then, the embedded variables ($9$ binary variables) are connected to the $j$-th non-binary check. This connection is captured by the matrix $\bfP_j$ in~\eqref{eq.lpdecoding}. We can think of this as the first hierarchy. There is another hierarchy of factor graph inside the $j$-th constraint set, namely, those constraints defined by Lemma~\ref{lemma.equiv_integer_constraint}. The variables are first permuted by the rotation steps defined in Section~\ref{subsec.rotation} where the permutations are captured by the matrix $\rotmat$. Then, they are connected to a normalized constraint set defined by three at-most-one-on-checks and three parity-checks. The three at-most-one-on-checks are connected to the vector triplets $\bfx_{i_1}$, $\bfx_{i_2}$ and $\bfx_{i_3}$ respectively. The three parity-check are specified by $\mcBB(\K,1)$ for $\K = \{1\},\{2\}$ and $\{1,2\}$; and are connected to the \textbf{permuted} entries of $\bfx_{i_1}$, $\bfx_{i_2}$ and $\bfx_{i_3}$. These connections are shown in Figure~\ref{fig.layered_factor_graph}.
\begin{figure}
\psfrag{&x_i1}{{$\bfx_{i_1}$}}
\psfrag{&x_i2}{{$\bfx_{i_2}$}}
\psfrag{&x_i3}{{$\bfx_{i_3}$}}
\psfrag{&...}{\scriptsize{$\vdots$}}
\psfrag{&P_j}{\scriptsize{$\bfP_j$}}
\psfrag{&D(q,h_j,i1)}{\scriptsize{\hspace{-3.7mm}$\rotmat(2^2,h_{ji_1})$}}
\psfrag{&D(q,h_j,i2)}{\scriptsize{\hspace{-3.7mm}$\rotmat(2^2,h_{ji_2})$}}
\psfrag{&D(q,h_j,i3)}{\scriptsize{\hspace{-3.7mm}$\rotmat(2^2,h_{ji_3})$}}
\psfrag{&checkj}{\scriptsize{\hspace{-3mm}$j$-th non-binary check}}
	\begin{center}
    \includegraphics[width= 5.55cm]{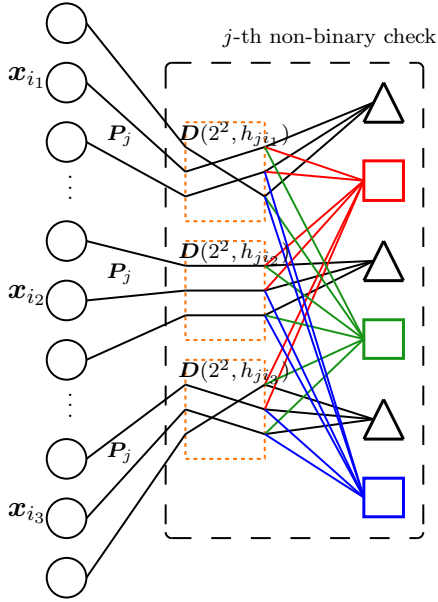}
    \end{center}
    \caption{Factor graph of embeddings showing two hierarchies. The code under consideration is in $\GF_{2^2}$. }
    \label{fig.layered_factor_graph}
\end{figure}
\end{example}

We note that we obtain Figure~\ref{fig.layered_factor_graph} by going through two hierarchies of constraints. However, the resulting graph is nothing but a factor graph of embeddings with two types of factor nodes: at-most-one-on-check nodes and parity-check nodes. The connections between variable nodes and factor nodes are specified by both the parity-check matrix of the code and the constraints specified for SPC codes (i.e., those in Lemma~\ref{lemma.equiv_integer_constraint}).

We now formally state the ADMM formulation based on the factor graph of embeddings. We will assume that the code has regular variable degree $\variabledegree$ for conciseness. However, our derivation can easily extend to irregular codes. By Lemma~\ref{lemma.rotation_equivalence}, we can rewrite the LP decoding problem using the normalized polytope $\relaxedcodepolytopeN$ and rotation matrix $\rotmat$. From now on, we use $\rotmat_j$ to denote the rotation matrix $\rotmat(q,\bfh_j)$, where $\bfh_j$ is the non-zero sub-vector of the $j$-th check of the code. Then, \LPDFlanaganRelax can be rewritten as
\begin{equation}
\label{eq.lpdecoding_new_new}
\begin{split}
\min \; \bfgamma^T \bfx \quad
\st \; \bfD_j^{-1}\bfP_j \bfx \in \relaxedcodepolytopeN, \text{ for all } j\in \mcJ,
\end{split}
\end{equation}

Note that $\relaxedcodepolytopeN$ is defined as the intersection of simplexes (condition $(a)$ and $(b)$ in Definition~\ref{def.relaxed_code_polytope}) and parity polytopes (condition $(c)$ in Definition~\ref{def.relaxed_code_polytope}). For condition $(c)$ we define a select-and-add matrix $\bfT_\kk$ using the following procedure. We first let $\K_{\kk}$, $\kk \in [2^m-1]$, be all non-empty subsets of $[m]$. We then use $\bfT_{\kk}$ to denote the matrix that selects the entries in $\mcBB(\K_{\kk},1)$ and adds them to form the vector $\bfg^{\K_{\kk}}$. As a result, we can rewrite~\eqref{eq.lpdecoding_new_new} as
\begin{equation}
\begin{split}
\min \quad & \bfgamma^T \bfx \quad \\
\st \quad  & \bfT_k \rotmat_j^{-1} \bfP_j \bfx  \in \PP_\checkdegree, \text{ for all } j\in \mcJ \text{ and }\kk \in [2^m-1]\\
\quad &  \bfx_i \in \PS_{q-1} \text{ for all } i\in \mcI.
\end{split}
\end{equation}
We introduce replicas $\bfz_{j,k}$ and $\bfs_i$. We let $\bfz_{j,k} = \bfT_k \rotmat_j^{-1} \bfP_j \bfx$ and $\bfs_i = \bfx_i$. However, note that we draw a distinction between $\bfz$ and $\bfs$ only for notation purposes. Both $\bfz$ and $\bfs$ serve the same purpose (i.e., replicas) in the ADMM algorithm.
\begin{equation}
\begin{split}
\min \quad & \bfgamma^T \bfx \quad \\
\st \quad &\text{for all } j\in\mcJ, \kk \in [2^m-1] \text{ and }i\in\mcI,\\
\quad  & \bfz_{j,\kk} = \bfT_k \rotmat_j^{-1} \bfP_j \bfx, \\
\quad & \bfz_{j,\kk} \in \PP_\checkdegree, \\
\quad & \bfs_i  = \bfx_i,\\
\quad &  \bfs_i \in \PS_{q-1}.
\end{split}
\end{equation}
By defining $\bfZ_{j,\kk} := \bfT_\kk \rotmat_j^{-1} \bfP_j$, we can write the augmented Lagrangian as
\begin{equation}
\label{eq.ADMM_with_mu}
\begin{aligned}
\Larg_{\mu} (\bfx, \bfz, \bfs, \bflambda, \bfeta) = &\bfgamma^T \bfx +\sum_{j,\kk}\bflambda^T_{j,\kk} (\bfZ_{j,\kk} \bfx - \bfz_{j,\kk}) + \frac{\mu}{2}\sum_{j, \kk}\|\bfZ_{j,\kk} \bfx - \bfz_{j,\kk}\|_2^2\\
& +\sum_{i}\bfeta^T_{i} (\bfS_{i} \bfx - \bfs_{i}) + \frac{\mu}{2}\sum_{i}\|\bfS_{i} \bfx - \bfs_{i}\|_2^2,
\end{aligned}
\end{equation}
where $\bfS_{i}$ selects the sub-vector $\bfx_i$ from $\bfx$. The ADMM updates for this problem are 
\begin{align*}
\bfx\text{-update: }\bfx^{*} &= \argmin_{\bfx} \Larg_{\penpara}(\bfx,\bfz,\bfs,\bflambda, \bfeta), \\
\bfz\text{-update: }\bfz^{*} &= \argmin_{\bfz \in \PP_\checkdegree} \Larg_{\penpara}(\bfx^{*},\bfz,\bfs,\bflambda, \bfeta), \\
\bfs\text{-update: }\bfs^{*} &= \argmin_{\bfs \in \PS_{q-1}} \Larg_{\penpara}(\bfx^{*},\bfz^{*},\bfs,\bflambda, \bfeta),\\
\bflambda\text{-update: }\bflambda_j^{*} &=\bflambda_j^k + \penpara\left(\bfZ_{j,k}\bfx^{*}-\bfz^{*}_{j,k}\right),\\ 
\bfeta\text{-update: }\bfeta_j^{*} &=\bfeta_j^k + \penpara\left(\bfS_{i}\bfx^{*}-\bfs^{*}_{i}\right). 
\end{align*}
We make the following remarks. First, the $\bfx$-update is an unconstrained optimization problem, which is different from the case in Section~\ref{subsec.genericADMM}. Second, we separate the $\bfz$-update and the $\bfs$-update because the two types of replicas are in two different polytopes ($\PP_\checkdegree$ and $\PS_{q-1}$ respectively).
\subsubsection{$\bfx$-update}
By taking the gradient of $\Larg_{\mu} (\bfx, \bfz, \bfs, \bflambda)$ and setting it to the all-zeros vector we obtain the following $\bfx$-update
\begin{equation*}
\bfx^* = (\bfZ + \bfI)^{-1} \left(\sum_{j,\kk} \bfZ_{j,\kk}^T ( \bfz_{j,\kk} - \frac{\bflambda_{j,\kk}}{\mu}) +\sum_{i} \bfS_{i}^T ( \bfs_{i} - \frac{\bfeta_i}{\mu})- \frac{\bfgamma}{\mu} \right),
\end{equation*}
where $\bfZ = \sum_{j,\kk}\bfZ_{j,\kk}^T\bfZ_{j,\kk}$ and the identity matrix $\bfI$ comes from the fact that $ \sum_i \bfS_i^T \bfS_i = \bfI$. 
Note that $\bfT_{\kk}^T$ has at most one $1$ in each row due to constraints defined in Definition~\ref{def.relaxed_code_polytope}. Since the permutation matrix $\rotmat_j^{-1}$ and the selection matrix $\bfP_j$ all have at most one $1$ in each column, $\bfZ_{j,\kk}^T$ simply selects the entries in $\bfz_{j,\kk}$ that correspond to some entries in $\bfx$\footnote{Recall that $\bfZ_{j,\kk}$ maps $\bfx$ to the replica $\bfz_{j,\kk}$. In addition, there is at most one $1$ per column in $\bfZ_{j,\kk}$. This implies that $\bfZ_{j,\kk}^T$ has at most one $1$ per row. Thus, $\bfZ_{j,\kk}^T$ simply selects entries from $\bfz_{j,\kk}$.}. We use the notation $\bfz_{j,\kk}^{(i)}$ to represent the length-$(q-1)$ sub-vector of $\bfz_{j,\kk}$ that corresponds to $\bfx_i$. Therefore, we can rewrite $\bfx$-update as
\begin{align*}
\bfx^* =   (\bfZ + \bfI)^{-1} \bft,
\end{align*}
where $\bft = (\bft_1,\dots,\bft_{\blocklength})$ and 
\begin{align}
\label{eq.t_vector}
\bft_i = \sum_{(j,\kk)\in \Nev(i)} \left(\bfz_{j,\kk}^{(i)} - \frac{\bflambda_{j,\kk}^{(i)} }{\mu}\right)  + \bfs_i - \frac{\bfeta_i}{\mu} - \frac{\bfgamma_i}{\mu}.
\end{align}
Note that the matrix $(\bfZ + \bfI)^{-1}$ is fixed by the code. Therefore one needs to calculate it only once per code. However, a naive calculation of $(\bfZ + \bfI)^{-1} \bft$ has complexity $O(\blocklength^2 (2^{m}-1)^2)$ which can be prohibitive for large $\blocklength$. We therefore introduce a method to calculate this product in linear complexity $O(\blocklength (2^m-1))$.
\begin{lemma}
\label{lemma.x_update_algorithm}
$(\bfZ + \bfI)^{-1} \bft$ can be calculated in linear complexity $O(\blocklength (2^m-1))$.
\end{lemma}
\begin{proof}
See Appendix~\ref{appendix.x_update_algorithm}.
\end{proof}

\subsubsection{Other ADMM updates}
The $\bfz$- and $\bfs$-updates are the same as~\eqref{eq.zupdate_projection} except that we now perform the respective projections onto $\PP_{\checkdegree}$ and $\PS_{2^m-1}$. 
\small
\begin{align*}
\bfs_{i}^* = &\argmin_{\bfs_{i} \in \PS_{2^m - 1}} \|\bfu_i - \bfs_i \|_2^2  \text{ for } i \in \mcI\\
\bfz_{j,\kk}^* = &\argmin_{\bfz_{j,\kk} \in \PP_{\checkdegree}} \|\bfv_{j,\kk} - \bfz_{j,\kk} \|_2^2  \text{ for } j \in \mcJ \text{ and }\kk \in [2^m\!-\!1],
\end{align*}
\normalsize
where $\bfu_i = \bfS_i\bfx + \bfeta_i/\mu$ and $\bfv_{j,\kk} = \bfZ_{j,\kk} \bfx + \bflambda_{j,\kk}/\mu$. Both projection operations have complexity that is linear in $q (= 2^m)$ and $\checkdegree$ respectively (cf.~\cite{zhang2013large} and~\cite{duchi2008efficient}). 

\begin{proposition}
\label{proposition.computationalcomplexity}
Let $\tilde{\bfx}$ be the solution of the LP decoding problem~\eqref{eq.lpdecoding_relax}. For any $\eta > 0$, Algorithm~\ref{algorithm.ADMMLP_noprojection} will, in $O(\blocklength q^2)$ time, determine a vector $\hat{\bfx}$ that satisfies the constraints in~\eqref{eq.lpdecoding_relax} and that also satisfies the following bound:
\begin{equation*}
\label{eq.admm_tolerance}
\embedllrvec(\bfy)^T \hat{\bfx} - \embedllrvec(\bfy)^T \tilde{\bfx} < q\blocklength \checkdegree \eta,
\end{equation*}
where $\bfy$ is the channel output, $q = 2^m$ is the field size and $\checkdegree$ is the check degree.
\end{proposition}
\begin{proof}
This proposition is similar to Proposition 1 in~\cite{barman2013decomposition}. We omit the details but point out that ADMM takes $O(1)$ iterations to converge to a point with $q \blocklength \checknumber \eta$ gap to optimal. For each iteration, the complexity of the $\bfx$-update is $O(\blocklength q^2)$. The complexity of each $(\bfz,\bfs)$-update is $O(\blocklength q + \checknumber (q-1) \checkdegree)$. Since $\checkdegree$ does not scale, we deduce that the complexity of each iteration is $O(\blocklength q^2)$. Therefore due to the $O(1)$ requirement on the number of iterations, the overall complexity is $O(\blocklength q^2)$.
\end{proof}

We summarize ADMM LP decoding in Algorithm~\ref{algorithm.ADMMLP_noprojection}.

\begin{algorithm}
\caption{ADMM LP decoding of LDPC codes in $\GF_{2^m}$}
\label{algorithm.ADMMLP_noprojection}
\begin{algorithmic}[1]

\STATE Construct the $\checkdegree \times (2^m-1)\blocklength $ matrices $\bfZ_{j,\kk}$ for all $j \in \mcJ$ and $\kk \in [2^m-1]$ based on the factor graph of embeddings.
\STATE Construct the $(2^m - 1) \times (2^m-1)\blocklength $ matrices $\bfS_i$ for all $i \in \mcI$ based on the at-most-one-on constraints.

\STATE Initialize all entries of $\bflambda_{j,\kk}$ and $\bfeta_i$ to $0$. Initialize all entries of $\bfz_{j,\kk}$ and $\bfs_{i}$ to $\frac{1}{2^m}$. Initialize iterate $\delta = 0$. To simplify the notation, we drop iterate $\delta$ except when determining the stopping criteria.

\REPEAT

\FORALL{ $ i = 1,\dots,\blocklength$}
\STATE Compute vector $\bft_i$ using equation~\eqref{eq.t_vector}. Store $\|\bft_i\|_1$.

\STATE Compute the variables $a$ and $b$ per Lemma~\ref{lemma.invertZplusI}.

\STATE Update $\bfx_i \leftarrow (a-b)\bft_i + b\|\bft_i\|_1$.
\ENDFOR

\FORALL{ $ j = 1,\dots,\checknumber$ and $\kk = 1,\dots,2^{m}-1$}

\STATE Set $ \bfv_{j,\kk}  = \bfZ_{j,\kk} \bfx + \bflambda_{j,\kk} /\penpara $.

\STATE Update $ \bfz_{j,\kk} \leftarrow \Proj_{\PP_{\checkdegree}} (\bfv_{j,\kk}) $ where
$\Proj_{\PP_{\checkdegree}} (\cdot)$ is a projection onto the parity polytope of dimension $\checkdegree$.

\STATE Update $\bflambda_{j,\kk} \leftarrow \bflambda_{j,\kk} + \mu \left( \bfZ_{j,\kk} \bfx - \bfz_{j,\kk}\right) $.
\ENDFOR

\FORALL{ $ i = 1,\dots,\blocklength$}

\STATE Set $ \bfu_i  = \bfS_i \bfx + \bfeta_i /\mu $.
 
\STATE Update $ \bfs_i \leftarrow \Proj_{\PS_{q-1}} (\bfu_i) $ where
$\Proj_{\PS_{q-1}} (\cdot)$ is a projection onto $\PS_{q-1}$.

\STATE Update $\bfeta_i \leftarrow \bfeta_i + \mu \left( \bfS_i \bfx - \bfs_i\right) $.
\ENDFOR
\STATE $ \delta \! \leftarrow \! \delta+1$.
\UNTIL{ $ \sum_i { \| \bfS_i \bfx^\delta - \bfs^\delta_i \|^2_{2} } 
+\sum_{j,\kk} { \| \bfZ_{j,\kk} \bfx^\delta - \bfz^\delta_{j,\kk} \|^2_{2} } < \epsilon^2 ((2^{m} - 1)\blocklength + \checknumber(2^m-1)\checkdegree) $ \\and  $ \sum_i { \|  \bfs^\delta_i - \bfs^{\delta - 1}_i \|^2_{2} } 
+\sum_{j,\kk} { \| \bfz^\delta_{j,\kk} - \bfz^{\delta-1}_{j,\kk} \|^2_{2} } < \epsilon^2 ((2^{m} - 1)\blocklength + \checknumber(2^m-1)\checkdegree) $}\\
{\bf return} $\bfx$.
\end{algorithmic}
\end{algorithm}
\subsubsection{Early termination and over-relaxation}
\label{subsubsection.early_term_and_overrel}
In practice, ADMM can be terminated whenever a codeword is found. Doing so often reduces the number of iterations needed to decode. We observe empirically that early termination does affect the WER for some very short block length codes. However, it does not create observable effects for the codes we study in Section~\ref{section.numerical}. We use the vectors $\bfs_i$ for all $i \in \mcI$ to determine whether or not the corresponding non-binary vector is a codeword. The decision process is as follows. For each $i\in\mcI$, we first obtain the vector $\hat{\bfx} := (\hat{x}_0,\hat{x}_1,\dots,\hat{x}_{2^m-1})$ by letting $\hat{x}_0 = 1 - \|\bfs_i\|$ and $\hat{x}_k = s_{i,k}$ for all $k \in [2^m-1]$. Then, the $i$-th non-binary symbol is determined by $\hat{c}_i = \argmax_k \hat{x}_k$. After we have obtained the length-$\blocklength$ non-binary vector $\hat{\bfc}$, we terminate ADMM if $\bfH \hat{\bfc} = 0$ in $\GF_{2^m}$ where $\bfH$ is the parity-check matrix of the code. 

Another useful technique in practice is over-relaxation, which is already used in~\cite{barman2013decomposition}. Readers are referred to~\cite{boyd2010distributed} and references therein for details on over-relaxation. We observe that over-relaxation is effective for Algorithm~\ref{algorithm.ADMMLP_noprojection}. Experiments surrounding the choices of over-relaxation parameter $\rho$ are presented in Section~\ref{section.numerical}.
\section{ADMM penalized decoding of non-binary codes}
\label{section.penalized}
In this section, we introduce an ADMM penalized decoder for non-binary codes. This decoder is analogous to the penalized decoder for binary codes introduced in~\cite{liu2014the}. It is shown in~\cite{liu2014the} that adding a non-convex penalty term to the LP decoding objective can improve the low SNR performance of LP decoding. The penalty term penalizes fractional solutions and can improve the performance because correct solutions should be integral. Herein, we extend this idea to embeddings of non-binary symbols. 

One important characteristic of linear codes is that, at least when codewords are transmitted over a symmetric channel, the probability of decoding failure should be independent of the codeword that is transmitted. To get this property to hold for the non-binary penalized decoder, herein, we use \nameTheCW. This desired ``codeword-independent'' property follows because all non-binary symbols are embedded symmetrically with respect to each other. As a reminder, for \nameF, symbol $0$ is treated differently from others symbol in the field. We discuss this issue in details in Section~\ref{subsection.discussions}.
\subsection{ADMM penalized decoding algorithm}
Formally, we consider the following decoding problem:
\begin{equation}
\label{eq.penalized}
\begin{split}
\min \;  &  \bfgamma'^T \bfx  -\alpha \sum_i \|\bfx_i - \bfr\|_2^2\\
\st  \; & \bfP'_j \bfx \in \cwrelaxedcodepolytope_j, \forall j\in \mcJ,
\end{split}
\end{equation}
where $\cwrelaxedcodepolytope_j$ is the relaxed polytope under \nameTheCW (Definition~\ref{def.cw_code_polytopes}), $\bfgamma' := \cwembedllrvec(\bfy)$ and $$\bfr = \left(\frac{1}{q},\frac{1}{q},\dots,\frac{1}{q}\right).$$

We pick the $\ell_2$ norm as the penalty in~\eqref{eq.penalized} because it yields simple ADMM update rules and because it yields good empirical performance (cf.~Section~\ref{section.numerical}). For both the $\bfx$-update and the $\bfz$-update, we complete the square and then perform minimizations. As a result, the $\bfx$-update can be expressed as a projection onto $\SS_q$; and the $\bfz$-update can be expressed as a projection onto $\cwrelaxedcodepolytope_j$. In particular, the $\bfx$-update rule is
\begin{align*}
\bfx_i^{*} \! = \! \Proj_{\SS_q} \! \left[\frac{1}{d_i \! - \! \frac{2\alpha}{\mu}}\left(\sum_{j\in\Nev(i)} \left( \bfz_j^{(i)} \! - \! \frac{\bflambda_j^{(i)}}{\mu}\right) \! - \!  \frac{\bfgamma'_i}{\mu} \! - \! \frac{2\alpha \bfr}{\mu}\right)\right].
\end{align*}

The $\bfz$-update can be written as a projection onto $\cwrelaxedcodepolytope$. For this projection operation, we use the rotate-and-project technique introduced in Section~\ref{subsec.rotation} (see also~\cite{liu2014admm} for more details). The ADMM penalized decoder is summarized in Algorithm~\ref{algorithm.penalized_decoder}.

\begin{algorithm}
\caption{ADMM penalized decoding. \textbf{Input}: Received vector $\bfy \in \Sigma^N$.
\textbf{Output}: Decoded vector $\bfx$.}
\label{algorithm.penalized_decoder}
\begin{algorithmic}[1]

\STATE Construct the $qd_j \times qN $ selection matrix $\bfP_j'$ for all $j \in
\mathcal{J}$ based on the parity-check matrix $\bfH$.

\STATE Construct the log-likelihood ratio $\cwembedllrvec(\bfy)$. 

\STATE For all $j \in \mathcal{J}$, initialize all entries of $\bflambda_j$ to $0$ and initialize all entries of $\bfz_j$ to $0.5$. Initialize iterate $\delta = 0$. To simplify the notation, we drop iterate $\delta$ except when determining the stopping criteria.

\REPEAT

\FORALL{ $ i \in \mathcal{I} $}
\STATE \label{algstep.xupdate1} Update $\bfx_i$ by \\
\footnotesize{
$
\bfx_i \! = \! \Proj_{\SS_q} \! \left[\frac{1}{d_i \! - \! \frac{2\alpha}{\mu}}\left(\sum_{j\in\Nev(i)} \left( \bfz_j^{(i)} \! - \! \frac{\bflambda_j^{(i)}}{\mu}\right) \! - \! \frac{\bfgamma'_i}{\mu} \! - \! \frac{2\alpha \bfr}{\mu}\right)\right]
$}
\ENDFOR
\FORALL{ $ j \in \mathcal{J} $}

\STATE Set $ \bfv_j  \leftarrow \bfP_j \bfx + \bflambda_j /\mu $.

\STATE Update $ \bfz_j \leftarrow \Proj_{\cwrelaxedcodepolytope_j} (\bfv_j) $ where
$\Proj_{\cwrelaxedcodepolytope_j} (\cdot)$ is a projection onto the relaxed code polytope defined by the $j$-th check.

\STATE Update $\bflambda_j \leftarrow \bflambda_j + \mu \left( \bfP'_j \bfx - \bfz_j\right) $.
\ENDFOR
\STATE $ \delta \! \leftarrow \! \delta+1$.
\UNTIL{ $ \sum_j { \| \bfP_j' \bfx^\delta - \bfz^\delta_j \|^2_{2} } < \epsilon^2 qMd_j $ \\and  $\sum_j { \| \bfz^{\delta}_j - \bfz^{\delta-1}_j \|^2_{2} } < \epsilon^2 qMd_j$}\\
{\bf return} $\bfx$.
\end{algorithmic}
\end{algorithm}

\begin{theorem}
\label{theorem.penalized_decoder_allzero}
Under the channel symmetry condition in the sense of~\cite{flanagan2009linearprogramming}, the probability that Algorithm~\ref{algorithm.penalized_decoder} fails
is independent of the codeword that was transmitted.
\end{theorem}
\begin{proof}
See Appendix~\ref{proof.penalized_decoder_allzero}
\end{proof}

\subsection{Discussions}
\label{subsection.discussions}
The penalized decoder described above requires a sub-routine that projects onto the code polytope. The ADMM projection technique introduced in~\cite{liu2014admm} scales linearly with $dq^2$, where $d$ is the degree of the check node and $q$ is the field size. However, this projection technique is iterative and inaccurate, because it accepts an error tolerance $\epsilon_p$. As a result, the provable number of iterations scales linearly with $\frac{1}{\epsilon_p}$ due to the linear convergence rate of ADMM (cf.~\cite{wang:12online}). This means that each iteration of the ADMM penalized decoder scales linearly with $1/\epsilon_p$. On the contrary, each iteration of Algorithm~\ref{algorithm.ADMMLP_noprojection} does not depend on such error tolerance. As a result, the computational complexity of the penalized decoder is much higher than the ADMM LP decoder in Algorithm~\ref{algorithm.ADMMLP_noprojection}. We observe empirically that in our implementation of ADMM penalized decoding is around $20$ times slower than the ADMM LP decoder in Algorithm~\ref{algorithm.ADMMLP_noprojection}\footnote{As measured in execution time (sec).}. 

We note that Algorithm~\ref{algorithm.ADMMLP_noprojection}, in fact, can be used to try to solve the penalized objective~\eqref{eq.penalized}. However, we observe empirically that the resulting decoder does not have the codeword symmetry property in the sense of Theorem~\ref{theorem.penalized_decoder_allzero}. We briefly discuss why the proof technique of Theorem~\ref{theorem.penalized_decoder_allzero} cannot be applied to this decoder. First, note that the penalized decoder tries to solve a non-convex program. As a result, the output of the penalized decoder cannot be determined solely by the optimization problem. It depends on the ADMM update rules (the $\bfx$-, $\bfz$- and $\bflambda$-update), which are determined by the way the optimization problem is stated\footnote{That is,~\eqref{eq.penalized} can be rewritten into other equivalent forms. Each form may result in a different ADMM algorithm.}. In the proof of Theorem~\ref{theorem.penalized_decoder_allzero}, we show that the decoding process when decoding the all-zeros codeword is ``symmetric'' with respect to any other codeword's decoding process in the sense of Lemma~\ref{lemma.quiviter} (see rigorous discussions in Appendix~\ref{proof.penalized_decoder_allzero}). This symmetry behavior is due to the symmetric structure of $\cwrelaxedcodepolytope$ (Lemma~\ref{lemma.relative_code_polytope}). If we were to use Algorithm~\ref{algorithm.ADMMLP_noprojection} to try to solve~\eqref{eq.penalized}, the polytopes used to describe the constraint set are parity polytopes and simplexes, not $\cwrelaxedcodepolytope$. As a result, we cannot have the symmetry behavior for parity polytopes and simplexes in the sense of Lemma~\ref{lemma.relative_code_polytope}. 

On the other hand, we note that applying Algorithm~\ref{algorithm.ADMMLP_noprojection} to~\eqref{eq.penalized} results in a significant improvement in terms of decoding efficiency when compared to Algorithm~\ref{algorithm.penalized_decoder}. Furthermore, we observe empirically that this technique can achieve a much reduced error rate when compared to LP decoding (\LPDFlanaganRelax and \LPDCWRelax) for the all-zeros codeword. Open questions include whether all codewords experience some improvement and how much improvement each codeword receives. These questions are non-trivial because of the non-convexity of the penalized objective. 
\section{Numerical results and discussions}
\label{section.numerical}
In this section we present numerical results for ADMM LP decoding (Algorithm~\ref{algorithm.ADMMLP_noprojection}) and ADMM penalized LP decoding (Algorithm~\ref{algorithm.penalized_decoder}). First, we show the error rate performance of the two decoders as a function of SNR. In addition, we compare our decoders to the low-complexity LP (LCLP) decoding technique proposed by Punekar \emph{et al.} in~\cite{punekar2012low}. Next, in Section~\ref{subsection.parameter_lp}, we focus on choosing a good set of parameters for Algorithm~\ref{algorithm.ADMMLP_noprojection}. Finally, we show how the penalty coefficient $\alpha$ (cf.~\eqref{eq.penalized}) affects the WER performance of ADMM penalized decoding.

\subsection{Performance of the proposed decoders}
\label{subsection.snr_performance}
In this subsection we simulate three codes and demonstrate their error rate performance as a function of SNR. The first code we simulate is derived from a binary length-$2048$ progressive edge growth (PEG) code that can be obtained from~\cite{mackaydatabase}. The other two codes are derived from two binary Tanner codes obtained from\cite{tanner2001class}. We select these codes for the following reasons. First, we select codes that are of different block lengths ($2048$, $1055$ and $755$) and in different fields ($\GF_{2^2}$ and $\GF_{2^3}$). Second, the two Tanner codes have been studied in~\cite{punekar2012low}. Therefore we can make direct comparisons with the results therein. Last but not least, the parity-check matrices of these codes are easy to obtain (e.g., from~\cite{mackaydatabase}). Our intent is to make our simulations using these codes easy to repeat and thus to compare to.

Figure~\ref{fig.sim_peg2048_wer_ser} plots the word-error-rate (WER) and symbol-error-rate (SER) of a length-$2048$ LDPC code when decoded using ADMM LP decoding and ADMM penalized decoding. This code is derived from the PEG $[2048, 1018]$ binary LDPC code obtained from~\cite{mackaydatabase}. We use the same parity-check matrix as the original PEG code except that we let each non-zero check value be $1 (= \xi^0) \in \GF_{2^2}$. The code symbols are modulated using quaternary phase-shift keying (QPSK) and transmitted over an AWGN channel. Denote by $(x,y)$ the in-phase and quadrature components. We modulate the symbols in the following way: $0 \mapsto (1, 0)$, $\xi^0 \mapsto (0, 1)$, $\xi^1 \mapsto (-1, 0)$ and $\xi^2 \mapsto (0, -1)$. We use energy per information symbol ($E_s/N_0$) as our unit of SNR. We show explicitly how we calculate $E_s/N_0$ in Appendix~\ref{appendix.esn0}. In this simulation, we decode using ADMM LP decoding (Algorithm~\ref{algorithm.ADMMLP_noprojection}) and ADMM penalized decoding (Algorithm~\ref{algorithm.penalized_decoder}). We note that both algorithms require many parameters. We show how to choose parameters in Sections~\ref{subsection.parameter_lp} and~\ref{subsection.penalized}. For ADMM LP decoding we use the following parameter settings: the ``step size'' of ADMM $\mu = 2$ (cf.~\eqref{eq.ADMM_with_mu}); the maximum number of iterations $T_{\max} = 200$; the ending tolerance $\epsilon = 10^{-5}$ (cf. Algorithm~\ref{algorithm.ADMMLP_noprojection}), and the over-relaxation parameter $\rho = 1.9$ (cf. Section~\ref{subsubsection.early_term_and_overrel}). For ADMM penalized decoding, we let $\mu = 4$, $T_{\max} = 200$, $\rho = 1.5$, $\epsilon = 10^{-5}$, and $\alpha= 0.8$ (cf.~\eqref{eq.penalized}). For each data point, we collect more than $100$ word errors.

We make the following observations. First, both decoding algorithms display a ``waterfall'' behavior in terms of error rates. However, penalized decoding initiates the waterfall at a much lower SNR (about $0.7$dB lower in this example). This shows that ADMM penalized decoding significantly improves LP decoding at low SNRs. Second, ADMM penalized decoding displays an error-floor at WER $10^{-4}$ (SER $10^{-6}$). However, LP decoding does not display an error-floor for WER above $10^{-5}$ (SER above $10^{-8}$). Both observations are consistent with binary LP decoding and binary penalized decoding~\cite{liu2014the}. We note that unlike the case with binary decoders where the binary ADMM penalized decoder outperforms the binary ADMM LP decoding in terms of both the number of iterations and execution time, the average decoding time of Algorithm~\ref{algorithm.penalized_decoder} using our implementation is much longer than that of Algorithm~\ref{algorithm.ADMMLP_noprojection} due to the reasons discussed in Section~\ref{section.penalized}. Although Algorithm~\ref{algorithm.penalized_decoder} is competitive in terms of WER performance, it is not competitive in terms of execution time. 

\begin{figure}[!htbp]
\psfrag{&esn0}{\scriptsize{$E_s/N_0$ (dB)}}
\psfrag{&wer}{\hspace{-0.2cm}\scriptsize{Error rate}}
\psfrag{&admmlpdecodingwer}{\scriptsize{ADMM LP, WER}}
\psfrag{&admmlpdecodingser}{\scriptsize{ADMM LP, SER}}
\psfrag{&admmpenanlizeddecodingwerpadp}{\scriptsize{ADMM penalized decoder, WER}}
\psfrag{&admmpenanlizeddecodingser}{\scriptsize{ADMM penalized decoder, SER}}

    \begin{center}
    \includegraphics[width=21pc]{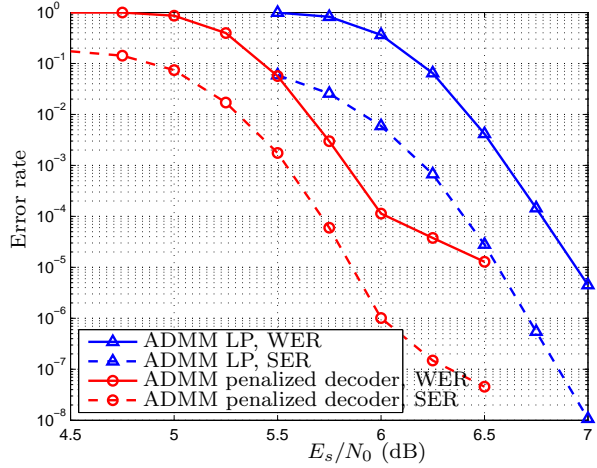}
    \end{center}
    \caption{ADMM LP decoding and ADMM penalized decoding error rates plotted as a function of SNR for the $[2048,1018]$ PEG code in $\GF_{2^2}$.}
    \label{fig.sim_peg2048_wer_ser}
\end{figure}

The second code we simulate is the length-$1055$ code used in~\cite{punekar2012low}. This code is derived from the Tanner $[1055,424]$ binary LDPC code introduced by Tanner in~\cite{tanner2001class}. Similar to the previous case, each binary non-zero entry is replaced by the value $1 \in \GF_{2^2}$. We keep all other simulation settings the same as the previous case, i.e., we use the same AWGN channel with QPSK modulation and the same parameter settings for Algorithm~\ref{algorithm.ADMMLP_noprojection}. For the penalized decoder, we change the parameter settings. We use the following settings: $\mu = 4$, $\rho = 1.5$, $T_{\max} = 100$, $\epsilon = 10^{-5}$, and $\alpha = 0.6$. We first plot the WER performance of the code in Figure~\ref{fig.sim_tanner1055_wer}. The performance of LCLP and the sum-product algorithm (SPA) decoders are plotted for comparison. The data for these decoders is obtained from~\cite{punekar2012low}.

We make the following observations. First, we observe that ADMM LP decoding outperforms both SPA and LCLP in terms of WER. ADMM LP decoding has a $0.6$dB SNR gain when compared to the LCLP algorithm of~\cite{punekar2012low}. This result is interesting because ADMM LP decoding uses the relaxed code polytope $\relaxedcodepolytope$. Thus, one might expect that in general it would perform worse than Flanagan's LP decoding algorithm. However, we recall that we validated numerically that $\relaxedcodepolytope = \tightcodepolytope$ for $\GF_{2^2}$ (cf. Appendix~\ref{appendix.conjecture}). Figure~\ref{fig.sim_tanner1055_wer} suggests that LCLP may have a $0.6$dB SNR loss due to the approximations made in that algorithm (see~\cite{punekar2012low} for details).  Second, the penalized decoder improves the WER performance by $0.4$dB. This effect is consistent with Figure~\ref{fig.sim_peg2048_wer_ser}. Third, we note that none of the decoders  used in Figure~\ref{fig.sim_tanner1055_wer} displays an error-floor.

\begin{figure}[!htbp]
\psfrag{&esn0}{\scriptsize{$E_s/N_0$ (dB)}}
\psfrag{&wer}{\hspace{-0.5cm}\scriptsize{Error rate (WER)}}
\psfrag{&lclp}{\scriptsize{LCLP,~\cite{punekar2012low}}}
\psfrag{&sumproductBP}{\scriptsize{Sum-product,~\cite{punekar2012low}}}
\psfrag{&ADMMLPdecoding}{\scriptsize{ADMM LP decoding}}
\psfrag{&ADMMpenalizeddecodingpad}{\scriptsize{ADMM penalized decoding}}
    \begin{center}
    \includegraphics[width=21pc]{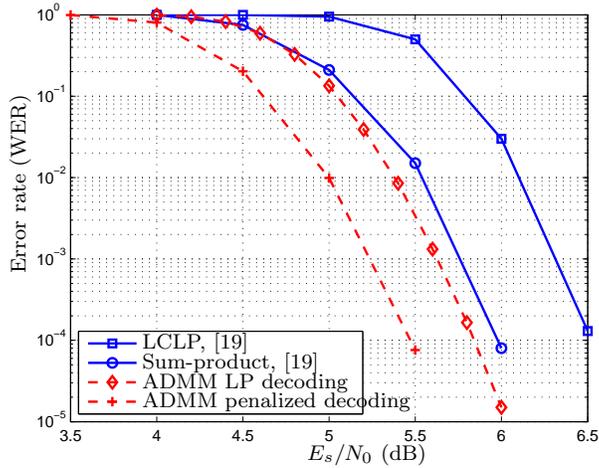}
    \end{center}
    \caption{WER plotted as a function of SNR for the Tanner $[1055,424]$ code in $\GF_{2^2}$.}
    \label{fig.sim_tanner1055_wer}
\end{figure}
In Figure~\ref{fig.sim_tanner1055_iter}, we plot the average number of iterations for erroneous decoding events, all decoding events, and correct decoding events. We observe that for correct decoding events, the average number of iterations is less than $100$ at all SNRs. This indicates that we can lower $T_{\max}$ without greatly impacting WER. In fact, we observe that there is less than a $0.05$dB loss if we use $T_{\max} = 100$ (data not shown). 

In Figure~\ref{fig.sim_tanner1055_time}, we plot execution time statistics for the decoding events. In these simulations, the decoder is implemented using C++ and data is collected on a 3.10GHz Intel(R) Core(TM) i5-2400 CPU. We make the following observations: First, the average decoding time for correctly decoded events decreases as the SNR increases. However, the average decoding time for erroneous decodings does not vary significantly with SNR. Second, the average time per decoding  decreases rapidly at low SNRs. This is due to the waterfall behavior of the error rate. Third, we note that ADMM LP decoding (Algorithm~\ref{algorithm.ADMMLP_noprojection}) is a much faster algorithm in terms of execution time than ADMM penalized decoding (Algorithm~\ref{algorithm.penalized_decoder}). As an example, the average decoding time for ADMM LP decoding at $E_s/N_0 = 5$dB is $0.033$s. In contrast, the average decoding time for ADMM penalized decoding at $E_s/N_0 = 5$dB is $0.61$s (data not shown), which is about 15-20 times slower than Algorithm~\ref{algorithm.ADMMLP_noprojection}.
 
\begin{figure}[!htbp]
\psfrag{&esn0}{\scriptsize{$E_s/N_0$ (dB)}}
\psfrag{&iter}{\hspace{-0.7cm}\scriptsize{Number of iterations}}
\psfrag{&ErroneousDecodingEvents-}{\scriptsize{Erroneous decoding events}}
\psfrag{&Average}{\scriptsize{Average}}
\psfrag{&CorrectDecodingEvents}{\scriptsize{Correct decoding events}}
    \begin{center}
    \includegraphics[width=21pc]{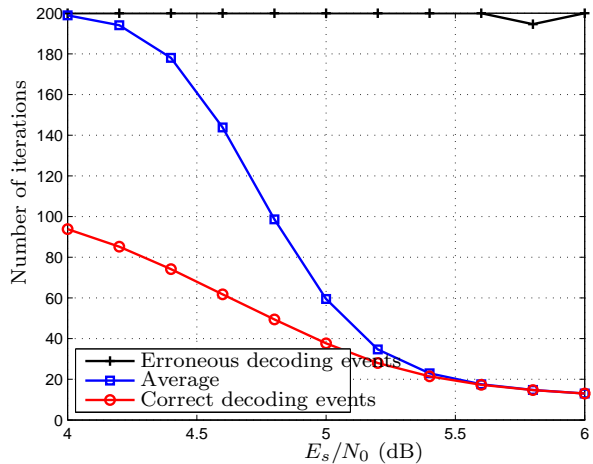}
    \end{center}
    \caption{Number of iterations of Algorithm~\ref{algorithm.ADMMLP_noprojection} plotted as a function of SNR for the Tanner $[1055,424]$ code in $\GF_{2^2}$.}
    \label{fig.sim_tanner1055_iter}
\end{figure}

\begin{figure}[!htbp]
\psfrag{&esn0}{\scriptsize{$E_s/N_0$ (dB)}}
\psfrag{&exet}{\hspace{-0.5cm}\scriptsize{Execution time (s)}}
\psfrag{&ErroneousDecodingEvents-}{\scriptsize{Erroneous decoding events}}
\psfrag{&Average}{\scriptsize{Average}}
\psfrag{&CorrectDecodingEvents}{\scriptsize{Correct decoding events}}
    \begin{center}
    \includegraphics[width=21pc]{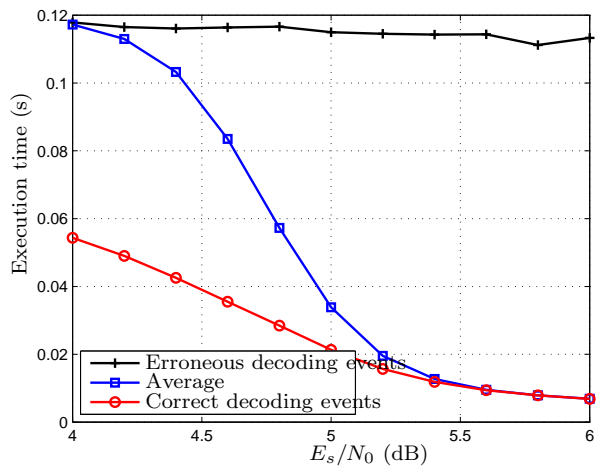}
    \end{center}
    \caption{Software implementation execution time of Algorithm~\ref{algorithm.ADMMLP_noprojection} plotted as a function of SNR for the Tanner $[1055,424]$ code in $\GF_{2^2}$.}
    \label{fig.sim_tanner1055_time}
\end{figure}

In Figure~\ref{fig.sim_tanner755_wer}, we show results for the length-$755$ code also studied in~\cite{punekar2012low}. This code has symbols in $\GF_{2^3}$ and is derived from the Tanner $[755,334]$ binary LDPC code~\cite{tanner2001class}. We use the same non-binary parity-check matrix and the same $8$-PSK modulation method as in~\cite{punekar2012low}. We make the following remarks. First, ADMM LP decoding performs worse than LCLP in our simulations. This is due to the fact that ADMM LP decoding uses the relaxed code polytope. Unlike in the $\GF_{2^2}$ case, relaxing the polytope using Definition~\ref{def.relaxed_code_polytope} induces a loss in terms of SNR performance. We observe that ADMM LP decoding is about $0.3$dB worse than LCLP. Second, and similar to what we saw in the simulations for the previous two codes, the penalized decoder improves the low-SNR performance of LP decoding by around $1.5$dB for $E_s/N_0 \leq 9$dB. However, the WER curve displays a significant error-floor. We observe that the optimal value for the penalty coefficient $\alpha$ decreases as SNR increases (data not shown). At $E_s/N_0 = 10$dB, adding a positive penalty does not lead to an observable improvement in WER. 
Third, in comparison to Figure~\ref{fig.sim_tanner1055_wer} in which no error-floor is observed for any decoder, we believe that the $[755, 344]$ Tanner code is a problematic code. Thus, we think designing a good code for ADMM decoding is important future work.
\begin{figure}[!htbp]
\psfrag{&esn0}{\scriptsize{$E_s/N_0$ (dB)}}
\psfrag{&wer}{\hspace{-0.5cm}\scriptsize{Error rate (WER)}}
\psfrag{&lclp}{\scriptsize{LCLP,~\cite{punekar2012low}}}
\psfrag{&sumproductBP}{\scriptsize{Sum-product,~\cite{punekar2012low}}}
\psfrag{&ADMMLPdecoding}{\scriptsize{ADMM LP decoding}}
\psfrag{&ADMMpenalizeddecodingpad}{\scriptsize{ADMM penalized decoding}}
    \begin{center}
    \includegraphics[width=21pc]{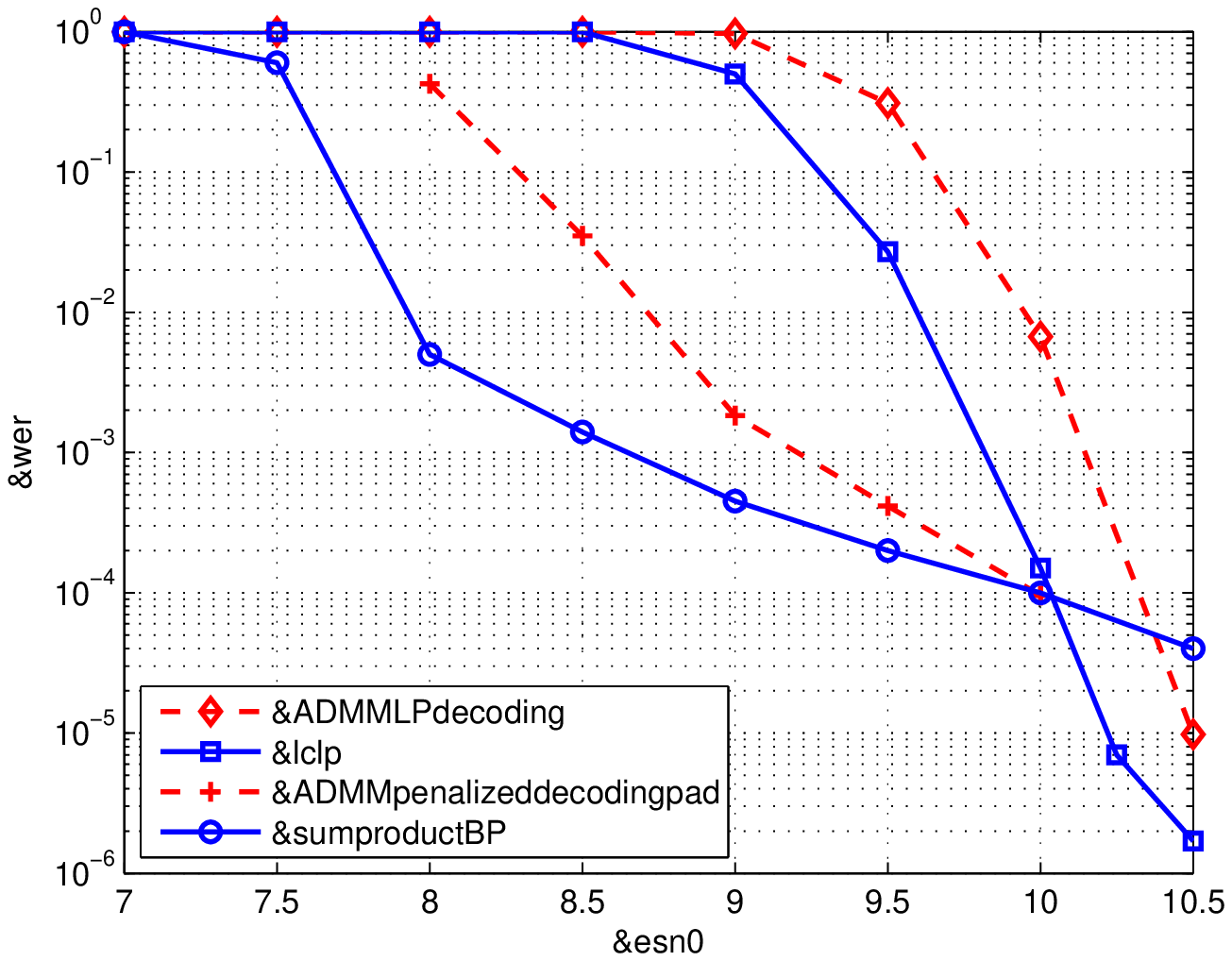}
    \end{center}
    \caption{WER plotted as a function of SNR for the Tanner $[755, 344]$ code in $\GF_{2^3}$.}
    \label{fig.sim_tanner755_wer}
\end{figure}

We make two concluding remarks. First, Algorithm~\ref{algorithm.ADMMLP_noprojection} solves the relaxed LP decoding problem efficiently. It is competitive with earlier decoders in terms of WER performance, especially for $\GF_{2^2}$. Second, the penalized decoder is a very promising decoder because it improves the low-SNR performance of LP decoding consistently. However, it suffers from high-complexity and may suffer from an early error-floor. It is important in future work to improve the complexity of the penalized decoder.

\subsection{Parameter choices for ADMM LP decoding}
\label{subsection.parameter_lp}
In this section we study the effects of the choice of parameters for Algorithm~\ref{algorithm.ADMMLP_noprojection}. We conduct simulations based on the $[1055,424]$ Tanner code described in the previous subsection, and use the early termination technique described in Section~\ref{subsubsection.early_term_and_overrel}. For each data point, we collect more than $100$ word errors.

The ADMM algorithm has many parameters: 
the ``step size'' of ADMM $\mu$ (cf.~\eqref{eq.ADMM_with_mu}), the maximum number of iterations $T_{\max}$, the ending tolerance $\epsilon$ (cf. Algorithm~\ref{algorithm.ADMMLP_noprojection}), and the over-relaxation parameter $\rho$ (cf. Section~\ref{subsubsection.early_term_and_overrel}). We fix the number of iterations to be $T_{\max} = 200$ and ending tolerance to be $\epsilon = 10^{-5}$ because the decoder is less sensitive to the settings of these two parameters than it is to $\mu$ or $\rho$. 

In Figures~\ref{fig.parameter_mu_different_rho_WER} and~\ref{fig.parameter_mu_different_rho_iter} we study the effects of the choice of parameter $\mu$ on the decoder. In Figure~\ref{fig.parameter_mu_different_rho_WER} we plot the WER as a function of $\mu$ when $\rho = 1$, $1.5$ and $1.9$. We ran simulations for $E_s/N_0 = 5$dB and $E_s/N_0 = 5.5$dB. In Figure~\ref{fig.parameter_choice_mu_cw_vs_flanagan_iter} we plot the corresponding average number of iterations. We make the following observations. First, the WER does not change much as a function of $\mu$ as long as $\mu \in [1,3]$. In addition, the lowest number of iterations is attained when $\mu \in [1,3]$. This means that a choice of $\mu$ within the range $[1,3]$ is good for practical reasons. Second, the trend of the curves for different values of $\rho$ is similar. In other words, the choice of $\rho$ does not strongly affect the choice of $\mu$, at least empirically. This is useful because it simplifies the task of choosing good parameters. Third, we observe that the WER curves for both SNRs flatten when $\mu \in [1,3]$  \footnote{We use the same seed to generate noises for different values of $\mu$.}. This indicates the parameter choice for this data point allows ADMM to converge fast enough to achieve the LP decoding solution within the set number of iterations, at least in most cases. In particular, if we increase $T_{\max}$, the range of the flat region increases. We observe that for $E_s/N_0 = 5$dB, $T_{\max} = 1000$, and $\rho = 1$, the WERs are between $0.122$ and $0.128$ for $\mu \in [0.3,3.1]$. This also shows that in this case increasing $T_{\max}$ does not greatly improve WER performance. Fourth, we observe that the WER performance of the decoder can be improved by assigning a large value to $\rho$. This effect is demonstrated further in Figures~\ref{fig.parameter_rho_different_mu_WER} and~\ref{fig.parameter_rho_different_mu_iter}. Finally, we observe that the average numbers of iterations are small for both SNRs. Note that the SNRs in Figures~\ref{fig.parameter_mu_different_rho_WER} and~\ref{fig.parameter_mu_different_rho_iter} are low SNRs and the corresponding WERs are high. Recall that in Figure~\ref{fig.sim_tanner1055_iter} the average number of iterations is much smaller than $T_{\max} = 200$ for high SNRs. This means that $T_{\max}$ does not need to be large.

\begin{figure}[!htbp]
\psfrag{&mu}{\scriptsize{$\mu$}}
\psfrag{&WER}{\hspace{-0.5cm}\scriptsize{Error rate (WER)}}
\psfrag{&esn05rho1padpadpadpadpad}{\scriptsize{$E_s/N_0 = 5$dB, $\rho = 1$}}
\psfrag{&esn05rho15}{\scriptsize{$E_s/N_0 = 5$dB, $\rho = 1.5$}}
\psfrag{&esn05rho19}{\scriptsize{$E_s/N_0 = 5$dB, $\rho = 1.9$}}
\psfrag{&esn055rho1}{\scriptsize{$E_s/N_0 = 5.5$dB, $\rho = 1$}}
\psfrag{&esn055rho15}{\scriptsize{$E_s/N_0 = 5.5$dB, $\rho = 1.5$}}
\psfrag{&esn055rho19}{\scriptsize{$E_s/N_0 = 5.5$dB, $\rho = 1.9$}}
    \begin{center}
    \includegraphics[width= 21pc]{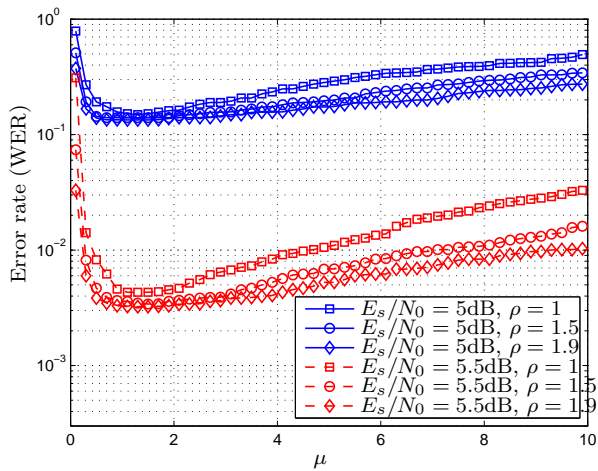}
    \end{center}
    \caption{ADMM step size $\mu$: WER plotted as a function of $\mu$ for the
      Tanner $[1055,424]$ code in $\GF_{2^2}$, $T_{\max} = 200$, and $\epsilon = 10^{-5}$.}
    \label{fig.parameter_mu_different_rho_WER}
\end{figure}

\begin{figure}[!htbp]
\psfrag{&mu}{\scriptsize{$\mu$}}
\psfrag{&iter}{\hspace{-0.7cm}\scriptsize{Number of iterations}}
\psfrag{&esn05rho1padpadpadpadpad}{\scriptsize{$E_s/N_0 = 5$dB, $\rho = 1$}}
\psfrag{&esn05rho15}{\scriptsize{$E_s/N_0 = 5$dB, $\rho = 1.5$}}
\psfrag{&esn05rho19}{\scriptsize{$E_s/N_0 = 5$dB, $\rho = 1.9$}}
\psfrag{&esn055rho1}{\scriptsize{$E_s/N_0 = 5.5$dB, $\rho = 1$}}
\psfrag{&esn055rho15}{\scriptsize{$E_s/N_0 = 5.5$dB, $\rho = 1.5$}}
\psfrag{&esn055rho19}{\scriptsize{$E_s/N_0 = 5.5$dB, $\rho = 1.9$}}
    \begin{center}
    \includegraphics[width=21pc]{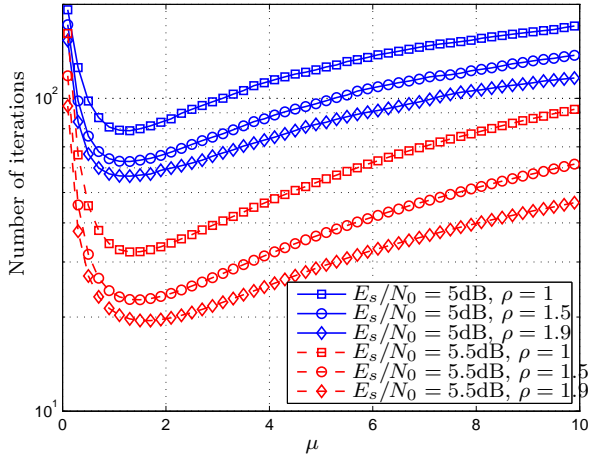}
    \end{center}
    \caption{ADMM step size $\mu$: Number of iterations plotted as a function of $\mu$ for the
      Tanner $[1055,424]$ code in $\GF_{2^2}$, $T_{\max} = 200$, and $\epsilon = 10^{-5}$.}
    \label{fig.parameter_mu_different_rho_iter}
\end{figure}

In Figures~\ref{fig.parameter_rho_different_mu_WER} and~\ref{fig.parameter_rho_different_mu_iter}, we study the effect of the choice of parameter $\rho$ on the decoder. In Figure~\ref{fig.parameter_rho_different_mu_WER} we plot WER as a function of $\rho$ when $\mu$ takes on values $1.5$ and $2$. Again we simulate for $E_s/N_0 = 5$dB and $E_s/N_0 = 5.5$dB. In Figure~\ref{fig.parameter_rho_different_mu_iter}, we plot the corresponding average number of iterations as a function of $\rho$. In these two figures we observe results that are consistent with Figures~\ref{fig.parameter_mu_different_rho_WER} and~\ref{fig.parameter_mu_different_rho_iter}, i.e., the number of iterations decreases as $\rho$ increases. However, the WER does not decrease significantly for $\rho > 1$. 

\begin{figure}[!htbp]
\psfrag{&rho}{\scriptsize{$\rho$}}
\psfrag{&wer}{\hspace{-0.5cm}\scriptsize{Error rate (WER)}}
\psfrag{&esn05mu15datapadpadpadpad}{\scriptsize{$E_s/N_0 = 5$dB, $\mu = 1.5$}}
\psfrag{&esn055mu15data}{\scriptsize{$E_s/N_0 = 5.5$dB, $\mu = 1.5$}}
\psfrag{&esn05mu2data}{\scriptsize{$E_s/N_0 = 5$dB, $\mu = 2$}}
\psfrag{&esn055mu2data}{\scriptsize{$E_s/N_0 = 5.5$dB, $\mu = 2$}}
    \begin{center}
    \includegraphics[width=21pc]{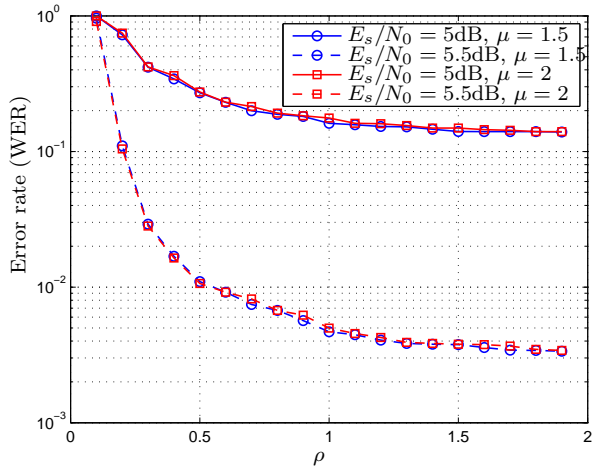}
    \end{center}
    \caption{Over-relaxation parameter $\rho$: WER plotted as a function of $\rho$ for the
      Tanner $[1055,424]$ code in $\GF_{2^2}$, $T_{\max} = 200$, and $\epsilon = 10^{-5}$.}
    \label{fig.parameter_rho_different_mu_WER}
\end{figure}

\begin{figure}[!htbp]
\psfrag{&rho}{\scriptsize{$\rho$}}
\psfrag{&iter}{\hspace{-0.7cm}\scriptsize{Number of iterations}}
\psfrag{&esn05mu15datapadpadpadpad}{\scriptsize{$E_s/N_0 = 5$dB, $\mu = 1.5$}}
\psfrag{&esn055mu15data}{\scriptsize{$E_s/N_0 = 5.5$dB, $\mu = 1.5$}}
\psfrag{&esn05mu2data}{\scriptsize{$E_s/N_0 = 5$dB, $\mu = 2$}}
\psfrag{&esn055mu2data}{\scriptsize{$E_s/N_0 = 5.5$dB, $\mu = 2$}}
    \begin{center}
    \includegraphics[width=21pc]{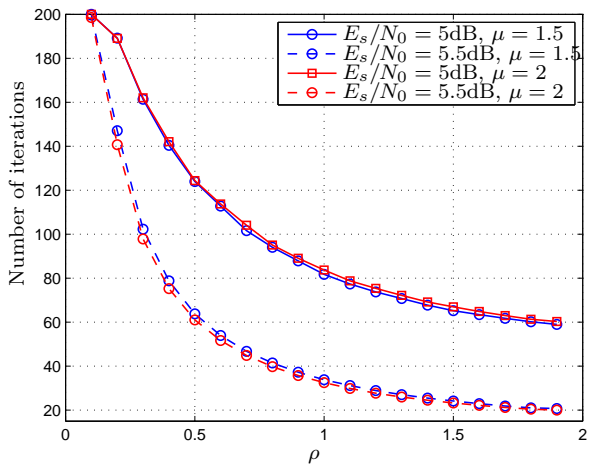}
    \end{center}
    \caption{Over-relaxation parameter $\rho$: Number of iterations plotted as a function of $\rho$ for the Tanner $[1055,424]$ code in $\GF_{2^2}$, $T_{\max} = 200$, and $\epsilon = 10^{-5}$.}
    \label{fig.parameter_rho_different_mu_iter}
\end{figure}

Finally, we compare \nameF and \nameTheCW in terms of decoding performance. For \nameF we use Algorithm~\ref{algorithm.ADMMLP_noprojection}. For \nameTheCW we apply the method in Algorithm~\ref{algorithm.ADMMLP_noprojection} to the LP decoding problem~\eqref{eq.cwlpdecoding_relax}. Due to the similarity between \nameF and \nameTheCW, results in Section~\ref{section.LPandADMM_noADMMProj} can be easily extended to \nameTheCW. In Figures~\ref{fig.parameter_choice_mu_cw_vs_flanagan_wer} and~\ref{fig.parameter_choice_mu_cw_vs_flanagan_iter}, we compare both WER and the average number of iterations for the two decoders. In this set of simulations we let $E_s/N_0 = 5$dB,  $T_{\max} = 200$, and $\epsilon = 10^{-5}$. In Figure~\ref{fig.parameter_choice_mu_cw_vs_flanagan_wer} we observe that the WER for \nameTheCW is higher than it is for \nameF in most cases. The exception is the data points at $\rho = 1.5$ and $\mu \in [0.7, 1.7]$, where both embeddings attain the same WERs. This exception is because that parameter choices allow ADMM to converge fast enough to obtain the LP decoding solution. Recall that we proved that the two embedding methods are equivalent in terms of LP decoding (cf. Corollary~\ref{corollary.equiv_error_rate_relaxed}). It is not surprising that both decoders can achieve the same WERs when a sufficient number of iterations is allowed.
From Figure~\ref{fig.parameter_choice_mu_cw_vs_flanagan_iter} we observe that the number of iterations for \nameTheCW is greater than that required for \nameF. Recall that \nameF saves exactly $1$ coordinate for each embedded vector when compared to \nameTheCW (cf. Section~\ref{sec.embedding_finite_fields}). As a result, \nameTheCW requires slightly more memory than \nameF. Thus, in practice, \nameF for LP decoding would be preferred. 
\begin{figure}[!htbp]
\psfrag{&mu}{\scriptsize{$\mu$}}
\psfrag{&WER}{\hspace{-0.5cm}\scriptsize{Error rate (WER)}}
\psfrag{&esn05rho1}{\scriptsize{$\rho = 1$}}
\psfrag{&esn05rho15}{\scriptsize{$\rho = 1.5$}}
\psfrag{&esn05rho1cw}{\scriptsize{$\rho = 1$, CW}}
\psfrag{&esn05rho15cw}{\scriptsize{$\rho = 1.5$, CW}}
    \begin{center}
    \includegraphics[width=21pc]{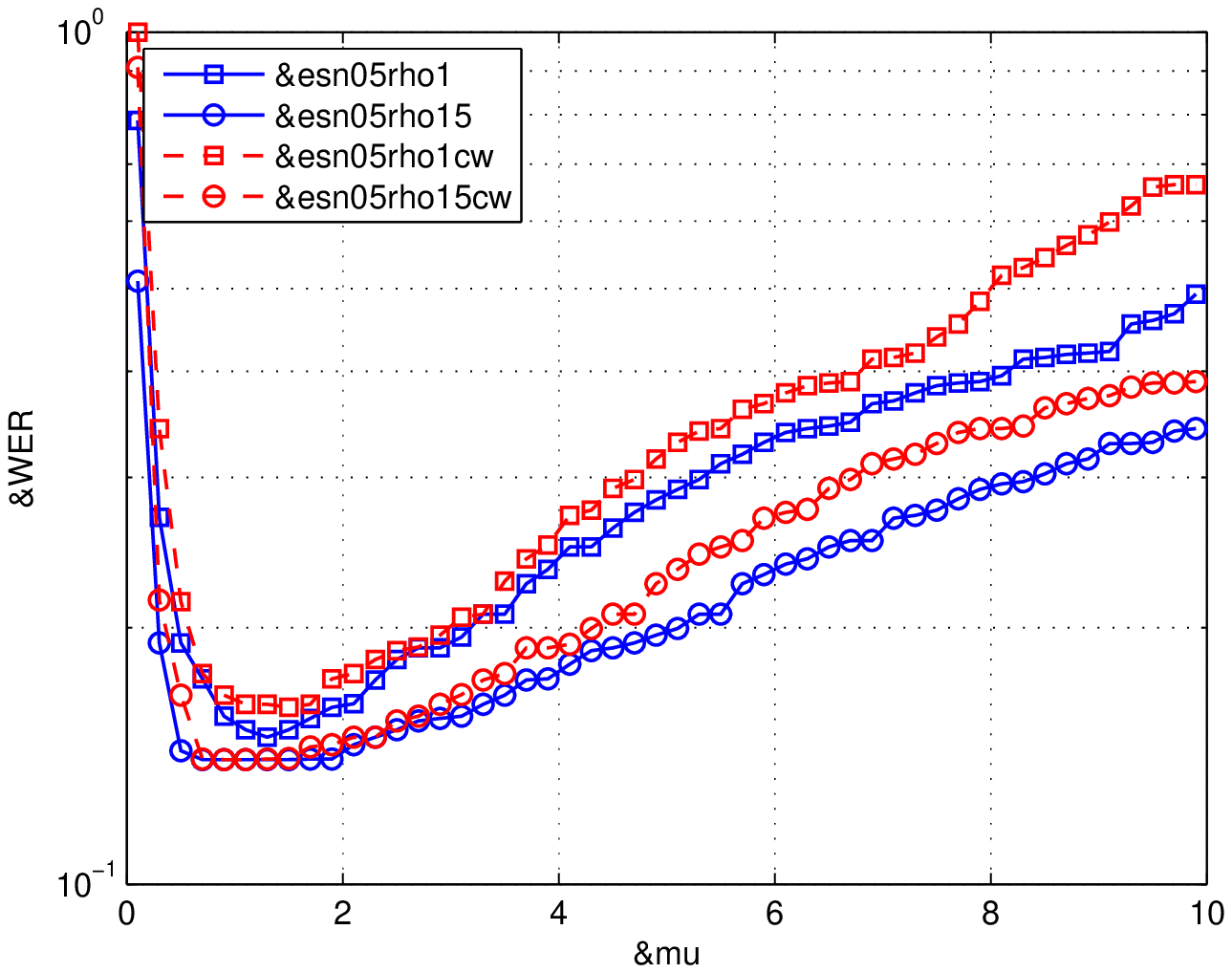}
    \end{center}
    \caption{Choice of embedding method: WER plotted as a function of ADMM step size $\mu$ for the Tanner $[1055,424]$ code in $\GF_{2^2}$. The SNR is $E_s/N_0 = 5$dB and $T_{\max} = 200$. ``CW'' stands for ``constant weight''.}
    \label{fig.parameter_choice_mu_cw_vs_flanagan_wer}
\end{figure}

\begin{figure}[!htbp]
\psfrag{&mu}{\scriptsize{$\mu$}}
\psfrag{&iter}{\hspace{-0.7cm}\scriptsize{Number of iterations}}
\psfrag{&esn05rho1}{\scriptsize{$\rho = 1$}}
\psfrag{&esn05rho15}{\scriptsize{$\rho = 1.5$}}
\psfrag{&esn05rho1cw}{\scriptsize{$\rho = 1$, CW}}
\psfrag{&esn05rho15cw}{\scriptsize{$\rho = 1.5$, CW}}
    \begin{center}
    \includegraphics[width=21pc]{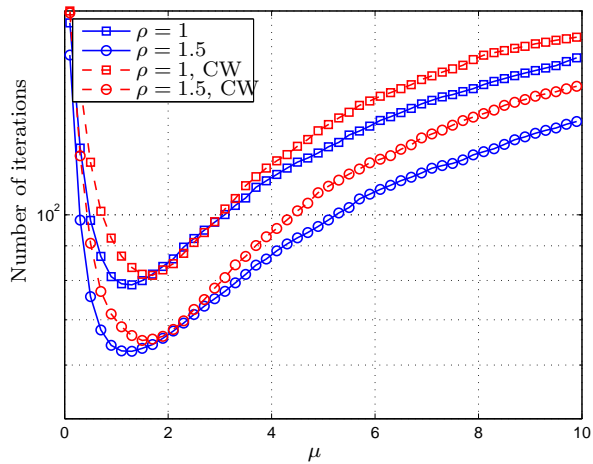}
    \end{center}
    \caption{Choice of embedding method: Number of iterations plotted as a function of ADMM step size $\mu$ for the
      Tanner $[1055,424]$ code in $\GF_{2^2}$. The SNR is $E_s/N_0 = 5$dB and $T_{\max} = 200$. ``CW'' stands for ``constant weight''.}
    \label{fig.parameter_choice_mu_cw_vs_flanagan_iter}
\end{figure}

To summarize the above discussions, we emphasize four points:
\begin{itemize}
\item $\mu \in [1,2]$ and $\rho = 1.9$ are good parameter choices for the Tanner $[1055,424]$ code.
\item The error performance of the decoder only weakly depends on parameter settings.
\item \nameF is slightly better than \nameTheCW in terms of computational complexity and memory usage. 
\item The effects of the choice of parameters for the PEG $[2048, 1018]$ code discussed in Section~\ref{subsection.snr_performance} are similar to that for the Tanner $[1055,424]$ code. We observe that $\mu \in [1,3.3]$ and $\rho = 1.9$ are good parameter choices for the PEG $[2048, 1018]$ code.
\end{itemize}

\subsection{Parameter choices for penalized decoding}
\label{subsection.penalized}
We now study the ADMM penalized decoder and show how to choose a good penalty parameter $\alpha$. We use the same Tanner $[1055,424]$ code from Section~\ref{subsection.snr_performance} and simulate for the AWGN channel. For this simulation, we do not perform an extensive grid search over the possible parameter choices because of limited computational resources. However, the following parameters are good enough to produce good decoding results: $\mu = 4$, $\rho = 1.5$, and $T_{\max} = 100$. 

In Figure~\ref{fig.parameter_pd_alpha_wer} we plot WER as a function of the penalty parameter $\alpha$. We observe that the WER first decreases as the value of $\alpha$ increases. But when we penalize too heavily (e.g, when $\alpha > 0.5$), the WER starts to increase again. A similar effect can be observed in Figure~\ref{fig.parameter_pd_alpha_iter}, i.e., the average number of iterations decreases 
for $\alpha$ small, and then starts to increase once $\alpha$ becomes sufficiently large ($\alpha > 0.8$ in Figure~\ref{fig.parameter_pd_alpha_iter}). 
We make two remarks. First, we observe that $\alpha = 0.5$ is optimal in terms of WER for all SNRs we simulated. This is \textbf{not} necessarily the case for other codes or other SNRs. We do observe that for some codes (e.g., the Tanner $[755, 334]$ code studied in Section~\ref{subsection.snr_performance}) the optimal value for $\alpha$ is a (non-constant) function of SNR. Second, we observe that the optimal $\alpha$ in terms of WER is \textbf{not} the optimal in terms of the average number of iterations. As a result, there is a trade-off between the optimal computational complexity and the optimal WERs in choosing $\alpha$. 
\begin{figure}[!htbp]
\psfrag{&alpha}{\scriptsize{$\alpha$}}
\psfrag{&wer}{\hspace{-0.5cm}\scriptsize{Error rate (WER)}}
\psfrag{&esn045penalized-}{\scriptsize{$E_s/N_0 = 4.5$dB}}
\psfrag{&esn047penalized}{\scriptsize{$E_s/N_0 = 4.7$dB}}
\psfrag{&esn05penalized}{\scriptsize{$E_s/N_0 = 5$dB}}
    \begin{center}
    \includegraphics[width=21pc]{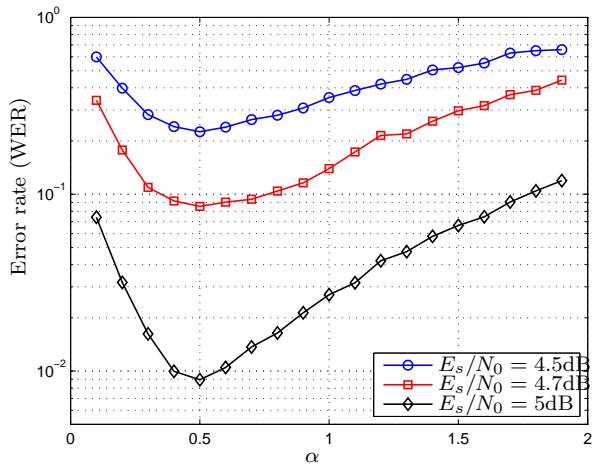}
    \end{center}
    \caption{Penalty coefficient $\alpha$: ADMM penalized decoding WER plotted as a function of $\alpha$ for the
      Tanner $[1055,424]$ code in $\GF_{2^2}$.}
    \label{fig.parameter_pd_alpha_wer}
\end{figure}
\begin{figure}[!htbp]
\psfrag{&alpha}{\scriptsize{$\alpha$}}
\psfrag{&iter}{\hspace{-0.7cm}\scriptsize{Number of iterations}}
\psfrag{&esn045penalized-}{\scriptsize{$E_s/N_0 = 4.5$dB}}
\psfrag{&esn047penalized}{\scriptsize{$E_s/N_0 = 4.7$dB}}
\psfrag{&esn05penalized}{\scriptsize{$E_s/N_0 = 5$dB}}
    \begin{center}
    \includegraphics[width=21pc]{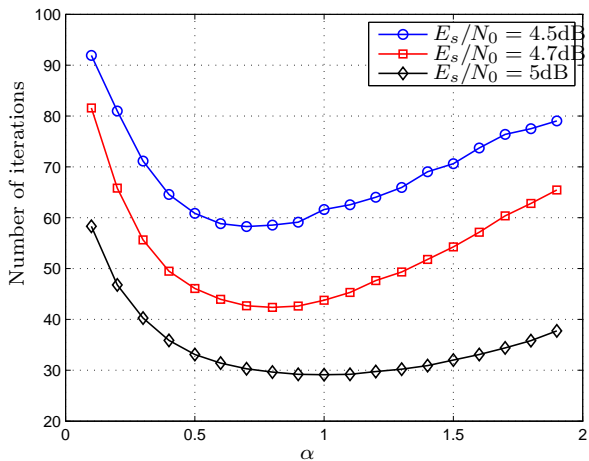}
    \end{center}
    \caption{Penalty coefficient $\alpha$: ADMM penalized decoding number of iterations plotted as a function of $\mu$ for the
      Tanner $[1055,424]$ code in $\GF_{2^2}$.}
    \label{fig.parameter_pd_alpha_iter}
\end{figure}
\section{Conclusions}
\label{section.conclusions}
In this paper, we develop two types of ADMM decoders for decoding non-binary codes in fields of characteristic two. The first is a natural extension of the binary ADMM decoder introduced by Barman \emph{et al.} in~\cite{barman2013decomposition}. This ADMM algorithm requires a sub-routine that projects a vector onto the polytope formed by the embeddings of an SPC code. The second ADMM algorithm leverages the factor graph representation of the $\GF_{2^m}$ SPC code embedding and does not require the projection sub-routine that is necessary for the first ADMM algorithm. We further improve the efficiency of this latter algorithm using new properties of the embedding of SPC codes that we have discovered. 

We summarize three important future directions relevant to this topic. First, the relaxed code polytope considered in this work is conjectured to be tight for $\GF_{2^2}$. The proof of this conjecture is still open. To the best of the authors' knowledge, there has been no known characterization of $\GF_{2^2}$ pseudocodewords. Therefore we believe this conjecture, if proved, is an important theoretical result to develop en-route to understanding pseudocodewords of non-binary LP decoding. Second, as mentioned many times in the paper, improving the projection sub-routine of ADMM decoding is important future work (i.e., for the ``first'' type of ADMM algorithm as per the previous paragraph). 
Finally, we show through execution time statistics that the proposed ADMM LP decoder is efficient and could be interesting in practical (hardware) implementation. However, many missing pieces remain to be developed before the decoder can be applied in practice. In particular, there have been no approximation algorithms developed for ADMM LP decoding that are analogous to the relationship between the min-sum algorithm and the sum-product algorithm. Such algorithms may substantially reduce computational complexity and hence are important directions of future work. Furthermore, studying finite precision effects on ADMM LP decoding is another crucial direction to pursue, one that is very  important to hardware implementation.

\section*{Acknowledgments}
The authors would like to thank Vitaly Skachek for useful discussions and references. The authors would also like to thank the Associate Editor, Prof. Roxana Smarandache, and three anonymous
reviewers for their comments and suggestions.
\section*{Appendix}
\section{Definitions and Lemmas for \nameTheCW}
\label{appendix.def_for_cw_embedding}
In this section, we state the definitions and lemmas for \nameTheCW that are omitted in Section~\ref{sec.embedding_and_SPC}.
\begin{definition}
\label{def.valid_cw_embedding}
Under \nameTheCW let $q = 2^{m}$ and denote by $\mcE'$ the set of \textit{valid embedded matrices} defined by check $\bfh$. An $q\times \checkdegree$ binary matrix $\bfF \in \mcE'$ if and only if it satisfies the following conditions: 
\begin{enumerate}
\item[$(a)$] $f_{ij} \in \{0,1\}$, where $i = 0,\dots,2^{m-1}$ and $j = 1,\dots,\checkdegree$.
\item[$(b)$] $\sum_{i = 0}^{q-1} f_{ij} = 1$ where $q = 2^m$.
\item[$(c)$] For any non-zero $h\in \GF_{2^m}$ and $k\in[m]$, let $\mcB(k, h) := \{ \alpha | \embedbit(h \alpha)_k = 1, \alpha \in \GF_{2^m}\}$,  where $\cdot_k$ denotes the $k$-th entry of the vector. Let $g^k_j =  \sum_{i \in \mcB(k, h_j)} f_{ij}$, then
\begin{equation}
\sum_{j = 1}^\checkdegree g^k_j = 0 \qquad \text{for all }k \in [m]
\end{equation}
where the addition is in $\GF_2$.
\end{enumerate} 
\end{definition}
Note that condition $(c)$ is the same for both embeddings. This is due to two facts. First, $0 \not\in\mcB(k, h)$. Second, \nameTheCW is indexed starting from $0$ (cf. Definition~\ref{def.constant_weight_embedding}). Thus, entries from Definition~\ref{def.validembedding} that participate in $\mcB(k, h)$ remain unchanged.
\begin{lemma}
\label{lemma.valid_cw_codeword}
In $\GF_{2^m}$, let $\mcC$ be the set of codewords that correspond to check $\bfh = (h_1,h_2,\dots,h_\checkdegree)$. Then
\begin{equation}
\cwembedmat(\mcC) = \mcE',
\end{equation}
where $\mcE'$ is defined by Definition~\ref{def.valid_cw_embedding}.
\end{lemma}
\begin{proof}
See Appendix~\ref{proof.valid_codeword}.
\end{proof}

\begin{definition}
\label{def.cw_factor_graph_SPC}
The conditions in Definition~\ref{def.valid_cw_embedding} can be represented using factor graphs. In this graph, there are three types of nodes: (i) \emph{variable nodes}, (ii) \emph{parity-check nodes} and (iii) \emph{one-on-check node}. Each variable node corresponds to an entry $f_{ij}$ for $0\leq i \leq 2^m-1$ and $1\leq j \leq \checkdegree$. Each parity-check node corresponds to a parity-check of the $k$-th bit where $k \in [m]$; it connects all variable nodes in $\mcB(k,h_j)$ for all $j\in[d_c]$. Each one-on-check node corresponds to a constraint $\sum_{i = 0}^{q-1} f_{ij} = 1$. 
\end{definition}

\begin{lemma}
\label{lemma.cw_equiv_integer_constraint}
Consider \nameTheCW and let $\mcF'$ be the set of $q \times \checkdegree$ binary matrices defined by the first two conditions of Definition~\ref{def.valid_cw_embedding} but with the third condition replaced by the condition $(c^*)$ defined below, then $\mcF' = \mcE'$. The complete set of conditions are
\begin{enumerate}
\item[$(a)$] $f_{ij} \in \{0,1\}$.
\item[$(b)$] $\sum_{i = 0}^{q-1} f_{ij} = 1$ where $q = 2^m$.
\item[$(c^*)$] Let $\K$ be any nonempty subset of $[m]$. For any non-zero $h\in \GF_{2^m}$, let $\mcBB(\K, h) := \left\{\alpha| \sum_{k \in \K}\embedbit(h \alpha)_k = 1\right\}$, where $\embedbit(h \alpha)_k$ denotes the $k$-th entry of the vector $\embedbit(h \alpha)$. Let $g^\K_j =  \sum_{i \in \mcBB(\K, h_j)} f_{ij}$, then $\sum_{j = 1}^\checkdegree g^\K_j = 0$ for all $\K \subset [m]$, $\K\neq \emptyset$, where the addition is in $\GF_{2}$.
\end{enumerate} 
\end{lemma}
\begin{proof}
See Appendix~\ref{proof.equiv_integer_constraint}
\end{proof}

\begin{definition}
\label{def.cw_code_polytopes}
Let $\mcC$ be a non-binary SPC code defined by check vector $\bfh$. Denote by $\cwtightcodepolytope$ the ``tight code polytope'' for \nameTheCW where 
$$\cwtightcodepolytope = \conv(\cwembedmat(\mcC)).$$ 

In $\GF_{2^m}$ denote by $\cwrelaxedcodepolytope$ the ``relaxed code polytope'' for \nameTheCW where a $q \times \checkdegree$ matrix $\bfF \in \cwrelaxedcodepolytope$ if and only if the following constraints hold:
\begin{enumerate}
\item[$(a)$] $f_{ij} \in [0,1]$.
\item[$(b)$] $\sum_{i = 0}^{q-1} f_{ij} = 1$.
\item[$(c)$] Let $\mcBB(\K,h)$ be the same set defined in Lemma~\ref{lemma.cw_equiv_integer_constraint}. Let $\bfg^\K$ be a vector of length $\checkdegree$ such that $g^\K_j = \sum_{i \in \mcBB(\K, h_j)} f_{ij}$, then
\begin{equation}
\bfg^\K \in \PP_{\checkdegree} \qquad \text{for all }\K \subset [m] \text{ and } \K\neq \emptyset
\end{equation}
\end{enumerate}
\end{definition} 

\begin{definition}
\label{def.cw_permutation}
Let $h$ be a non-zero element in $\GF_{q}$. Let $\cwrotmat(q,h)$ be a $q \times q$ permutation matrix indexed starting from row $0$ and column $0$, and satisfying the following equation
\begin{equation*}
\cwrotmat(q,h)_{ij} = 
\begin{cases}
1 & \quad \text{ if } i \cdot h = j, \\
0 & \quad \text{ otherwise,}
\end{cases}
\end{equation*}
where the multiplication $i\cdot h$ is in $\GF_{q}$. For a vector $\bfh \in \GF_{q}^\checkdegree$ such that $h_j\neq 0 $ for all $j \in [d_c]$, let $$\cwrotmat(q,\bfh) = \diag(\cwrotmat(q,h_1),\dots,\cwrotmat(q,h_\checkdegree)).$$
\end{definition}

We note that Lemma~\ref{lemma.general_rotation_equivalence}, Lemma~\ref{lemma.general_rotation_projection}, Lemma~\ref{lemma.rotation_equivalence} and Lemma~\ref{lemma.rotation_projection} all hold for \nameTheCW when $\cwrotmat$ is used. The proofs for \nameTheCW can be easily extended from the proofs for \nameF. 

\section{Proofs}
\subsection{Proof of Proposition~\ref{proposition.number_bit_constraints}}
\label{proof.number_bit_constraints}
First, $|\mcB(k,1)| = 2^{m-1}$ because the $k$-th bit is fixed to $1$. Since dividing by a fixed $h$ only permutes the elements in finite fields, $|\mcB(k,h)|$ is also $2^{m-1}$.

\subsection{Proof of Lemma~\ref{lemma.valid_codeword} and Lemma~\ref{lemma.valid_cw_codeword}}
\label{proof.valid_codeword}
\subsubsection*{Proof of Lemma~\ref{lemma.valid_codeword}}
We first prove that for a vector $\bfc \in \GF_{2^m}^{\checkdegree}$ and its embedded matrix $\bfF$, $g_j^k = 1$ if and only if $\embedbit(h_j c_j)_k = 1$ for all $j$ and $k$. For a fixed $j$, if $\embedbit(h_j c_j)_k = 1$, then $c_j \in \mcB(k, h_j)$ by the definition of $\mcB(k, h_j)$. Note that the $j$-th column of $\bfF$, denoted by $\bff_j$, is a vector such that $f_{c_jj} = 1$ and $f_{ij} = 0$ for all $i \neq c_j$. Hence $(g_j^k := )\sum_{i\in \mcB(k, h_j)} f_{ij} = f_{c_jj} = 1$. If $\embedbit(h_j c_j)_k = 0$, then $c_j \notin \mcB(k, h_j)$. Since $f_{ij} = 0$ for all $i \neq c_j$, we get $\sum_{i\in \mcB(k, h_j)} f_{ij} = 0$.

Using this result we show $\bfF \in \embedmat(\mcC)$ implies $\bfF \in \mcE$. If $\bfF \in \embedmat(\mcC)$, then there exists a vector $\bfc \in \mcC$ such that $\bfF = \embedmat(\bfc)$.  By the definition of $\embedmat(\bfc)$, condition $(a)$ and $(b)$ per Definition~\ref{def.validembedding} are satisfied by $\bfF$. For condition $(c)$, note that $\bfc$ is a codeword. This implies that $v := \sum_{j = 1}^{\checkdegree} h_j c_j = 0$ in $\GF_{2^m}$. Further, due to the fact that we are working in extension fields of $2$, $\embedbit(v)$ is an all-zero vector of length $m$. This means that for all $k = 1,\dots,m$, $\embedbit(v)_k = \sum_{j = 1}^\checkdegree \embedbit(h_j c_j)_k = 0$.  Thus, $\sum_{j = 1}^{\checkdegree} g_j^k = \sum_{j = 1}^{\checkdegree} \embedbit(h_j c_j)_k = 0$. This implies that condition $(c)$ is satisfied.

Conversely, let $\bfF\in \mcE$. Then by condition $(a)$ and $(b)$, we know that there is a unique vector $\bfc \in \GF_{2^m}^\checkdegree$ such that $\bfF = \embedmat(\bfc)$. By condition $(c)$, we know that for any fixed $k$, $\sum_{j = 1}^{d_c} \sum_{i\in \mcB(k, h_j)} f_{ij} = \sum_{j = 1}^{d_c} \embedbit(h_j c_j)_k = 0$. Therefore, $\embedbit(v)_k = 0$ for all $k$, where $v := \sum_{j = 1}^\checkdegree h_j c_j$. This implies that $\bfc$ is a valid codeword, or equivalently, $\bfc\in \mcC$ and hence $\bfF \in  \embedmat(\mcC)$.

\subsubsection*{Proof of Lemma~\ref{lemma.valid_cw_codeword}}
This is a simple consequence due to the bijective relationship between \nameTheCW and \nameF. Formally, if a matrix $\bfF' \in \mcE'$, we construct $\bfF$ by eliminating the first row of $\bfF'$. Then it is simple to verify that $\bfF \in \mcE$. By Lemma~\ref{lemma.valid_codeword}, $\bfF$ is the embedding for some codeword $\bfc$. This shows that $\bfF' \in \cwembedmat(\mcC)$. Conversely, for each codeword $\bfc \in \mcC$, we can reverse the process and find an $\bfF'$ that satisfies Definition~\ref{def.valid_cw_embedding}.

\subsection{Proof of Lemma~\ref{lemma.equiv_integer_constraint} and Lemma~\ref{lemma.cw_equiv_integer_constraint}}
\label{proof.equiv_integer_constraint}
\subsubsection*{Proof of Lemma~\ref{lemma.equiv_integer_constraint}}
We first show that $\sum_{k\in \K}\embedbit(h_j c_j)_k = 1$ if and only if $g^\K_j = 1$. Suppose $g^\K_j = 1$, then there exists a unique $i$ such that $i \in \mcBB(\K, h_j)$ and $f_{ij} = 1$. Since $f_{c_jj}= 1$ and condition $(b)$ holds, $i = c_j$. Therefore $c_j \in \mcBB(\K, h_j)$. By the definition of $\mcBB(\K, h_j)$, $\sum_{k\in \K}\embedbit(h_j c_j)_k = 1$. On the other hand, if $\sum_{k\in }\embedbit(h_j c_j)_k = 1$, then $c_j\in \mcBB(\K, h_j)$. This implies $g^\K_j := \sum_{i\in \mcBB(\K,h_j)}f_{ij} = 1$. Next, note that condition $(c^*)$ only adds more constraints to Definition~\ref{def.validembedding}. Thus $\mcF \subset \mcE$. Conversely, using the same definition in the proof in Appendix~\ref{proof.valid_codeword}, $\embedbit(v)_k = \sum_{j = 1}^\checkdegree \embedbit(h_j c_j)_k = 0$. This means that for any subset $\K$, $\sum_{k\in \K}(\sum_{j = 1}^\checkdegree \embedbit(h_j c_j)_k) = 0$. In addition,  $$\sum_{k\in \K}\sum_{j = 1}^\checkdegree \embedbit(h_j c_j)_k = \sum_{j = 1}^\checkdegree \sum_{k\in \K}\embedbit(h_j c_j)_k.$$ Using the result we just proved, we deduce that $\sum_{j = 1}^\checkdegree g^\K_j = 0$.  This shows that $\mcE \subset \mcF$. In other words, we have shown that any valid codeword $\bfc \in \mcC$ satisfies condition $(c^*)$ in Lemma~\ref{lemma.equiv_integer_constraint}.
\subsubsection*{Proof of Lemma~\ref{lemma.cw_equiv_integer_constraint}}
It is easy to verify this lemma following the same logic in Appendix~\ref{proof.valid_codeword}.

\subsection{Proof of Lemma~\ref{lemma.number_bitset_constraints}}
\label{proof.number_bitset_constraints}
First consider all binary vectors of length $|\K|$ that sum up to $1$. The number of such vectors is $2^{|\K| - 1}$. The number of binary vectors of length $m - |\K|$ is $2^{m - |\K|}$. Thus $|\mcBB(\K,h)| = 2^{|\K| - 1}\times 2^{m - |\K|} =  2^{m-1}$.

\subsection{Proof of Proposition~\ref{prop.ppd_tight}}
\label{proof.ppd_tight}
We first introduce notation that we use in this proof. Let ``$\prec$'' denote the relationship of majorization. Let $\bfa$ and $\bfb$ be length-$n$ vectors, then $\bfb$ majorizes $\bfa$ is written as $\bfa \prec \bfb$. We refer readers to Definition~1 in~\cite{barman2013decomposition} for the definition of majorization. We use Theorem~1 in~\cite{barman2013decomposition} in this proof.

Let $\bfy \in \PP_s$ and consider the length-$st$ vector $\bfy^* = (\bfy, \bfzero_{1\times (st - s)})$. First, we show $\bfy^* \in \PP_{st}$. This immediately implies that all permutations of $\bfy^*$ are in $\PP_{st}$. Then we show $\bfx \prec \bfy^*$. By Theorem~1 in~\cite{barman2013decomposition}, this implies that $\bfx$ is in the convex hull of permutations of $\bfy^*$, each of which is in $\PP_{st}$. Therefore $\bfx \in \PP_{st}$. 

Since $\bfy \in \PP_s$, there exists a set $\mcR$ of length-$s$ binary vectors with even parity such that $\bfy = \sum_{i = 1}^{|\mcR|} w_i \bfr_i$ for some $w_i\geq 0$ and $\sum_{i=1}^{|\mcR|} w_i = 1$. Let $\bfr_i^* =  (\bfr_i , \bfzero_{1\times (st - s)})$, then $\bfr_{i}^*$ is a binary vector of length $st$ with even parity. Also note that $\bfy^* = \sum_{i = 1}^{|\mcR|} w_i \bfr^*_i$. This implies $\bfy^* \in \PP_{st}$.

In order to show $\bfx \prec \bfy^*$, we need to consider partial sums based on the sorted vectors of $\bfx$ and $\bfy^*$. Let $\mcQ_k$ (resp. $\mcU_k$) be the set indicies of the $k$ largest entries in $\bfx$ (resp. $\bfy^*$). Let the operator $\cover(\cdot)$ be defined as $i := \cover(j)$ if $j \in \{(i-1)s+1,(i-1)s+2,\dots,is\}$. Since all entries of $\bfx$ and $\bfy^*$ are non-negative, $x_j \leq y^*_{\cover(j)}$ for all $j$. Therefore $\sum_{j\in\mcQ_k} x_j \leq \sum_{j\in\mcQ_k}  y^*_{\cover(j)} \leq \sum_{i\in\mcU_k}  y^*_i$, where the second inequality is because $\mcU_k$ contains the largest $k$ entries of $\bfy^*$. In addition, $\|\bfx\|_1 = \|\bfy^*\|_1$. By the definition of majorization, these imply that $\bfx \prec \bfy^*$. 
\subsection{Proof of lemmas on rotation}
\label{proof.rotation}
\subsubsection{Proof of Lemma~\ref{lemma.general_rotation_equivalence}}
We first show that if $\bfuN \in \tightcodepolytopeN$ then $\bfu = \rotmat(q, \bfh) \bfuN \in \tightcodepolytope $. Let $\bfeN_k$, $k = 1,\dots,|\mcEN|$, be vectors in $\mcEN$, then
\begin{align*}
\rotmat(q, \bfh) \bfuN & = \rotmat(q, \bfh) \sum_k \alpha_k \bfeN_k\\
& = \sum_k \alpha_k (\rotmat(q, \bfh)\bfeN_k)\\
& = \sum_k \alpha_k \bfe_k,
\end{align*}
where $\bfe_k$ are vectors in $\mcE$. Therefore $\bfv \in \tightcodepolytope$.

Conversely, if $\bfu \in \tightcodepolytope$, $\bfuN = \rotmat(q, \bfh)^{-1} \bfu $. Let $\bfe_k$, $k = 1,\dots,|\mcE|$, be vectors in $\mcE$, then
\begin{align*}
\rotmat(q, \bfh)^{-1} \bfu & =\rotmat(q, \bfh)^{-1} \sum_k \alpha_k \bfe_k\\
& = \sum_k \alpha_k (\rotmat(q, \bfh)^{-1}\bfe_k)\\
& = \sum_k \alpha_k \bfeN_k,
\end{align*}
where $\bfeN_k$ are vectors in $\mcEN$. Therefore $\bfu \in \tightcodepolytopeN$.

\subsubsection{Proof of Lemma~\ref{lemma.general_rotation_projection}}
First, by Lemma~\ref{lemma.general_rotation_equivalence}, $\bfu \in \tightcodepolytope$. Therefore we need to show that there is no vector in $\tightcodepolytope$ that has a smaller Euclidean distance to $\bfv$ then does $\bfu$. Suppose that there is a vector $\bfp$ such that $\|\bfp - \bfv\|_2^2 < \|\rotmat(q, \bfh) \bfuN - \bfv\|_2^2$. Then $\|\bfp - \bfv\|_2^2 = \|\rotmat(q, \bfh)^{-1}(\bfp - \bfv)\|_2^2 = \|\rotmat(q, \bfh)^{-1}\bfp - \rotmat(q, \bfh)^{-1}\bfv\|_2^2 = \|\rotmat(q, \bfh)^{-1}\bfp - \bfvN\|_2^2$, where the second equality is due to the fact that $\rotmat(q, \bfh)$ is a permutation matrix. We then deduce that 
\begin{align*}
\|\rotmat(q, \bfh)^{-1}\bfp - \bfvN\|_2^2 < \|\bfuN -\rotmat(q, \bfh)^{-1} \bfv\|_2^2,
\end{align*}
which contradicts with the fact that $\bfuN$ is the projection of $\bfvN$ onto $\tightcodepolytopeN$.

\subsubsection{Proof of Lemma~\ref{lemma.rotation_equivalence}}
Let $\bfw \in \relaxedcodepolytopeN$, we need to show $\bfu := \rotmat(q, \bfh) \bfw$ is in $\relaxedcodepolytope$. Let $\bfU$ and $\bfW$ be the matrix form for $\bfu$ and $\bfw$ respectively (i.e. $\vect(\bfU) = \bfu$ and $\vect(\bfW) = \bfw$). Note that the rotations are done per embedded vector. In other words, columns of $\bfU$ are permuted independently of each other. By the definition of the rotation matrix, for each column $j$,
\begin{equation}
\label{eq.rot_lemma1}
u_{ij} = w_{lj} \text{ for all $l = i\cdot h$},
\end{equation}
where the multiplication is in $\GF_{2^m}$.
Since $\bfw \in \relaxedcodepolytopeN$, $\bfW$ satisfies Definition~\ref{def.relaxed_code_polytope} for check $(1,1,1,..,1)$. We first verify that $\bfU$ satisfies conditions $(a)$ and $(b)$ in Definition~\ref{def.relaxed_code_polytope}. Note that these two conditions are defined on a column-wise basis. Thus these two conditions are invariant to column-wise permutations. For condition $(c)$, note that $\mcBB(\K,h)$ can be obtained by $\mcBB(\K,1)$ as follows:
\begin{equation}
\label{eq.rot_lemma2}
\mcBB(\K,h) = \{y | y = x/h, x \in\mcBB(\K,1) \}.
\end{equation}
Then, $i\in\mcBB(\K,h_j)$ if $l\in \mcBB(\K,1)$, provided that $l = i\cdot h_j$. Let $\bff_{\mcBB(\K, h_j),j}$ be the sub-vector constructed by selecting entries $f_{ij}$ where $i\in \mcBB(\K, h_j)$. Combining \eqref{eq.rot_lemma1} and \eqref{eq.rot_lemma2}, we deduce that for fixed $j$, $\K$ and $l$, if $w_{lj}$ is in the vector $\bff_{\mcBB(\K,1),j}$, then $u_{ij}$ is in the vector $\bff_{\mcBB(\K,h_j),j}$ and $u_{ij} = w_{lj}$. This shows that $g^\K_j = \|\bff_{\mcBB(\K, h_j),j}\|_1$ is not changed under rotation. Therefore $\bfu$ satisfies condition $(c)$. This implies that $\bfu\in \relaxedcodepolytope$. 

Conversely, let $\bfu \in \relaxedcodepolytope$, we need to show $\bfw := \rotmat(q, \bfh)^{-1} \bfu$ is in $\relaxedcodepolytopeN$. Note that $\rotmat(q, \bfh)$ is a permutation matrix and therefore $\rotmat(q, \bfh)^{-1} = \rotmat(q, \bfh)^{T}$. As in the previous case, the columns of $\bfU$ and $\bfW$ are permuted independently. Therefore condition $(a)$ and $(b)$ hold for $\bfW$. We now verify condition $(c)$. First, permutation using $\rotmat(q, \bfh)^{T}$ implies that $u_{ij} = w_{lj}$ for all $l = i\cdot h$. Second, $\mcBB(\K,1)$ can be obtained by $\mcBB(\K,h)$ as follows:
\begin{equation}
\label{eq.rot_lemma3}
\mcBB(\K,1) = \{y | y = x\cdot h, x \in\mcBB(\K,h) \}.
\end{equation}
Thus, for fixed $j$, $\K$ and $l$, if entry $u_{ij}$ is contained in the vector $\bff_{\mcBB(\K,h_j),j}$, then entry $w_{lj}$ is contained in the vector $\bff_{\mcBB(\K,1),j}$ and $u_{ij} = w_{lj}$. Therefore the vector $\bfg^\K$ remains unchanged, which implies that condition $(c)$ is satisfied by $\bfW$.

\subsubsection{Proof of Lemma~\ref{lemma.rotation_projection}}
Using Lemma~\ref{lemma.rotation_equivalence}, we can follow exactly the same logic as in the proof of Lemma~\ref{lemma.rotation_equivalence}. The details are omitted.

\subsection{Proof of Theorem~\ref{theorem.equiv_LP}} 
\label{proof.equiv_LP}
We first prove several lemmas that will be useful for proving the theorem.
\begin{lemma}
\label{lemma.mapping_lemma_between_two_embeddings}
For any $\alpha \in \GF_q$, $\cwembed(\alpha) = (1- \|\embed(\alpha)\|_1, \embed(\alpha))$.
\end{lemma}
\begin{proof}
If $\alpha \neq 0 $, $\|\embed(\alpha)\|_1 = 1$ and $\cwembed(\alpha) = (0, \embed(\alpha))$. If $\alpha = 0 $, $\|\embed(\alpha)\|_1 = 0$ and $\cwembed(\alpha) = (1,0,0,\dots,0)$.
\end{proof}
\begin{lemma}
\label{lemma.equiv_convex_hull}
Let $\mcC$ be the SPC code defined by check $\bfh$. If $$\bffhat = (\bff_1|\dots|\bff_{\checkdegree})^T \in \conv(\embedvec(\mcC)),$$ then $$\bffbar = (a_1,\bff_1|\dots|a_{\checkdegree},\bff_{\checkdegree})^T \in \conv(\cwembedvec(\mcC)),$$ where $a_i = 1- \|\bff_i\|_1$ for all $i = 1,\dots,\checkdegree$. Conversely, if $$\bffbar = (a_{1},\bff_1|\dots|a_{\checkdegree},\bff_{\checkdegree})^T \in \conv(\cwembedvec(\mcC))$$ for some $a_1,\dots,a_{\checkdegree}$, then  $$\bffhat = (\bff_1|\dots|\bff_{\checkdegree})^T\in \conv(\embedvec(\mcC)).$$
\end{lemma}
\begin{proof}
Let $\bfc^k$, $k = 1,\dots,|\mcC|$ be all codewords in $\mcC$. Let $\bffhat\in  \conv(\embedvec(\mcC))$, then there exists a set of coefficients $0\leq\w_k \leq 1$, $k = 1,\dots,|\mcC|$, satisfying $\sum_{k = 1}^{|\mcC|}\w_k = 1 $ such that  $\bffhat = \sum_{k = 1}^{|\mcC|} \w_k \embedvec(\bfc^k)$. Then, $\bff_i = \sum_{k = 1}^{|\mcC|} \w_k \embed(c^k_i)$ for all $i$. On the other hand, $\sum_{k = 1}^{|\mcC|} \w_k \cwembedvec(\bfc^k) \in \conv(\cwembedvec(\mcC)) $ by the definition of convex hull.
Using Lemma~\ref{lemma.mapping_lemma_between_two_embeddings},
\begin{align*}
\sum_{k = 1}^{|\mcC|} \w_k \cwembedvec(\bfc^k) =& \left(\sum_{k = 1}^{|\mcC|} \w_k \cwembed(c^k_1),\dots,\sum_{k = 1}^{|\mcC|} \w_k \cwembed(c^k_{\checkdegree})\right) \\
= &\left(\sum_{k = 1}^{|\mcC|} \w_k  (1- \|\embed(c^k_1)\|_1, \embed(c^k_1)),\dots,\sum_{k = 1}^{|\mcC|} \w_k (1- \|\embed(c^k_\checkdegree)\|_1, \embed(c^k_\checkdegree))\right)\\
= &\left(1- \|\bff_1\|_1,\bff_1,\dots,1- \|\bff_\checkdegree\|_1,\bff_{\checkdegree}\right) = \bffbar.
\end{align*}
Thus we conclude $\bffbar \in \conv(\cwembedvec(\mcC))$.

Conversely, if $\bffbar\in  \conv(\cwembedvec(\mcC))$, then there exists a set of coefficients $0\leq\w_k \leq 1$, $k = 1,\dots,|\mcC|$, satisfying $\sum_{k = 1}^{|\mcC|}\w_k = 1 $ such that  $\bffbar = \sum_{k = 1}^{|\mcC|} \w_k \cwembedvec(\bfc^k)$. By definition, $ \sum_{k = 1}^{|\mcC|} \w_k \embedvec(\bfc^k) \in \conv(\embedvec(\mcC))$. Similar to the previous case, it is easy to verify using Lemma~\ref{lemma.mapping_lemma_between_two_embeddings} that $\sum_{k = 1}^{|\mcC|} \w_k \embedvec(\bfc^k) = \bffhat$.
\end{proof}
\begin{lemma}
\label{lemma.equiv_feasibility}
If $\bffhat = (\bff_1| \dots|\bff_{\blocklength})^T$ is feasible for \LPDFlanagan, then $\bffbar = (1- \|\bff_1\|_1,\bff_1| \dots|1-\|\bff_{\blocklength}\|,\bff_{\blocklength})^T$ is feasible for \LPDCW. Conversely, if $\bffbar = (f_{1},\bff_1| \dots|f_{\blocklength},\bff_{\blocklength})^T$ is feasible to \LPDCW, then  $\bffhat = (\bff_1| \dots|\bff_{\blocklength})^T$ is feasible fow \LPDFlanagan.
\end{lemma}
\begin{proof}
For LDPC codes, check $j \in \mcJ$ selects a sub-vector $\bffhat_j$ from $\bffhat$ (or $\bar{\bff_j}$ from $\bffbar$) that participates in that check. Therefore we apply Lemma~\ref{lemma.equiv_convex_hull} to each sub-vector of $\bffhat$ (and of $\bffbar$). In addition, $\bffhat$ (or $\bffbar$) is feasible if and only if $\bffhat_j \in \conv(\embedvec(\mcC_j)$ ($\bar{\bff_j} \in \conv(\cwembedvec(\mcC_j)$) for all $j \in \mcJ$, where $\mcC_j$ is the SPC code defined by the $j$-th check. Therefore we conclude that the lemma holds.
\end{proof}
\begin{lemma}
\label{lemma.equiv_obj}
Let $\bfy \in \Sigma$, $\bffhat = (\bff_1| \dots|\bff_{\blocklength})^T$ be a vector with non-negative entries and such that $\|\bff_i\|_1 \leq 1$. Further, let $\bffbar = (1- \|\bff_1\|_1,\bff_1| \dots|1-\|\bff_{\blocklength}\|,\bff_{\blocklength})^T$. Then $$\cwembedllrvec(\bfy) \bffbar = \embedllrvec(\bfy)\bffhat  - \sum_{i = 1}^{\blocklength} \log ( \Pr[y_i|0]).$$
\end{lemma}
\begin{proof}
For simplicity, we index entries in $\bffhat$ starting from $0$. By definitions of $\bffhat$ and $\bffbar$, we have 
\begin{align*}
\bff_{i} &= (\bar{f}_{(i-1)q +1}, \bar{f}_{(i-1)q +2},\dots,\bar{f}_{(i-1)q + q-1})\\
 &= (\hat{f}_{(i-1)(q-1) +1}, \hat{f}_{(i-1)(q-1) +2},\dots,\hat{f}_{(i-1)(q-1) + q-1}),
\end{align*}
and $$\bar{f}_{(i-1)q }= 1- \|\bff_{i}\|_1.$$ Note that $\bar{f}_{(i-1)q +\delta} = \hat{f}_{(i-1)(q-1) +\delta}$ for all $i\in \mcI$ and $\delta \in [q-1]$. Then  
\small{
\begin{align*}
\cwembedllrvec(\bfy)\bffbar &= \sum_{i = 1}^{\blocklength} \sum_{\delta = 0}^{q-1}\log \left( \frac{1}{\Pr[y_i|\delta]}\right)\bar{f}_{(i-1)q+\delta}\\
&= \sum_{i = 1}^{\blocklength}\left[ \sum_{\delta = 1}^{q-1}\log \left( \frac{1}{\Pr[y_i|\delta]}\right)\bar{f}_{(i-1)q+\delta} + (1-\|\bff_{i}\|_1) \log \left( \frac{1}{\Pr[y_i|0]}\right)\right]\\
&=\sum_{i = 1}^{\blocklength}\left[ \sum_{\delta = 1}^{q-1}\log \left( \frac{\Pr[y_i|0]}{\Pr[y_i|\delta]}\right)\bar{f}_{(i-1)q+\delta} +  \log \left( \frac{1}{\Pr[y_i|0]}\right)\right]\\
&= \sum_{i = 1}^{\blocklength}\left[ \sum_{\delta = 1}^{q-1}\log \left( \frac{\Pr[y_i|0]}{\Pr[y_i|\delta]}\right)\hat{f}_{(i-1)(q-1) +\delta}\right] + \sum_{i = 1}^{\blocklength} \log \left( \frac{1}{\Pr[y_i|0]}\right)\\
&= \embedllrvec(\bfy)\bffhat - \sum_{i = 1}^{\blocklength} \log ( \Pr[y_i|0]).
\end{align*}
}
\end{proof}

Now we state the proof of Theorem~\ref{theorem.equiv_LP}:
\begin{proof}
By Lemma~\ref{lemma.equiv_feasibility}, $\bffbar = (a_1,\bff_1| \dots|a_{\blocklength}\|,\bff_{\blocklength})^T$ is feasible for \LPDCW if and only if $\bffhat = (\bff_1| \dots|\bff_{\blocklength})^T$ is feasible for \LPDFlanagan. By Lemma~\ref{lemma.equiv_obj}, we can subtract a constant, $\sum_{i = 1}^{\blocklength} \log ( \Pr[y_i|0])$, from the objective function of \LPDCW to obtain the same objective function as \LPDFlanagan. 

We now prove the theorem by contradiction. Suppose $\bffbar$ is the \minimizer of \LPDCW, but $\bffbar$ is not the \minimizer of \LPDFlanagan. Then there exists a feasible point $\tilde{\bff}$ that attains a lower cost than $\bffhat$. Using Lemma~\ref{lemma.equiv_feasibility}, we can construct a feasible point for \LPDCW that attains a fewer cost than $\bffbar$. This contradicts with the assumption that $\bffbar$ is the \minimizer of \LPDCW. Similar arguments also hold for the converse statement.
\end{proof}

\subsection{Proof of Theorem~\ref{theorem.equiv_LP_relax}}
\label{proof.equiv_LP_relax}
We first restate Lemma~\ref{lemma.equiv_convex_hull} and Lemma~\ref{lemma.equiv_feasibility} for \nameTheCW:
\begin{lemma}
\label{lemma.equiv_relaxed_code_polytope}
Let $\relaxedcodepolytope$ (resp. $\cwrelaxedcodepolytope$) be the relaxed code polytope for \nameF (resp. \nameTheCW) for check $\bfh$. If $\bffhat = (\bff_1| \dots|\bff_{\checkdegree})^T \in \relaxedcodepolytope$, then $\bffbar = (a_1,\bff_1| \dots|a_{\checkdegree},\bff_{\checkdegree})^T \in \cwrelaxedcodepolytope$ where $a_i = 1- \|\bff_i\|_1$ for all $i = 1,\dots,\checkdegree$. Conversely, if $\bffbar = (a_{1},\bff_1| \dots|a_{\checkdegree},\bff_{\checkdegree})^T \in \cwrelaxedcodepolytope$ for some $a_1,\dots,a_{\checkdegree}$, then  $\bffhat = (\bff_1| \dots|\bff_{\checkdegree})^T\in \relaxedcodepolytope$.
\end{lemma}
\begin{proof}
For any value of $i \in [\checkdegree]$, $f_{0i}$ does not participate in the condition $(c)$ in Definition~\ref{def.cw_code_polytopes}. Thus $\bffhat$ satisfies condition $(c)$ in Definition~\ref{def.relaxed_code_polytope} if and only if $\bffbar$ satisfies condition $(c)$ in Definition~\ref{def.cw_code_polytopes}. Further, it is easy to verify that $\bffhat$ satisfies condition $(a)$ and $(b)$ in Definition~\ref{def.relaxed_code_polytope} if and only if $\bffbar$ satisfies condition $(a)$ and $(b)$ in Definition~\ref{def.cw_code_polytopes}.
\end{proof}
\begin{lemma}
\label{lemma.equiv_feasibility_relaxed_polytope}
If $\bffhat = (\bff_1| \dots|\bff_{\blocklength})^T$ is feasible for \LPDFlanaganRelax, then $\bffbar = (1- \|\bff_1\|_1,\bff_1| \dots|1-\|\bff_{\blocklength}\|,\bff_{\blocklength})^T$ is feasible for \LPDCWRelax. Conversely, if $\bffbar = (f_{1},\bff_1| \dots|f_{\blocklength},\bff_{\blocklength})^T$ is feasible to \LPDCWRelax, then  $\bffhat = (\bff_1| \dots|\bff_{\blocklength})^T$ is feasible to \LPDFlanaganRelax.
\end{lemma}
\begin{proof}
The proof is the same as for Lemma~\ref{lemma.equiv_feasibility} except it is based on Lemma~\ref{lemma.equiv_relaxed_code_polytope}.
\end{proof}

We can prove Theorem~\ref{theorem.equiv_LP_relax} using Lemma~\ref{lemma.equiv_obj} and Lemma~\ref{lemma.equiv_feasibility_relaxed_polytope}. The logic is the same as that for Theorem~\ref{theorem.equiv_LP}.

\subsection{Proof of Corollary~\ref{corollary.codewordindependence}}
\label{proof.codewordindependence}
\LPDFlanaganRelax has the same embedding and objective as \LPDFlanagan. Thus we can follow the same proof as in~\cite[Thm.~5.1]{flanagan2009linearprogramming}. The only difference is in the constraint set. Thus we only need to prove that the ``relative matrix'' (defined below) to the all-zeros matrix satisfies the constraint set in \LPDFlanaganRelax. The rest of the proof is identical to that in~\cite{flanagan2009linearprogramming}. The following definition rephrases Eq.~(19) in~\cite{flanagan2009linearprogramming}.
\begin{definition}
\label{def.relative}
Let $\mcC$ be a SPC code defined by a length-$d$ check $\bfh$. Let $\bfF$ and $\tilde{\bfF}$ be $(2^m-1) \times d$ matrices. We say that $\tilde{\bfF}$ is the 
``relative matrix'' for $\bfF$ based on a codeword $\bfc\in\mcC$ if
\begin{equation}
\label{eq.def_relative}
\tilde{f}_{ij} = 
\begin{cases}
1 - \sum_{i} f_{ij}\quad &\text{if }i + c_j = 0\\
f_{(i+c_j)j} \quad& \text{otherwise}
\end{cases},
\end{equation}
where the sums are in $\GF_{2^m}$. (Note that in $\GF_{2^m}$, $p + q = 0$ means $q = p$.) Denote this mapping $\relative_{\bfc}(\cdot)$. Note that this is a bijective mapping. Its inverse is given by
\begin{equation}
\label{eq.def_relativeinv}
f_{ij} = 
\begin{cases}
1 - \sum_{i} \tilde{f}_{ij}\quad &\text{if }i - c_j = 0\\
\tilde{f}_{(i-c_j)j} \quad& \text{otherwise}
\end{cases}.
\end{equation}
Denote by $\relativeinv_{\bfc}(\cdot)$ this inverse mapping.
\end{definition}

\begin{lemma}
\label{lemma.relative}
Let $\bfc$ be a valid SPC for length-$d$ check $\bfh$. Then $\bfF \in \relaxedcodepolytope$ if and only if $\relative_{\bfc}(\bfF) \in \relaxedcodepolytope$.
\end{lemma}
\begin{proof}
\begin{itemize}
\item We first show that if $\bfF \in \relaxedcodepolytope$, then $\relative_{\bfc}(\bfF) \in \relaxedcodepolytope$. In other words, we need to verify that $\relative_{\bfc}(\bfF)$ satisfies all three conditions in Definition~\ref{def.relaxed_code_polytope}. The first two conditions are obvious. We focus on the third condition. Let $\K$ be a non-empty subset of  $[m]$ and let $j$ be an integer in $[d]$. Let $\mcBB(\K,h_j)$ be the set defined in Lemma~\ref{lemma.equiv_integer_constraint} for the $j$-th check entry $h_j$. There are two cases: (i) $c_j \notin \mcBB(\K,h_j)$ and (ii) $c_j\in \mcBB(\K,h_j)$. 

If $c_j \notin \mcBB(\K,h_j)$ then for all $i\in \mcBB(\K,h_j)$,
\begin{align*}
\sum_{k\in \K} \embedbit(h_j(i + c_j))_k &= \sum_{k\in \K} \embedbit(h_ji + h_jc_j)_k \\
&= \sum_{k\in \K} \left[\embedbit(h_ji)_k + \embedbit(h_jc_j)_k\right]\\
&= 1 + 0 = 1,
\end{align*}
where the second equality follows because that the addition in $\GF_{2^m}$ is equivalent to the vector addition of the corresponding binary vectors. The third equality follows because $i\in \mcBB(\K,h_j)$ and $c_j \notin \mcBB(\K,h_j)$. Similarly, for all $i \notin \mcBB(\K,h_j)$, $\sum_{k\in \K} \embedbit(h_j(i + c_j))_k = 0$. Thus $i + c_j \in \mcBB(\K,h_j)$ if and only if $i\in \mcBB(\K,h_j)$. 

In addition, since $c_j \notin \mcBB(\K,h_j)$, we find that $c_j \neq i$ for all $i\in \mcBB(\K,h_j)$. Then, by Definition~\ref{def.relative}, $\tilde{f}_{ij} = f_{(i+c_j)j}$. 
Let $g^\K_j := \sum_{i \in \mcBB(\K, h_j)} f_{ij}$ and $\tilde{g}^\K_j := \sum_{i \in \mcBB(\K, h_j)} \tilde{f}_{ij}$, then 
\begin{align*}
\sum_{i \in \mcBB(\K, h_j)} \tilde{f}_{ij} &= \sum_{i \in \mcBB(\K, h_j)} f_{(i+c_j)j} \\
&= \sum_{(i+c_j) \in \mcBB(\K, h_j)} f_{(i+c_j)j} \\
&= \sum_{l \in \mcBB(\K, h_j)} f_{lj} = g^\K_j,
\end{align*}
where the second equality follows because $i + c_j \in \mcBB(\K,h_j)$ if and only if $i \in \mcBB(\K, h_j)$. Therefore we conclude that $g^\K_j = \tilde{g}^\K_j$.

For case (ii), $c_j \in \mcBB(\K,h_j)$. Using the same argument as above, we can show that $i + c_j \in \mcBB(\K,h_j)$ if and only if $i \notin \mcBB(\K,h_j)$. Therefore
\begin{align*}
\sum_{i \in \mcBB(\K, h_j)} \tilde{f}_{ij} &= \tilde{f}_{c_j j} + \sum_{i\in\mcBB(\K,h_j)\text{ and } i\neq c_j } \tilde{f}_{ij}\\
&= \left(1 - \sum_{i = 1}^{2^{m}-1} f_{ij}\right) + \sum_{l\notin\mcBB(\K,h_j)\text{ and }l\neq 0 } f_{lj} \\
&= 1-\sum_{l \in \mcBB(\K, h_j)} f_{lj} \\&= 1-g^\K_j.
\end{align*}
Combining the two cases, we conclude that the vector $\tilde{\bfg}^\K$ satisfies the following conditions:
\begin{equation}
\label{eq.gk_relative}
\tilde{g}^\K_j = \begin{cases}
g^\K_j \quad &\text{if }c_j \notin \mcBB(\K,h_j)\\
1 - g^\K_j \quad &\text{if }c_j \in \mcBB(\K,h_j)
\end{cases}.
\end{equation} 
We can rephrase this condition by introducing the following notation: Let $\tilde{\bfF} = \embedmat(\bfc)$. Let $\bfg^{\K,\bfc}$ be a binary vector for $\K$ and $\bfc$ defined by $g^{\K,\bfc}_j := \sum_{i \in \mcBB(\K, h_j)} \tilde{f}_{ij}$. By Lemma~\ref{lemma.equiv_integer_constraint}, $\bfg^{\K,\bfc}$ is a binary vector with even parity. By its definition, $g^{\K,\bfc}_j = 1$ if and only if $c_j \in \mcBB(\K, h_j)$. Thus we can rewrite~\eqref{eq.gk_relative} as
\begin{equation}
\label{eq.gk_relative_with_reference}
\tilde{g}^\K_j = \begin{cases}
g^\K_j \quad &\text{if }  g^{\K,\bfc}_j  = 0\\
1 - g^\K_j \quad &\text{if } g^{\K,\bfc}_j = 1
\end{cases}.
\end{equation} 
These conditions satisfy the definition of ``relative solution'' defined in~\cite{feldman2005using}. When applying Lemma~17 in~\cite{feldman2005using} to the case of binary single parity-check code, we conclude that $\tilde{\bfg}^\K \in \PP_d$ if $\bfg^\K \in \PP_d$. This conclude our verification of the third condition of Definition~\ref{def.relaxed_code_polytope}.
\item Next we need to show that if $\tilde{\bfF} \in \relaxedcodepolytope$, then $\relativeinv_{\bfc}(\tilde{\bfF}) \in \relaxedcodepolytope$. Note that in~\eqref{eq.def_relativeinv}, $i - c_j = 0$ is equivalent to $i + c_j = 0$ for $\GF_{2^m}$. Therefore the proof is identical to the previous case. 
\end{itemize}
\end{proof}

\subsection{Proof of Theorem~\ref{theorem.penalized_decoder_allzero}}
\label{proof.penalized_decoder_allzero}
\subsubsection{Sketch of the proof}
 We need to prove $\Pr[\error|\bfzero] = \Pr[\error|\bfc]$, where $\bfc$ is any non-zero codeword. Let
\begin{equation*}
\mcB(\bfc):=\{\bfy|\text{Decoder fails to recover }\bfc\text{ if }\bfy \text{ is received}\}.
\end{equation*}
Then $\Pr[\error|\bfc] = \sum_{y\in \mcB(c)}\Pr[\bfy|\bfc]$. 

We first rephrase the symmetry condition of~\cite{flanagan2009linearprogramming}. This definition introduces a one-to-one mapping from the received vector $\bfy$ to a vector $\bfy^0$ such that
the following two statements hold:
\begin{enumerate}
\item[$(a)$] $\Pr[\bfy|\bfc] = \Pr[\bfy^0|\bfzero]$,
\item[$(b)$] $\bfy \in \mcB(\bfc)$ if and only if $\bfy^0\in \mcB(\bfzero)$.
\end{enumerate}
Statement $(a)$ is directly implied by the definition of the symmetry condition. We show in the following that statement $(b)$ is also true. Once we have both results,
\begin{align*}
\Pr[\error|\bfc] &= \sum_{\bfy\in \mcB(\bfc)}\Pr[\bfy|\bfc]\\
&= \sum_{\bfy^0\in \mcB(\bfzero)}\Pr[\bfy^0|\bfc] = \Pr[\error|\bfzero].
\end{align*}

\subsubsection{Symmetry condition}
A symmetry condition for rings is defined in~\cite{flanagan2009linearprogramming}. We now rephrase that definition for fields.
\begin{definition}(Symmetry condition)
For $\beta \in \GF_{q}$, there exists a bijection 
\begin{equation*}
\tau_{\beta}: \Sigma \mapsto \Sigma,
\end{equation*} 
such that the channel output probability conditioned on the channel input satisfies
\begin{equation}
P(y|\alpha) = P(\tau_{\beta}(y)|\alpha - \beta),
\end{equation}
for all $y\in \Sigma$, $\alpha \in \GF_{q}$. In addition, $\tau_{\beta}$ is isometric in the sense of~\cite{flanagan2009linearprogramming}.
\end{definition}

\subsubsection{Proof of statement (b)}
We define the concept of relative matrices for \nameTheCW.
\begin{definition}
Let $\alpha\in\GF_{2^m}$ and let $\bfx$ be a vector of length $2^m$. We say that $\tilde{\bfx}$ is the 
``relative vector'' for $\bfx$ based on $\alpha$ if
\begin{equation}
\label{eq.def_relative_cw}
\tilde{x}_{i} = x_{i+\alpha},
\end{equation}
where the sum is in $\GF_{2^m}$. This mapping is denoted by $\tilde{\bfx} = \relativecw_{\alpha}(\bfx)$. Its inverse is given by
\begin{equation}
\label{eq.def_relativeinv_cw}
x_{i} = \tilde{x}_{i-\alpha}.
\end{equation}
Denote by $\relativecwinv_{\bfc}(\cdot)$ this inverse mapping.

We reuse this notation in the context of non-binary vectors and let $\bfc \in \GF_{2^m}^d$. Let $\bfF'$ and $\tilde{\bfF}'$ be $2^m \times d$ matrices. We say that $\tilde{\bfF}'$ is the 
``relative matrix'' of $\bfF'$ based on $\bfc$ if for all $j$
\begin{equation}
\tilde{f}'_{ij} = f'_{(i+c_j)j},
\end{equation}
where the sum is in $\GF_{2^m}$. This mapping is denoted by $\tilde{\bfF}' = \relativecw_{\bfc}(\bfF')$. Its inverse is given by
\begin{equation}
f'_{ij} = \tilde{f}'_{(i-c_j)j}.
\end{equation}
Denote by $\relativecwinv_{\bfc}(\cdot)$ this inverse mapping. Finally, we note that for any vector $\bflambda$ of length $2^md$, we can think of this vector as $\bflambda = \vect(\bfF)$. Then we let $\tilde{\bflambda} = \relativecw_{\bfc}(\bflambda)$ where $\tilde{\bflambda} := \vect(\tilde{\bfF})$ and $\tilde{\bfF} := \relativecw_{\bfc}(\bfF)$.
\end{definition}
\begin{lemma}
\label{lemma.linearity_and_distance}
The relative operation is linear. That is, $$\relativecw_{\alpha}(\phi_x \bfx + \phi_y \bfy) = \phi_x\relativecw_{\alpha}(\bfx) + \phi_y \relativecw_{\alpha}(\bfy).$$ Further, the relative operator is norm preserving. That is, $\| \relativecw_{\alpha}(\bfx)\| = \|\bfx\|$. 
\end{lemma}
\begin{proof}
Linearity is easy to verify. We note that $\relativecw_{\alpha}(\cdot)$ permutes the input vector based on $\alpha$, and therefore the norm is preserved.
\end{proof}
\begin{lemma}
\label{lemma.y_relative}
Let $y \in \Sigma$. Then 
$\cwembedllr(y) =  \relativecw_{\beta}(\cwembedllr(\tau_{\beta}(y)))$.
\end{lemma}
\begin{proof}
For all $j = 0,\dots, 2^m-1$, $P(y|j) = P(\tau_{\beta}(y)|j - \beta)$. Therefore
\begin{equation*}
\log \left(\frac{1}{P(y|j)}\right) = \log \left(\frac{1}{P(\tau_{\beta}(y)|j - \beta)}\right).
\end{equation*}
This means that 
\begin{equation}
\cwembedllr(y)_j =\cwembedllr(\tau_{\beta}(y))_{j - \beta},
\end{equation}
satisfying the definition of relative vector in~\eqref{eq.def_relativeinv_cw}. Therefore $\cwembedllr(\tau_{\beta}(y))) = \relativecw_{\beta}(\cwembedllr(y))$.
\end{proof}
\begin{lemma}
\label{lemma.relative_code_polytope}
Let $\bfc$ be a valid SPC codeword for length-$d$ check $\bfh$. Then $\bfF' \in \cwrelaxedcodepolytope$ if and only if $\relativecw_{\bfc}(\bfF') \in \cwrelaxedcodepolytope$.
\end{lemma}
\begin{proof}
We can use the same logic in the proof of Lemma~\ref{lemma.relative}. Thus we omit the details.
\end{proof}
\begin{lemma}
\label{lemma.projection_preserving}
Suppose a convex set $\mathbb{C}$ is such that $\bfx \in \mathbb{C}$ if and only if $  \relativecw_{\alpha}(\bfx)\in \mathbb{C}$ for some $\alpha$, then $\Proj_{\mathbb{C}}(\relativecw_{\alpha}(\bfv)) = \relativecw_{\alpha}(\Proj_{\mathbb{C}}(\bfv))$.
\end{lemma}
\begin{proof}
Our proof is by contradiction. Suppose that the projection of $\relativecw_{\alpha}(\bfv)$ onto $\mathbb{C}$ is $\bfu \neq \relativecw_{\alpha}(\Proj_{\mathbb{C}}(\bfv))$. Then $\relativecwinv_{\alpha}(\bfu) \in \mathbb{C}$ and $\relativecwinv_{\alpha}(\bfu) \neq \Proj_{\mathbb{C}}(\bfv)$. Due to Lemma~\ref{lemma.linearity_and_distance}, we have
\begin{align*}
\| \relativecwinv_{\alpha}(\bfu) - \bfv \|_2 & = \| \bfu - \relativecw_{\alpha}(\bfv) \|_2\\
&< \| \relativecw_{\alpha}(\Proj_{\mathbb{C}}(\bfv)) - \relativecw_{\alpha}(\bfv)  \|_2\\
&=  \| \Proj_{\mathbb{C}}(\bfv) - \bfv \|_2.
\end{align*}
This contradicts the fact that $ \Proj_{\mathbb{C}}(\bfv)$ is the projection of $\bfv$ onto $\mathbb{C}$.
\end{proof}

\begin{lemma}
\label{lemma.quiviter}
In Algorithm~\ref{algorithm.penalized_decoder}, let  $\bfx^{(\iterk)}$, $\bfz_j^{(\iterk)}$ and $\bflambda_j^{(\iterk)}$ be the vectors after the $\iterk$-th iteration when decoding $\bfy$. Let $\bfx^{0,(\iterk)}$, $\bfz_j^{0,(\iterk)}$ and $\bflambda_j^{0,(\iterk)}$ be the vectors after the $\iterk$-th iteration when decoding $\bfy^0$. If $\bfx^{0,(\iterk)} = \relativecw_{\bfc}(\bfx^{(\iterk)})$,  $\bfz_j^{0,(\iterk)} = \relativecw_{\bfc}(\bfz_j^{(\iterk)})$ and $\bflambda_j^{0,(\iterk)} = \relativecw_{\bfc}(\bflambda_j^{(\iterk)})$ then $\bfx^{0,(\iterk+1)} = \relativecw_{\bfc}(\bfx^{(\iterk+1)})$,  $\bfz_j^{0,(\iterk+1)} = \relativecw_{\bfc}(\bfz_j^{(\iterk+1)})$ and $\bflambda_j^{0,(\iterk+1)} = \relativecw_{\bfc}(\bflambda_j^{(\iterk+1)})$.
\end{lemma}
\begin{proof} 
We drop the iterate $(\iterk)$ for simplicity and denote by $\bfx^{\new}$ , $\bfz_j^{\new}$ and $\bflambda^{\new}_j$ the updated vectors at the $(\iterk+1)$-th iteration. 
Let $\bfx^{0} = \relativecw_{\bfc}(\bfx)$,  $\bfz_j^{0} = \relativecw_{\bfc}(\bfz_j)$ and $\bflambda_j^{0} = \relativecw_{\bfc}(\bflambda_j)$. Also let $\bfgamma^0$ be the log-likelihood ratio for the received vector $\bfy^0$.
From Algorithm~\ref{algorithm.penalized_decoder}, 
\small{
\begin{equation}
\label{eq.x_update_relative}
\bfx_i^{0,\new} \! = \! \Proj_{\PS_q'} \! \left[\frac{1}{d_i \! - \! \frac{2\alpha}{\mu}}\left(\sum_{j\in\Nev(i)} \left( \bfz_j^{0,(i)} \! - \! \frac{\bflambda_j^{0,(i)}}{\mu}\right) \! - \!  \frac{\bfgamma^0_i}{\mu} \! - \! \frac{2\alpha \bfr}{\mu}\right)\right].
\end{equation}
}
By Lemma~\ref{lemma.y_relative}, $\bfgamma_i^0 = \relativecw_{c_i}(\bfgamma_i)$. Then~\eqref{eq.x_update_relative} can be rewritten as 
\begin{align*}
\bfx_i^{0, \new} \! & = \! \Proj_{\PS_q'} \! \left[\frac{1}{d_i \! - \! \frac{2\alpha}{\mu}}\left(\sum_{j\in\Nev(i)} \left( \relativecw_{c_i}(\bfz_j^{(i)} )\! - \! \frac{\relativecw_{c_i}(\bflambda_j^{(i)})}{\mu}\right) \! - \!  \frac{\relativecw_{c_i}(\bfgamma'_i)}{\mu} \!  - \! \frac{2\alpha \bfr}{\mu}\right)\right]\\
&= \Proj_{\PS_q'}[\relativecw_{c_i} (\bfu)],
\end{align*}
where
\begin{align*}
\bfu = \frac{1}{d_i \! - \! \frac{2\alpha}{\mu}}\left(\sum_{j\in\Nev(i)} \left( \bfz_j^{(i)} \! - \! \frac{\bflambda_j^{(i)}}{\mu}\right) \! - \!  \frac{\bfgamma'_i}{\mu} \! - \! \frac{2\alpha \bfr}{\mu}\right).
\end{align*}
By Lemma~\ref{lemma.projection_preserving},
\begin{equation*}
\bfx_i^{0,\new} = \relativecw_{\bfc}(\bfx^{(\iterk)}).
\end{equation*}

Let $\bfv_j  = \bfP_j \bfx + \bflambda_j /\mu$ and $\bfv^0_j  = \bfP_j \bfx^0 + \bflambda^0_j /\mu$, then $\bfv^{0, (i)}_j =\relativecw_{c_i}( \bfv^{ (i)}_j)$. In addition, $ \bfz^{0, \new}_j = \Proj_{\cwrelaxedcodepolytope} (\bfv^0_j) $ and $ \bfz^{\new}_j = \Proj_{\cwrelaxedcodepolytope} (\bfv_j) $. By Lemma~\ref{lemma.projection_preserving}, $ \bfz^{0, \new}_j = \relativecw_{\bfc}(\bfz^{\new}_j)$.

It remains to verify one more equality:
\begin{align*}
\bflambda^{0,\new,(i)}_j \! &= \! \bflambda^{0,(i)}_j \! + \! \mu \left( (\bfP_j \bfx^{0,\new})^{(i)} \! - \! \bfz_j^{0,\new,(i)}\right) \\
&= \relativecw_{c_i} (\bflambda^{(i)}_j) \! + \! \mu \left( (\bfP_j \relativecw_{c_i}(\bfx^{\new})^{(i)}) \! - \!  \relativecw_{c_i}(\bfz_j^{\new,(i)})\right) \\
&= \relativecw_{c_i}(\bflambda^{\new,(i)}_j).
\end{align*}
\end{proof}

\begin{lemma}
\label{lemma.equaldecoding}
Let $\hat{\bfx} = \mathscr{D}(\bfy)$ be the output of the decoder if $\bfy$ is received and $\hat{\bfx}^0 = \mathscr{D}(\bfy^0)$. Then $\hat{\bfx}^0 = \relativecw_{\bfc}(\hat{\bfx})$.
\end{lemma}
\begin{proof}
We note that we initialize Algorithm~\ref{algorithm.penalized_decoder} so that the conditions in Lemma~\ref{lemma.quiviter} are satisfied. By induction, we always obtain relative vectors at each iteration. It is easy to verify that both decoding processes stop at the same iteration. Therefore $\hat{\bfx}^0 = \relativecw_{\bfc}(\hat{\bfx})$.
\end{proof}

Due to Lemma~\ref{lemma.equaldecoding}, if the decoder recovers a codeword $\bfc$ from $\bfy$, it means that the decoded vector $\bfx$ is the embedding of $\bfc$ in the sense of Definition~\ref{def.constant_weight_embedding}. Therefore by~\eqref{eq.def_relative_cw}, $\relativecw_{\bfc}(\hat{\bfx})$ is the embedding of $\bfzero$ in the sense of Definition~\ref{def.constant_weight_embedding}. This means that the decoder can recover $\bfzero$ from $\bfy^0$. One the other hand, if the decoder fails to recover codeword $\bfc$ from $\bfy$, then $\relativecw_{\bfc}(\hat{\bfx})$ is not a integral vector. This means that the decoder cannot recover $\bfzero$ from $\bfy^0$. Combining both arguments, we deduce that the decoder can recover $\bfc$ from $\bfy$ if and only if it can recover $\bfzero$ from $\bfy^0$, which completes the proof.

\section{Linear time $\bfx$-update algorithm}
\label{appendix.x_update_algorithm}
First note that $\rotmat^{-1} = \rotmat^T$, therefore 
$$\bfZ_{j,\kk}^T \bfZ_{j,\kk} = \bfP_j^T \rotmat_j  \bfT_\kk^T  \bfT_\kk \rotmat_j^T \bfP_j.$$
As a result,
\begin{align*}
\bfZ &= \sum_{j,\kk}\bfZ_{j,\kk}^T\bfZ_{j,\kk} \\
&= \sum_{j,\kk}  \bfP_j^T \rotmat_j  \bfT_\kk^T  \bfT_\kk \rotmat_j^T \bfP_j\\
& = \sum_{j}  \bfP_j^T \rotmat_j  \left( \sum_k \bfT_\kk^T  \bfT_\kk \right) \rotmat_j^T \bfP_j.
\end{align*}
Let $\bfT = \sum_\kk \bfT_\kk^T  \bfT_\kk $. Then
\begin{align*}
\bfZ = \sum_{j}  \bfP_j^T \rotmat_j  \bfT \rotmat_j^T \bfP_j.
\end{align*}
We first prove the following Lemma that is useful in studying the structure of $\bfT$.
\begin{lemma}
\label{lemma.binarysum}
For an integer $n\geq 2$, let $\mcKK_n = \{ \mcK | \mcK \subset [n]\text{ and } \mcK \neq \emptyset\}$. Let $\bfu, \bfv$ be arbitrary binary vectors of length $n$ that are not equal to $\bfzero$.  Then
\begin{enumerate}
\item The number of sets $\mcK\in \mcKK_n$ such that $\sum_{i\in \mcK} v_i = 1$ in $\GF_2$ is $2^{n-1}$.
\item If $\bfv \neq \bfu$, the number of sets $\mcK\in \mcKK_n$ such that $\sum_{i\in \mcK} u_i = 1$ and $\sum_{i\in \mcK} v_i = 1$ in $\GF_2$ is $2^{n-2}$.
\end{enumerate}
\end{lemma}
\begin{proof}
We prove both parts of the lemma using inductions.
\begin{enumerate}
\item Proof by induction: When $n = 2$, the statement holds. Assume that the statement holds for $n = l$. Denote by $\mcKK_{l,0}(\bfv)$ (resp. $\mcKK_{l,1}(\bfv)$) the set of $\K \in \mcKK_l$ such that $\sum_{i\in \K} v_i = 0$ (resp. $\sum_{i\in \K} v_i = 1$). This implies $$|\mcKK_{l,0}(\bfv)| = |\mcKK_{l,0}(\bfv)| = 2^{l-1}.$$ When $n= l+1$, we need to consider binary vectors of length $l+1$. We use $\bfv$ to denote the first $l$ bits of the vector. 
Let $x \in \{0,1\}$ be the $l+1$-th bit of the vector. 
If $x = 0$, then $\mcKK_{l+1,0}((\bfv, 0))$ contains the following sets: (a) all sets $\K \in \mcKK_{l,0}(\bfv)$ and (b) all sets $ \K\cup \{l+1\}$ where $\K \in \mcKK_{l,0}(\bfv)$. Therefore $|\mcKK_{l+1,0}((\bfv, 0))| = 2^{l}$. Since $\mcKK_{l+1,1}((\bfv, 0))$ is the complement of $\mcK_{l+1,0}((\bfv, 0))$, $|\mcKK_{l+1,1}((\bfv, 0))| = 2^{l}$. 
If $x = 1$, then $\mcKK_{l+1,0}((\bfv, 1))$ contains the following sets: 
(a) all sets $\K \in \mcKK_{l,0}(\bfv)$ and (b) all sets $ \K\cup \{l+1\}$ where $\K \in \mcKK_{l,1}(\bfv)$. Therefore, we get the same result, i.e. $$|\mcKK_{l+1,0}((\bfv, 1))| = |\mcKK_{l+1,1}((\bfv, 1))| = 2^{l}.$$ Combining the two cases, we conclude that $$|\mcKK_{l+1,0}((\bfv, 1))|  = 2^{l},$$ which completes the proof.

\item Proof by induction. When $n = 2$, the statement holds. Assume that the statement holds for $n= l$. Denote by $\mcKK_{l,(0,0)}(\bfu,\bfv)$ (resp. $\mcKK_{l,(0,1)}(\bfu,\bfv)$, $\mcKK_{l,(1,0)}(\bfu,\bfv)$, $\mcKK_{l,(1,1)}(\bfu,\bfv)$) the set of $\K \in \mcKK_l$ such that $\sum_{i\in \K} u_i = 0$ and $\sum_{i\in \K} v_i = 0$ (the resp. sums equal to $(0,1)$, $(1,0)$ and $(1,1)$). Note all these sets have $2^{l-2}$ elements. Similar to the previous proof, we add one more bit to both $\bfu$ and $\bfv$. Depending on the combination of these two bits ($(0,0)$, $(0,1)$, $(1,0)$ or $(1,1)$), the four sets $\mcKK_{l+1,(0,0)}(\bfu,\bfv)$, $\mcKK_{l+1,(0,1)}(\bfu,\bfv)$, $\mcKK_{l+1,(1,0)}(\bfu,\bfv)$ and $\mcKK_{l+1,(1,1)}(\bfu,\bfv)$ all double in size. Thus the four new sets all have $2^{l-1}$ elements. 
\end{enumerate}
\end{proof}
We now describe the structure of $\bfT$.
\begin{lemma}
\label{lemma.rmatrix}
$\bfT$ is block diagonal matrix, denoted by $$\diag(\bfPhi, \bfPhi,\dots,\bfPhi),$$ where 
\begin{align}
\label{eq.phi_matrix}
\bfPhi = \begin{pmatrix}
2^{m-1} & 2^{m-2} &   2^{m-2} & \dots &   2^{m-2}\\
2^{m-2}  &  2^{m-1} &  2^{m-2} & \dots &   2^{m-2} \\
\dots & \dots & \dots & \dots & \dots\\
2^{m-2}  &  2^{m-2}  &  2^{m-2}  & \dots & 2^{m-1}
\end{pmatrix}.
\end{align}
\end{lemma}
\begin{proof}
The matrices $\bfT_\kk$ are in a one-to-one relation with the set $\mcBB(\mcK,1)$. Since $\bfT = \sum_{\kk = 1}^{2^m-1} \bfT_\kk^T \bfT_\kk$, $t_{ii} = \sum_\mcK \mathbb{I}(i \in \mcBB(\mcK,1))$, where $\mathbb{I}$ is the indicator function. In other words, $t_{ii}$ is the number of sets $\mcK$ such that $i\in \mcBB(\mcK,1)$. Similarly, for all $i\neq j$, $t_{ij} = \sum_\mcK \mathbb{I}(i,j \in \mcBB(\mcK,1))$. By Lemma~\ref{lemma.binarysum} we conclude that $t_{ii} = 2^{m-1}$ and $t_{ij} = 2^{m-2}$ for all $1\leq i,j \leq 2^{m}-1$. 
\end{proof}

Note that $\rotmat_j$ is also a block diagonal matrix. Therefore 
\begin{equation*}
\rotmat_j  \bfT \rotmat_j^T = \diag(\rotmat_j(h_1)\bfPhi \rotmat_j(h_1)^T,\rotmat_j(h_2) \bfPhi \rotmat_j(h_2)^T,\dots,\rotmat_j(h_d)\bfPhi \rotmat_j(h_d)^T).
\end{equation*}
\begin{lemma}
\label{lemma.useless_D}
Let $\bfD$ be an $n\times n$ permutation matrix. Let $\bfT$ be an $n\times n$ matrix whose entries are $t_{ii} = a$ for all $i\in[n]$ and some constant $a$, $t_{ij} = b$ for all $i\in[n]$, $j\in[n]$ and $i\neq j$, and some constant $b$. Then, $\bfD^T  \bfT \bfD = \bfT$.
\end{lemma}
\begin{proof}
Without loss of generality, let $t_{11} = a$ and $t_{12} = b$. Let $\bfX = \bfD^T  \bfT \bfD$. We need to prove that $x_{ii} = a$ for all $i$ and that $x_{ij} = b$ for all $i\neq j$.

Let $\bfY = \bfD^T  \bfT$, then $\bfX = \bfY\bfD$. Therefore $x_{ij} = \sum_{k = 1}^n y_{ik} d_{kj}$. Since $\bfD$ is a permutation matrix, there is one $1$ in the $j$-th column. Without loss of generality assume $d_{\alpha j} = 1$. Then $x_{ij} = y_{i \alpha}$. Further, $y_{i\alpha} = \sum_{k = 1}^n d_{ki} t_{k \alpha}$. If $i \neq j$, then there exists a $\beta \neq \alpha$ such that $\sum_{k = 1}^n d_{ki} t_{k \alpha} = t_{\beta \alpha} = t_{12} $. Thus $x_{ij} = t_{12} = b$. If $i = j$, note that $d_{\alpha i} = 1$ and $d_{k i} = 0$ for all $k\neq i$. Therefore $y_{i\alpha} = \sum_{k = 1}^n d_{ki} t_{k \alpha} = t_{\alpha \alpha} = t_{11} = a$. 
\end{proof}

By Lemma~\ref{lemma.useless_D} we conclude that $\rotmat_j  \bfT \rotmat_j^T = \bfT$ for all $\rotmat_j$. As a result, 
\begin{align*}
\bfZ &= \sum_{j}  \bfP_j^T \rotmat_j  \bfT \rotmat_j^T \bfP_j\\
& = \sum_{j}  \bfP_j^T \bfT \bfP_j \\
&= \diag(\variabledegree \bfT, \variabledegree \bfT,\dots,\variabledegree \bfT),
\end{align*}
where $\variabledegree$ is the variable degree, i.e, $\variabledegree = |\Nev(i)|$. Then,
\begin{equation*}
\bfZ + \bfI = \diag(\variabledegree \bfT + \bfI, \variabledegree \bfT + \bfI,\dots,\variabledegree \bfT + \bfI).
\end{equation*}
Since $\variabledegree \bfT + \bfI = (\variabledegree \bfT + \bfI)^T$ and its entries only have two values, we can calculate its inverse explicitly. We do this next in Lemma~\ref{lemma.invertZplusI}.
\begin{lemma}
\label{lemma.invertZplusI}
Let $\bfT$ be a $(2^m-1)d\times (2^m-1)d$ block diagonal matrix denoted by $\diag(\bfPhi, \bfPhi,\dots,\bfPhi)$, where each $\bfPhi$ is the same $(2^m-1)\times (2^m-1)$ matrix in~\eqref{eq.phi_matrix}. Then $(\bfZ + \bfI)^{-1} = \diag((\variabledegree \bfT + \bfI)^{-1}, (\variabledegree \bfT + \bfI)^{-1},\dots,(\variabledegree\bfT + \bfI)^{-1})$ and 
\begin{align*}
(\variabledegree \bfT + \bfI)^{-1} = \begin{pmatrix}
a & b & \dots & b \\
b & a & \dots & b \\
\dots & \dots & \dots & \dots \\
b & b & \dots & a
\end{pmatrix},
\end{align*} 
where $ a =\frac{1}{r - s} + b$, $b = \frac{-r}{[r+s(2^m-2)](r-s)}$, $r = 2^m\variabledegree/2+1$ and $s = 2^m\variabledegree/4$.
\end{lemma} 

\begin{proof}
It is easy to verify that the product of the two matrices is the identity matrix. 
\end{proof}

To summarize, $(\bfZ+\bfI)^{-1} \bft$ can be obtained by calculating each block multiplication due to the fact that $(\bfZ+\bfI)^{-1}$ is a block diagonal matrix. For each block we first calculate $b \|\bft_i\|_1$, where $\bft_i$ is the $i$-th block of vector $\bft$. We then calculate $(a -b)\bft_i + b \|\bft_i\|_1$. This algorithm for $x$-update has complexity $O(q^2\blocklength)$.

\section{Results and a conjecture for $\GF_{2^2}$}
\label{appendix.conjecture}
\begin{lemma}
\label{lemma.gf4_valid_embeddings}
In $\GF_{2^2}$, $\bfF$ is a valid embedded matrix for the all-ones check if and only if $\bfF$ satisfies the first two conditions of Definition~\ref{def.validembedding} and the rows of $\bfF$ have either all odd parity \textbf{or} all even parity.
\end{lemma}
\begin{proof}
Let $\mcE^1$ be the set of binary matrices satisfying conditions in this lemma. First, we show $\bfF \in \mcE^1$ implies $\bfF \in \mcE$, where $\mcE$ is defined in Definition~\ref{def.validembedding}. Since the field under consideration is $\GF_{2^2}$, $\bfF$ has three rows. Denote by $\bff^R_i$ the $i$-th row of $\bfF$. Then $\bfg^1 = \bff^R_1 + \bff^R_3$ and $\bfg^2 = \bff^R_2 + \bff^R_3$. If $\bfF\in\mcE^1$,
\begin{align*}
\sum_{j = 1}^\checkdegree g^1_j &=  \sum_{j = 1}^\checkdegree (f_{1j} + f_{3j})\\
& =  \sum_{j = 1}^\checkdegree f_{1j} +  \sum_{j = 1}^\checkdegree  f_{3j}. 
\end{align*} 
If $ \sum_{j = 1}^\checkdegree f_{1j}=1$, then $\sum_{j = 1}^\checkdegree  f_{3j}= 1$  by the definition of $\mcE^1$, which means that the overall sum is $0$ in $\GF_{2}$. If both the sums have even parity, then the overall sum is also $0$. It is easy to verify that the same situation holds for $\bfg^2$. Thus $\bfF \in \mcE$.

Now we show that $\bfF \in \mcE$ implies $\bfF \in \mcE^1$. Let $\bfF\in\mcE$. We need to show that if $\bff_1$ has odd (even) parities, $\bff_2$ and $\bff_3$ must also have odd (even) parities. This can be proved by contradiction. Suppose $\bff_1$ has odd parity and $\bff_2$ (or $\bff_3$) has even parity, then $\bfg^1 = \bff_1 + \bff_3$ (or $\bfg^2 = \bff_2 + \bff_3$) does not have even parity, contradicting the assumption that $\bfF\in \mcE$. Similarly, if $\bff_1$ has even parity, the other two rows must both have even parity.
\end{proof}

In $\GF_{2^2}$, considering the all-ones check, the characteristics of the polytope in Definition~\ref{def.relaxed_code_polytope} can be simplified to the following:
\begin{definition}
\label{def.gf4_polytope}
Denote by $\relaxedcodepolytope_4$ the relaxed code polytope for $\GF_{2^2}$ for the all-ones checks. A $3 \times \checkdegree$ matrix $\bfF \in \relaxedcodepolytope_4$ if the following constraints hold:
\begin{enumerate}
\item $f_{ij} \in [0,1]$.
\item $\sum_{i = 1}^3 f_{ij} \leq 1$. 
\item Let $\bff_i$ be the $i$-th row of $\bfF$. Let $\bfg^1 = \bff_1 + \bff_3$,  $\bfg^2 = \bff_2 + \bff_3$ and $\bfg^3 = \bff_1 + \bff_2$. Then $\bfg^1 \in \PP_{\checkdegree}$, $\bfg^2 \in \PP_{\checkdegree}$ and $\bfg^3 \in \PP_{\checkdegree}$.
\end{enumerate}
\end{definition}

\begin{conjecture}
Let $\mcE$ be defined by Lemma~\ref{lemma.equiv_integer_constraint} for $\GF_{2^2}$ and by the all-ones check of length $d$. Let $\relaxedcodepolytope_4$ be defined by Definition~\ref{def.gf4_polytope}. Then $\relaxedcodepolytope_4 = \conv(\mcE)$.
\end{conjecture}

We validated this result numerically using the following approach. We randomly generated $10^6$ vectors of length $d$ where the entries of each vector are i.i.d. uniformly distributed in $[0,4]$. We then projected these vectors onto both $\relaxedcodepolytope_4$ and $\conv(\mcE)$ using CVX (cf.~\cite{cvx}). What we observed is that the projections onto these two polytopes are the same for every vector. Therefore we believe that the conjecture should hold. However, a proof of the conjecture remains open. 

\section{Calculating $E_s/N_0$}
\label{appendix.esn0}
In QPSK, the coded symbols are $0,1,\xi^1,\xi^2$. We use vectors $(1,0)$, $(0,1)$, $(-1, 0)$ and $(0,-1)$ to represent the modulated signals. Let $\bfx$ be the modulated vector of length $2 \blocklength$, where $\blocklength$ is the block length. Then the received signal is $\bfy = \bfx + \bfn$. Where $\bfn$ is i.i.d. Gaussian of variance $\sigma^2$.

Let $N_0$ be the noise spectrum density. The noise is equivalent to have variance of $\sigma^2 = N_0/2$ per dimension. Let $R$ be the symbol rate and fix $E_s = 1/R$. Let $\gamma = 10^{\frac{E_s/N_0}{10}}$ be the (non-dB) value for $E_s/N_0$, then
\begin{equation*}
\sigma^2 = \frac{1}{2\gamma R}. 
\end{equation*}
Suppose we use the length-$80$ code\footnote{The parity-check matrix of this code is denoted by $\mathcal{H}^{(2)}$ in \cite{flanagan2009linearprogramming}.} defined in~\cite{flanagan2009linearprogramming} which has rate $R = 0.6$. Then $E_s/N_0 = 4$dB translates into $\sigma = \sqrt{\frac{1}{2\times 10^{\frac{4}{10}} \times 0.6}} = 0.5760$.

\end{document}